\documentclass[11pt]{article}
\usepackage[margin=1in]{geometry}
\usepackage[utf8]{inputenc}
\usepackage{outlines}
\usepackage{amsmath}
\usepackage{amsthm}
\usepackage{amssymb}
\usepackage{amsfonts}
\usepackage{comment}
\usepackage[table, dvipsnames]{xcolor}
\usepackage{graphicx}
\usepackage{capt-of}
\usepackage{mathtools}
\usepackage[font=small]{caption}
\usepackage[labelformat=simple]{subcaption}

\usepackage{mathdots}
\usepackage{algorithm}
\usepackage{algpseudocode}
\usepackage{enumitem}
\usepackage[normalem]{ulem}
\usepackage{physics}
\usepackage{slashed}
\usepackage{braket}
\usepackage{nccmath}
\usepackage{float}
\usepackage{dcolumn}
\usepackage{tikz}
\usepackage{complexity}
\usetikzlibrary{quantikz2}
\usepackage{bm}
\usepackage{array}
\usepackage{color}
\usepackage{wrapfig}
\usepackage{microtype}
\usepackage{empheq}

\usepackage{hyperref}
\hypersetup{
    colorlinks=true,
    linkcolor=blue,
    filecolor=magenta,      
    urlcolor=cyan,
    citecolor=cyan,
    pdftitle={Quantum lower bound for simulating fluid dynamics},
    pdfpagemode=FullScreen,
}
\usepackage{todonotes}
\usepackage{tablefootnote}
\usepackage{thmtools}
\usepackage{thm-restate}
\usepackage[noabbrev,capitalise]{cleveref}
\usepackage{tabularx}

\theoremstyle{plain}
\newtheorem{problem}{Problem}
\newtheorem{thm}{Theorem}

\newtheorem{corr}[thm]{Corollary}

\newtheorem{lem}[thm]{Lemma} 
\Crefname{lem}{Lemma}{Lemmas}
\crefname{lem}{Lemma}{Lemmas}
\Crefname{fact}{Fact}{Facts}
\crefname{fact}{Fact}{Facts}

\theoremstyle{definition}
\newtheorem{defn}[thm]{Definition}

\newtheorem{remark}{Remark}

\let\originalleft\left
\let\originalright\right
\renewcommand{\left}{\mathopen{}\mathclose\bgroup\originalleft}
\renewcommand{\right}{\aftergroup\egroup\originalright}

\usepackage{authblk}

\title{Quantum lower bounds for simulating fluid dynamics}
\author[1,*]{Abtin Ameri\rlap{,}}
\author[2,**]{Joseph Carolan\rlap{,}}
\author[2,\dag]{Andrew M.~Childs\rlap{,}}
\author[3,\ddag]{Hari Krovi}
\affil[1]{Plasma Science and Fusion Center, Massachusetts Institute of Technology, Cambridge, MA\vspace{2pt}}
\affil[2]{Department of Computer Science, Institute for Advanced Computer Studies, and Joint Center for Quantum Information and Computer Science, University of Maryland, College Park, MD\vspace{2pt}}
\affil[3]{IBM Research, Cambridge, MA\vspace{2pt}}
\affil[ ]{
$^{*}$aameri@mit.edu
 \quad
$^{**}$jcarolan@umd.edu  \quad
$^{\dag}$amchilds@umd.edu \quad
$^{\ddag}$hari.krovi@ibm.com
}
\date{\today}

\begin{document}

\maketitle

\begin{abstract}
    Developing quantum algorithms to simulate fluid dynamics has become an active area of research, as accelerating fluid simulations could have significant impact in both industry and fundamental science. While many approaches have been proposed for simulating fluid dynamics on quantum computers, it is largely unclear whether these algorithms will provide speedup over existing classical approaches. In this paper we give evidence that quantum computers cannot significantly outperform classical simulations of fluid dynamics in general. We study two models of fluids: the Korteweg-de Vries (KdV) equation, which models shallow water waves, and the incompressible Euler equations, which model ideal, inviscid fluids. We show that any quantum algorithm simulating the KdV equation or the Euler equations for time $T$ requires $\Omega(T^2)$ and $e^{\Omega(T)}$ copies of the initial state in the worst case, respectively. These lower bounds hold for the task of preparing the final state, and similar bounds hold for history state preparation. We prove the lower bound for the KdV equation by investigating divergence of solitons. For the Euler equations, we show that instabilities enable fast state discrimination.
\end{abstract}
\vspace{20pt}
\begin{center}
\noindent\rule{14cm}{0.5pt}\\
\vspace{5pt}
    \parbox{14cm}{\textit{This paper is dedicated to the memory of Nuno Filipe Gomes Loureiro (1977--2025). An exceptional plasma physicist, colleague, and mentor, Nuno was very enthusiastic about understanding the impact of quantum computing on simulating differential equations. His passion for bringing together different areas of research made this collaboration possible.}}\\ 
    \vspace{5pt}
\noindent\rule{14cm}{0.5pt}
\end{center}

\newpage

\tableofcontents
\newpage

\section{Introduction}

Computational fluid dynamics (CFD) is an active area of research with applications in many fields of science and engineering, including aerospace, astrophysics, atmospheric science, automotive design, and medicine, among others. While CFD has diverse applications, implementing it can be challenging, especially for two- and three-dimensional flows, complicated geometries, and high spatial resolution. A plethora of different approaches can be used to solve fluid equations, from grid-based methods such as finite difference, finite element, finite volume, and the discontinuous Galerkin method~\cite{toro2013riemann,hesthaven2008nodal,durran2010numerical,ferziger2019computational}, to kinetic approaches such as lattice Boltzmann methods~\cite{kruger2017lattice} and gridless methods such as smoothed-particle hydrodynamics~\cite{monaghan1992smoothed}. Some of these numerical methods are heavily employed in other fields such as plasma dynamics, where fluid models of plasmas (e.g., magnetohydrodynamics) are commonly solved using CFD techniques~\cite{jardin2010computational}.

In almost all cases of interest, supercomputers are required to run fluid simulations in a reasonable time. Despite access to high-performance computing, state-of-the-art simulations remain time- and energy-intensive. Thus, accelerating them could be a significant breakthrough with major economic impact.

Quantum computers have been considered as a possible way to speed up fluid simulations, considering their ability to accelerate other problems like prime factorization~\cite{shor1999polynomial} and quantum simulation~\cite{lloyd1996universal}. Indeed, designing quantum algorithms to simulate differential equations has been an active area of research over the past two decades. Such algorithms can achieve an exponential reduction in space, and potentially in running time, by storing the solution as a quantum state instead of as an explicit vector.
The first efficient quantum algorithm for linear ordinary differential equations (ODEs) was proposed by Berry~\cite{berry2014high}. Subsequent work reduced the complexity of the algorithm and relaxed certain constraints~\cite{berry2017quantum,krovi2023improved}. These algorithms rely on discretizing time and recasting the discretized differential equation as a system of linear equations, for which quantum linear system algorithms (QLSAs) can be used~\cite{harrow2009quantum,childs2017quantum,costa2022optimal}. The output of the QLSA is a state encoding the history of the dynamics,
from which the normalized solution at the final time can be obtained by post-selection.

Several other approaches to quantum algorithms for linear ODEs were later proposed that did not rely on QLSAs. The first is the time-marching algorithm~\cite{fang2023time}, which adapts the classical approach of updating the solution by marching forward in time for quantum computers. Another approach is based on linear combination of Hamiltonian simulations~\cite{an2023linear,an2023quantum,low2025optimal}, which recasts the solution of a linear inhomogeneous ODE as an integral of unitary operations, which is then discretized and computed through linear combination of unitaries. Other proposed approaches include quantum eigenvalue processing~\cite{low2024quantum} and using Lindbladians to simulate dissipative differential equations~\cite{shang2025designing}.

While quantum algorithms for linear ODEs have significantly matured, developing efficient quantum algorithms for nonlinear ODEs is much less explored. There are several reasons for this. First, the unitarity of quantum computation makes performing nonlinear operations expensive. In fact, before Berry's pioneering work on linear ODEs~\cite{berry2014high}, Leyton and Osborne~\cite{leyton2008quantum} showed that a simple time-marching quantum algorithm for general nonlinear dynamics (i.e., a quantization of existing classical algorithms) up to time $T$ uses $e^{O(T)}$ copies of the initial state. The main obstacle in their approach is the no-cloning theorem, which prevents one from copying arbitrary quantum states. Second, it has been shown that if quantum mechanics were nonlinear, one would be able to solve $\mathsf{NP}$- and $\#\mathsf{P}$-hard problems efficiently~\cite{abrams1998nonlinear}, which is considered to be unlikely. Furthermore, tasks like unstructured search and state discrimination can be accelerated under nonlinear quantum mechanics~\cite{abrams1998nonlinear,aaronson2005np,childs2016optimal}. This implies that one should not be able to develop an efficient quantum algorithm for simulating \textit{general} nonlinear dynamical systems. 

However, this leaves open the possibility of efficient quantum algorithms for special classes of nonlinear differential equations.
Such an algorithm was given by Liu et al.\ for dissipative ODEs with quadratic nonlinearity~\cite{liu2021efficient}:
\begin{equation} \label{eq:quadratic-nonlin}
    \frac{d\mathbf u}{dt} = F_1 \mathbf u + F_2 \mathbf u^{\otimes 2} + F_0(t), \ \mathbf{u}(0) = \mathbf{u}_0.
\end{equation}
The strength of the nonlinearity is characterized by the parameter
\begin{equation} \label{eq:R}
    \mathrm{R} \coloneq \frac{1}{|\mathrm{Re}(\lambda_1)|}\left (\|\mathbf u_0\|\|F_2\| + \frac{\|F_0\|}{\|\mathbf{u}_0\|} \right ),
\end{equation}
where $\lambda_1$ is the least dissipative eigenvalue of $F_1$. Provided $\mathrm{R}<1$---meaning that the nonlinearity is sufficiently weak compared to the dissipation and driving---and $\|F_0\| \le \|F_2\|$, the authors gave an algorithm running in time polynomial in $T$ and $\log N$, where $N$ is the dimension of $\mathbf{u}$, that produces a quantum state proportional to $\mathbf{u}(T)$. Subsequent papers generalized the quantum algorithm and improved its complexity~\cite{krovi2023improved,costa2025further,wu2025quantum,jennings2025quantum}. On the other hand, \cite{liu2021efficient} shows that for strong nonlinearities (specifically, for $\mathrm{R}\geq \sqrt{2}$), any quantum algorithm simulating Eq.~\eqref{eq:quadratic-nonlin} requires $e^{\Omega (T)}$ copies of the initial state. This lower bound was later improved to only require $\mathrm{R} \geq 1$~\cite{lewis2024limitations}.

The aforementioned lower bound for nonlinear ODEs is worst-case. Thus, it is unclear if it applies to specific nonlinear dynamical systems of interest, such as the equations of fluid dynamics. Indeed, following the results of~\cite{liu2021efficient}, there have been many proposed approaches for quantum simulation of fluid dynamics with high Reynolds numbers, with the Reynolds number being analogous to $\mathrm R$ in Eq.~\eqref{eq:R}. For instance, the lattice Boltzmann method gained popularity as it was conjectured to be able to simulate subsonic flows with arbitrarily high Reynolds numbers~\cite{itani2022analysis,itani2024quantum,li2025potential}. A more thorough analysis of lattice Boltzmann algorithms claims that polynomial quantum advantage in terms of the Reynolds number is achievable in an end-to-end setting~\cite{jennings2025end}. 

While the proposed algorithms above seem promising, Lewis et al.\ argue that turbulent flow (i.e., flow at high Reynolds number) cannot be simulated efficiently since such dynamics have chaotic behavior~\cite{lewis2024limitations}. As chaotic dynamics has at least one positive Lyapunov exponent~\cite{strogatz2024nonlinear}, two nearby trajectories can be pulled apart quickly in the linearized regime of the dynamics, allowing for accelerated state discrimination. However, there are a few limitations to this argument. First, its analysis considers only the linearized regime. This approximation is not accurate for all times, and there is no guarantee that the solution remains close to the linearized dynamics as the states significantly separate. In fact, bounding the error from the linearization is a highly nontrivial task. Second, the argument considers general chaotic dynamical systems, and the authors argue qualitatively that such features should arise in fluid dynamical simulations with high Reynolds number flows. However, they do not provide a rigorous proof of this. Such a proof might be very difficult, for two reasons. First, it is not known how to analytically identify the chaotic regime in phase space (i.e., the location of the chaotic attractors) for fluid flows without heuristics. Second, once the region is identified, the Lyapunov exponent also cannot be calculated analytically.

\subsection{Technical contribution}
Below we outline the main contributions of this work.
\subsubsection{Advances over prior work}
We improve upon the results of~\cite{lewis2024limitations} in several ways. First, we show lower bounds for equations specific to fluid dynamics, namely the Korteweg-de Vries (KdV) equation
\begin{equation}
    \frac{\partial \phi}{\partial t} + \frac{\partial^3 \phi}{\partial x^3} - 6\phi \frac{\partial \phi}{\partial x} = 0,
\end{equation}
that models shallow water waves, and the incompressible Euler equations
\begin{equation}
    \begin{aligned}
        \frac{\partial \mathbf u}{\partial t} + (\mathbf u \cdot \nabla)\mathbf u &= -\frac{1}{\rho}\nabla P \\
        \nabla \cdot \mathbf u &= 0,
    \end{aligned}
\end{equation}
that model inviscid fluids. Second, our proofs are rigorous and do not invoke any heuristics or assumptions. Last, for the Euler equations, we show that an exponential lower bound can be obtained by considering instabilities (specifically, unstable fixed points in phase space) instead of chaotic dynamics. Considering unstable fixed points as opposed to chaotic regions of phase space increases the range of systems for which similar lower bounds can be proven because there are many systems that have unstable fixed points but are not chaotic. A classic example is the Kuramoto-Sivashinsky (KS) equation, which is a 1D nonlinear partial differential equation (PDE) of the form
\begin{gather}
    \frac{\partial u}{\partial t} + \frac{\partial^2 u}{\partial x^2} + \nu \frac{\partial^4 u}{\partial x^4}+ u\frac{\partial u}{\partial x} = 0.
\end{gather}
The KS equation is one of the simplest nonlinear PDEs that has been extensively studied, and can be used to model various physical phenomena such as flame propagation~\cite{michelson1977nonlinear,kuramoto1978diffusion,sivashinsky1988nonlinear}. The form of its nonlinearity is also identical to that of the KdV equation and similar to that of the Euler equations. It is well known that the dynamics of the KS equation become increasingly complex as $\nu$ decreases. For large $\nu$, the dynamics are relatively simple and stable. For $\nu$ sufficiently small, the dynamics become chaotic. However, in some intermediate regime, the dynamics exhibit instabilities without becoming chaotic~\cite{papageorgiou1991route,smyrlis1991predicting}. Thus, our lower bound proof technique may also extend to the KS equation in this unstable regime.

\subsubsection{Summary of main results and high-level overview of proofs}

Our main result gives lower bounds that show hardness of simulating fluid dynamics on quantum computers. For the KdV equation, we prove a lower bound for final state preparation that is polynomial in simulation time $T$. For the Euler equations, the lower bound for preparing the final state is exponential in $T$. As a corollary to these results, we provide lower bounds for history state preparation as well; a similar exponential lower bound applies to the Euler equations, though we obtain a weaker bound for the KdV equation.

We prove the lower bounds by investigating how fast states can separate when simulating fluid dynamics. While similar approaches have been used previously in the context of nonlinear quantum mechanics~\cite{lloyd1996universal,aaronson2005np,childs2016optimal} and general nonlinear dynamics~\cite{liu2021efficient}, our present goal is to show hardness results for specific equations of fluid dynamics, which is more technically challenging.

For the KdV equation, we use well-known analytical solutions, namely solitons. By examining how quickly two solitons with large initial overlap separate from each other, we constrain how efficiently a quantum algorithm can reproduce the dynamics.

\begin{figure}[ht!]
    \centering
    \includegraphics[width=0.7\linewidth]{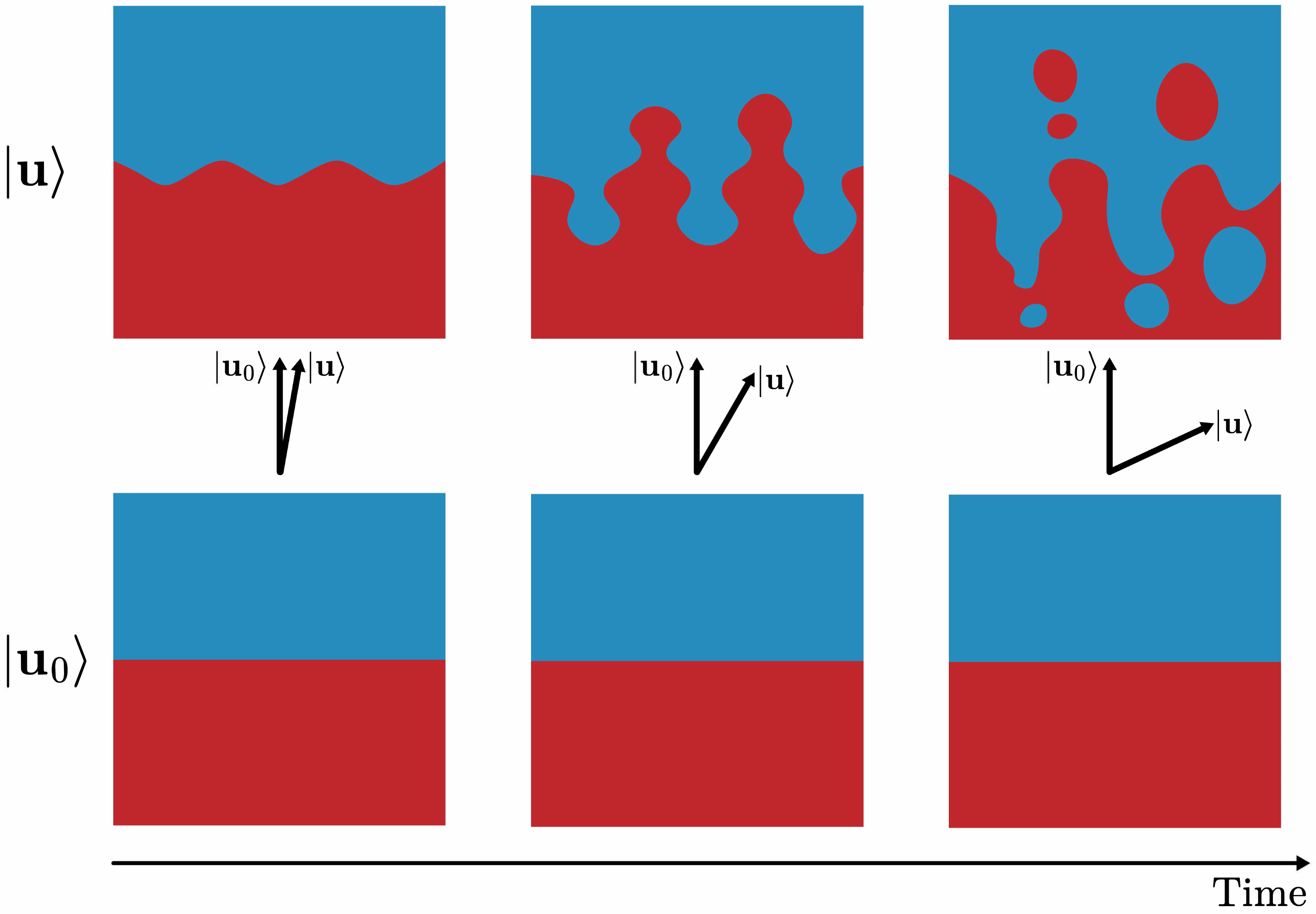}
    \caption{Visual demonstration of the technique used for proving our lower bound. Given an unstable equilibrium (bottom row) and a small perturbation to it (top row), the two states will be driven apart from each other very quickly. Visualized here is the Rayleigh-Taylor instability, whereby a dense fluid resting on top of a less dense fluid (with gravity pointing downwards) becomes unstable to perturbations~\cite{drazin2004hydrodynamic,chandrasekhar2013hydrodynamic}.}
    \label{fig:lower-bound-visual}
\end{figure}

For the Euler equations, we take advantage of well-known fluid instabilities. The Euler equations admit a wide variety of equilibria (i.e., fixed points of the dynamics). Some of these equilibria are stable, where small perturbations about the equlibrium relax and return the flow back to the equilibrium. However, there are also unstable equilibria, where a slight deviation from the equilibrium is amplified and the flow departs from the original fixed point. Such instabilities are characterized by an initial linear phase, where the perturbation grows at an exponential rate, followed by a transition to fully nonlinear dynamics. We show that instabilities can cause states with large initial overlap to become highly distinguishable, and in fact, can do so exponentially fast: for an initial overlap of $1-\epsilon$, the time to reduce the overlap to constant is only $\poly(\log(1/\epsilon))$. This is visually demonstrated in Figure~\ref{fig:lower-bound-visual}.
The main challenge with this approach is that we do not know how to give analytical solutions describing instabilities.\footnote{While some nonlinearly unstable solutions of the Navier-Stokes equations are known~\cite{polyanin2009nonlinearnavier,polyanin2009nonlinearhydro}, they are highly unphysical. We seek a lower bound for a physically meaningful scenario.} To circumvent the inaccessibility of exact nonlinear solutions, we first linearize the equations and solve those new equations. Then, we bound the error from linearization. This is a challenging task and cannot be done for general nonlinear dynamical systems. Fortunately, we find that it is possible to obtain an error bound for the incompressible Euler equations by taking advantage of a few key properties, namely the incompressibility condition, the two-dimensionality of the flow, as well as the boundary conditions we impose. Combining the linearized solution with the linearization error bound gives a bound on the inner product of the equilibrium and the nonlinear solution. It is crucial that the error bound from linearization is tight enough, as one would not otherwise be able to accurately pinpoint the behavior of the nonlinear solution, as illustrated in Figure~\ref{fig:error-bound-visual}. We find the bound to be sufficiently tight, allowing us to demonstrate the desired lower bound. 

\begin{figure}[ht!]
    \centering
    \includegraphics[width=0.7\linewidth]{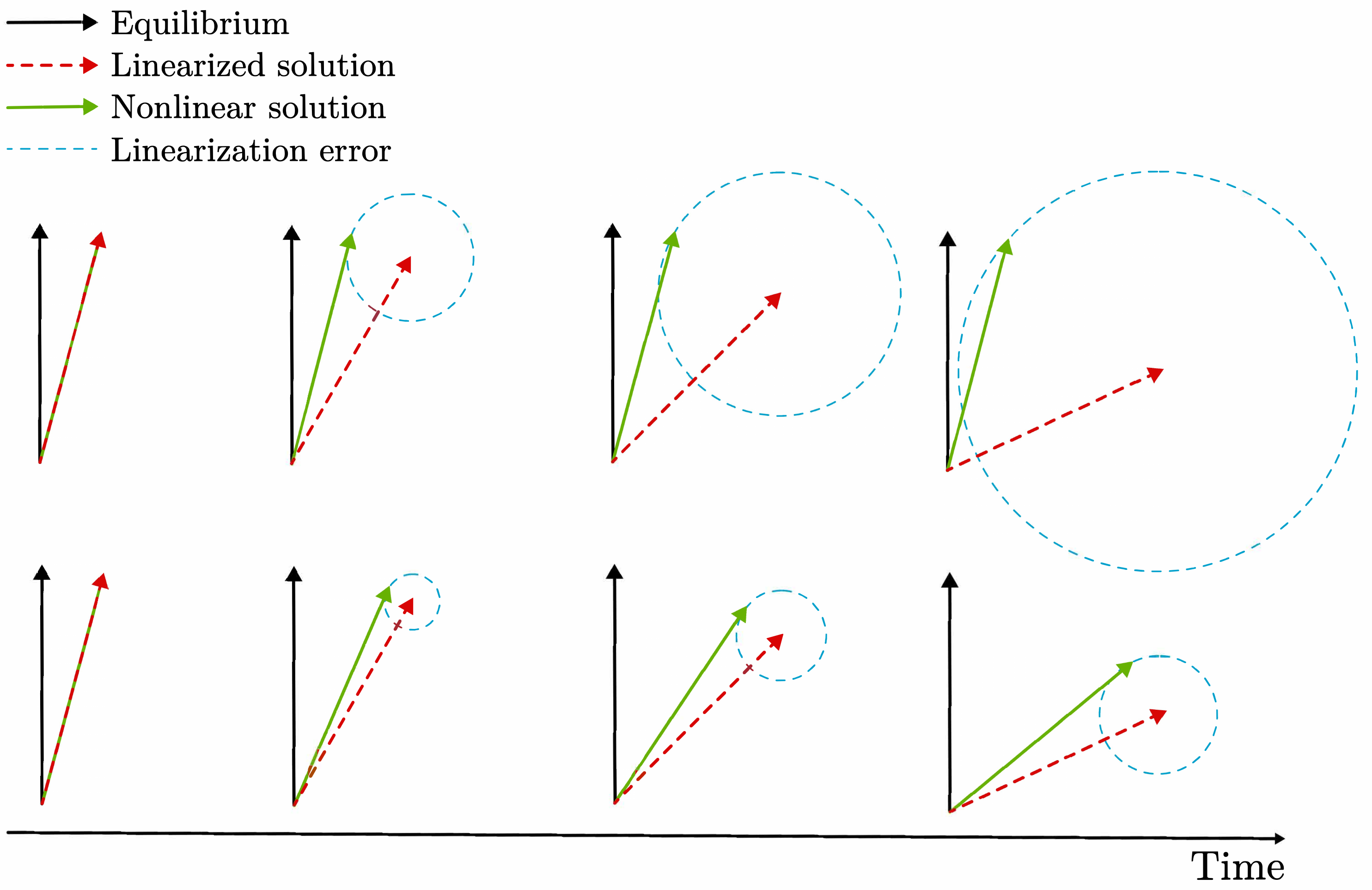}
    \caption{As the nonlinear solution (green) cannot be analytically obtained, we rely on the linearized solution (dashed red) to prove the lower bound. By bounding the error from linearization (dashed blue), we can bound the inner product between the equilibrium and the nonlinear solution. However, if the linearization error bound is not sufficiently tight, as shown in the top row, we would not be able to show the desired lower bound for simulating the nonlinear dynamics. Thus, it is crucial to obtain a sufficiently tight bound, as shown in the bottom row, to prove the desired result.
    }
    \label{fig:error-bound-visual}
\end{figure}

While our exponential lower bound is for the Euler equations, we expect that a similar bound should also apply to the Navier-Stokes equations (which model viscous fluids) with sufficiently high Reynolds numbers. Furthermore, we note that our lower bound is in terms of the simulation time $T$ only. Much of research in quantum algorithms for fluid dynamics is focused on attaining speedups with respect to the Reynolds number, as that determines the grid resolution required and is a resource bottleneck. Our work does not provide a lower bound in terms of the Reynolds number. However, if an exponential lower bound in $T$ indeed holds for high-Reynolds number viscous flows, this is likely to render long-time simulation of many fluid dynamics problems of interest intractable for quantum computers despite advantage with respect to the Reynolds number.

We also point out that both the KdV equation and the Euler equations are dissipationless. Thus, for any faithful discretization of these equations, the $\mathrm{R}$ parameter defined in~\cite{liu2021efficient} is infinite. However, this does not necessarily mean that the simulation of these equations should be intractable, as the $\mathrm{R}<1$ condition for the convergence of Carleman linearization is sufficient but not necessary. Our results show that simulation of some dissipationless equations is intractable, but of course there may still be special classes of fluid equations that can be simulated efficiently.

\subsection{Organization of paper}
The rest of the paper is structured as follows. Section~\ref{sec:preleminaries} clarifies the notation and provides background information on the KdV and Euler equations, state discrimination, as well as definitions of the problems for which we show lower bounds. Section~\ref{sec:kdv} establishes a polynomial lower bound for final state preparation of the KdV equation, by investigating how fast nearby solitons diverge. Section~\ref{sec:lb_final_state} proves the exponential lower bound for final state preparation of the Euler equations. This is done by using the smooth Kelvin-Helmholtz instability as a starting point, then showing an exponential lower bound for simulating the linearized Euler equations, and finally proving the lower bound for simulating the fully nonlinear equations by bounding the error from linearization. Section~\ref{sec:history-state} presents the lower bounds for history state preparation. Finally, Section~\ref{sec:future-work} discusses limitations of the lower bounds and whether one can circumvent them, lays out future research directions, and concludes.

\section{Preliminaries} \label{sec:preleminaries}
\subsection{Notation}
\paragraph{Sets.}Let $\mathbb N_0$, $\mathbb{N}$, $\mathbb Z$, $\mathbb R$, and $\mathbb C$ denote the set of nonnegative integers, positive integers, integers, real numbers, and complex numbers, respectively. For $d \in \mathbb N$, we define $[d]_0 \coloneq \{0,1,\dots,d\}$ and $[d]\coloneq \{1,2,\ldots,d\}$. Let $\Omega \subset \mathbb R^d$ be some some arbitrary subset of $\mathbb R^d$ with volume
\begin{equation}
    |\Omega| \coloneq \int_{\Omega} d\Omega.
\end{equation}
Finally, we denote the $d$-torus as $\mathbb T^d\coloneq [0,2\pi)^d$.
\paragraph{Vector fields and vector calculus.}Let $\boldsymbol{f}\colon\Omega \rightarrow \mathbb R^d$ be a vector-valued function. We denote the 2-norm of $\boldsymbol{f}$ as
\begin{equation}
    \|\boldsymbol{f}\| \coloneq \left (\sum_{j=1}^d |f_j|^2 \right)^{1/2}.
\end{equation}
We denote the $L^p$-norm of $\boldsymbol{f}$ for $p\in[1,\infty)$ by
\begin{equation}
    \|\boldsymbol{f}\|_{L^p}\coloneq \left (\int_\Omega \|\boldsymbol{f}\|^p \, d\Omega \right )^{1/p}.
\end{equation}
Furthermore, we define the $L^\infty$-norm as
\begin{equation}
    \|\boldsymbol{f}\|_{L^\infty}\coloneq \sup_{\mathbf{r}\in\Omega} \|\boldsymbol{f}\|.
\end{equation}
We also use the gradient operator, which, in $\mathbb T^d$ with coordinates $(x_1,\dots,x_d)$ and basis $\{\mathbf{e}_1,\dots,\mathbf{e}_d\}$, is defined as
\begin{equation}
    \nabla \coloneq \sum_{j=1}^d \frac{\partial}{\partial x_i} \mathbf{e}_i.
\end{equation}
For any vector field $\boldsymbol{f}$, we define the Jacobian
\begin{align}
    \mathcal D\boldsymbol{f} &\coloneq \begin{bmatrix}
        \frac{\partial f_1}{\partial x_1} & \dots & \frac{\partial f_1}{\partial x_d}  \\
        \vdots  & \ddots & \vdots \\
        \frac{\partial f_d}{\partial x_1} & \dots &  \frac{\partial f_d}{\partial x_d}
    \end{bmatrix}
\end{align}
and its transpose
\begin{equation}
    \nabla \boldsymbol{f} \coloneq \begin{bmatrix}
        \frac{\partial f_1}{\partial x_1} & \dots & \frac{\partial f_d}{\partial x_1}  \\
        \vdots  & \ddots & \vdots \\
        \frac{\partial f_1}{\partial x_d} & \dots &  \frac{\partial f_d}{\partial x_d}
    \end{bmatrix} = [\mathcal D\boldsymbol{f}]^T.
\end{equation}
The vorticity of a vector field $\boldsymbol f$ is defined as the skew-symmetric matrix
\begin{equation}
    \mathcal W_{\boldsymbol f} \coloneq \mathcal D\boldsymbol f - \nabla \boldsymbol f.
\end{equation}
When $d=3$, the vorticity matrix can be associated with the vector field
\begin{equation}
    \boldsymbol{\omega} \coloneq \nabla \times \boldsymbol{f},
\end{equation}
and when $d=2$, this reduces to the scalar
\begin{equation} \label{eq:vort-def}
    \omega \coloneq (\nabla \times \boldsymbol{f}) \cdot \hat{\mathbf k} = \frac{\partial f_2}{\partial x_1} - \frac{\partial f_1}{\partial x_2},
\end{equation}
where $\hat{\mathbf k}$ is the unit vector normal to $\Omega$.

We use the following powerful result.
\begin{lem}[\cite{bahouri2011fourier}, Proposition 7.5] \label{lem:calderon-zygmund}
Let $\mathbf u\colon\Omega\to \mathbb R^d$ be a vector field with gradient in $L^p$. Suppose $\nabla \cdot \mathbf u = 0$. Then there exists a constant $C$, depending only on the dimension $d$, such that for any $1<p<\infty$ we have
    \begin{equation}
        \|\nabla \mathbf u\|_{L^p} \leq C \frac{p^2}{p-1} \|\mathcal W_{\mathbf u} \|_{L^p}.
    \end{equation}
\end{lem}
\paragraph{Laplace equation.} When working with the Euler equations, we will encounter the Laplace equation
\begin{equation}\label{eq:laplace}
    \nabla^2 \phi = f.
\end{equation}
Here, $\phi,f\colon\Omega \to \mathbb R^d$. We will need the following result on the $L^2$-norm of the solution to the Laplace equation.

\begin{restatable}{lem}{laplace}\label{lem:laplace}
    Consider the Laplace equation in Eq.~\eqref{eq:laplace} on $\Omega = \mathbb T^d$. Then, 
    \begin{equation}
        \|\phi\|_{L^2} \leq \bar \phi + \|f\|_{L^2},
    \end{equation}
    where $\bar \phi\geq 0$ is a constant.
\end{restatable}
The proof of this Lemma can be found in Appendix~\ref{app:laplace}.

\paragraph{Hölder's inequality.}
Let $p,q\in[1,\infty]$ with $1/p + 1/q = 1$ and let $f,g\colon\Omega \to \mathbb R$ be scalar functions. Then
\begin{equation}
    \|fg\|_{L^1} \leq \|f\|_{L^p} \|g\|_{L^q}.
\end{equation}

\paragraph{Quantum.} For vectors in $\mathbb C^n$, we use $\norm{\cdot}$ to denote the standard Euclidean norm. For operators over the same space, we use $\norm{\cdot}_1$ to denote the trace norm, defined as
\begin{equation}
    \|A\|_{1} \coloneq \mathrm{Tr}\left[\sqrt{A^\dagger A}\right].
\end{equation}
Pure quantum states are described by unit-norm vectors. Given two pure states $\ket{\psi}, \ket{\phi}$ we denote the Euclidean distance by \begin{align}
    d(\ket{\psi}, \ket{\phi}) &\coloneqq \min_{z, |z|=1} \norm{\ket{\psi} - z\ket{\phi}}.
\end{align}
This distance can also be defined through the overlap as $d(\ket{\psi}, \ket{\phi})=\sqrt{2-2\abs{\braket{\psi \, | \, \phi}}}$.
Given two density matrices (i.e., unit-trace positive semidefinite matrices) $\rho, \sigma$, we denote their trace distance \begin{align}
    D(\rho, \sigma) &\coloneqq \frac{1}{2}\norm{\rho - \sigma}_1.
\end{align}
We exclusively use these distance metrics for quantum states. 

For a pure state $\ket{\psi}$, we denote the corresponding mixed state as $\rho_{\psi} = \ket{\psi}\bra{\psi}$. A quantum channel is a completely positive trace preserving linear map on the space of operators. These characterize all physical operations that can be performed on a quantum device.

We use the following standard facts. Lemmas~\ref{fact:euc-tr-rel} and~\ref{fact:tr-prod-rel} are proven in Appendix~\ref{app:std-proofs}.

\begin{restatable}{lem}{euclidtrace}
    For any two pure states $\ket{\psi}, \ket{\phi}$, the Euclidean and trace distance between them satisfy the inequalities \begin{align}
        \frac{1}{\sqrt{2}}d(\ket{\psi}, \ket{\phi}) \leq D(\rho_\psi, \rho_\phi) \leq d(\ket{\psi}, \ket{\phi}).
    \end{align}
    \label{fact:euc-tr-rel}
\end{restatable}

\begin{restatable}{lem}{kcopies}
    For any pure states $\ket{\psi}, \ket{\phi}$, trace distance between $k$ copies of the two states satisfies \begin{align}
        D(\rho_\psi^{\otimes k}, \rho_\phi^{\otimes k}) \leq \sqrt{k}\cdot D(\rho_\psi, \rho_\phi).
    \end{align}
    \label{fact:tr-prod-rel}
\end{restatable}

\begin{lem}[\cite{nielsen2010quantum}, Theorem 9.2]
    Quantum channels are contractive, meaning for any quantum channel $\mathcal E$ and mixed states $\rho, \sigma$, we have \begin{align}
        D(\rho, \sigma) \geq D(\mathcal E(\rho), \mathcal E(\sigma)).
    \end{align}
    \label{fact:maps-contract}
\end{lem}

\subsection{Korteweg-de Vries equation}
The Korteweg-de Vries (KdV) equation is a one-dimensional dispersive nonlinear PDE that takes the form
\begin{align} \label{eq:kdv}
    \frac{\partial \phi}{\partial t} + \frac{\partial^3 \phi}{\partial x^3} - 6\phi \frac{\partial \phi}{\partial x} = 0,
\end{align}
where $\phi\colon \mathbb \Omega \times \mathbb R_+ \to \mathbb R$. The KdV equation can model a wide variety of phenomena, such as shallow water waves and nonlinear Langmuir waves in plasmas~\cite{korteweg1895xli,chen1984introduction,linares2014introduction}. On the left-hand side, the second term is a linear dispersive term that causes waves of different wavenumbers to move at different velocities. The third term is the nonlinear term that introduces non-trivial phenomena, such as wave steepening, into the dynamics.

We consider $C^\infty$ solutions on $\Omega = \mathbb R$ that vanish at infinity, i.e., $\lim_{x\rightarrow \pm \infty} \phi(x,t) = 0$. With this property, the KdV equation preserves the norm.

\begin{lem} [Conservation of norm, KdV equation] \label{lem:kdv-norm-cons}
    Consider the KdV equation in Eq.~(\ref{eq:kdv}) with $\phi\in C^\infty$ and the additional property that $\lim_{x\to\pm \infty}\phi(x,t) = 0$. Then
    \begin{align}
        \frac{d\|\phi\|_{L^2}}{dt} = 0.
    \end{align}
    
\end{lem}
\begin{proof}
    A straightforward calculation shows
    \begin{align}
        \frac{d\|\phi\|_{L^2}^2}{dt} &= \frac{d}{dt}\int_{\mathbb R} \phi^2 \, dx \\
        &= 2\int_{\mathbb R} \phi \frac{\partial \phi}{\partial t} \, dx \\
        &= 2\int_{\mathbb R} \phi \left [ 6\phi \frac{\partial \phi}{\partial x} - \frac{\partial^3 \phi}{\partial x^3}  \right] dx \\
        &= 2\int_\mathbb R \left [2\frac{\partial \phi^3}{\partial x} - \phi \frac{\partial^3 \phi}{\partial x^3} \right] dx.
    \end{align}
    The first term vanishes by using the divergence theorem and the fact that $\phi$ vanishes at infinity. The second term also vanishes, as we can show using integration by parts:
    \begin{align}
        \int_\mathbb R \phi \frac{\partial^3 \phi}{\partial x^3} \, dx &= -\int_\mathbb R \frac{\partial \phi}{\partial x} \frac{\partial^2 \phi}{\partial x^2} \, dx \\
        &= \int_\mathbb R \frac{\partial^2 \phi}{\partial x^2} \frac{\partial \phi}{\partial x} \, dx \\
        &= - \int_\mathbb R  \frac{\partial^3 \phi}{\partial x^3} \phi \, dx.
    \end{align}
    Since $\|\phi\|_{L^2} \neq 0$,
    \begin{equation}
        \frac{d\|\phi \|_{L^2}^2}{dt} = 0 \Rightarrow \frac{d\|\phi \|_{L^2}}{dt}=0.
    \end{equation}
\end{proof}

\subsection{Euler equations of fluid dynamics}

The simplest equations that capture the essence of fluid dynamics are the Euler equations. Specifically, we consider the incompressible Euler equations, given by
\begin{equation} \label{eq:incomp-euler}
    \begin{aligned}
        \frac{\partial \mathbf u}{\partial t} + (\mathbf u \cdot \nabla)\mathbf u &= -\frac{1}{\rho}\nabla P \\
        \nabla \cdot \mathbf u &= 0,
    \end{aligned}
\end{equation}
where $\mathbf{u}\colon \Omega \times \mathbb R_+ \to \mathbb R^{d}$ is the fluid velocity, $\rho\colon \Omega \times \mathbb R_+ \to \mathbb R$ is the density, and $P\colon \Omega \times \mathbb R_+ \to \mathbb R$ is the pressure, with $\Omega$ being the physical domain and $d\in[3]$ being the physical dimension of the problem. The first equation is the momentum equation, and the second is the incompressibility condition. Nonlinearity appears through the advection term, which is the second term on the left-hand side of the momentum equation. The incompressible Euler equations describe the behavior of an ideal (inviscid) fluid. Since the flow is incompressible, $\rho$ is constant for a single-constituent fluid; we set $\rho=1$ henceforth. 

As the momentum equation only governs the behavior of $\mathbf{u}$, we need an equation to solve for $P$. We can obtain this by taking the divergence of the momentum equation and using the incompressibility condition to get
\begin{equation} \label{eq:pressure-laplace}
    \nabla^2 P = -\nabla \cdot ( (\mathbf{u}\cdot \nabla ) \mathbf u ).
\end{equation}

Since Eq.~\eqref{eq:pressure-laplace} is a Laplace equation and only the gradient of the pressure appears in Eq.~\eqref{eq:incomp-euler}, we are free to choose a constant gauge for the pressure. The pressure gauge is defined as
\begin{equation}
    \bar p \coloneq \frac{1}{|\Omega|} \int_\Omega P \, d\Omega.
    \label{eq:pressure_gauge}
\end{equation}
The value of the gauge is not important as long as we are consistent.

Under certain constraints, the Euler equations possess conserved quantities, which will be useful later. The first conserved quantity is the $L^2$-norm of $\mathbf u$, which corresponds to total energy.

\begin{lem}[Conservation of norm, Euler equations] \label{lem:euler-norm-cons}
Consider the Euler equations in Eq.~(\ref{eq:incomp-euler}) on $\Omega = \mathbb T^d$. For $\mathbf{u},p\in C^\infty$ we have 
\begin{equation}
    \frac{d\|\mathbf u \|_{L^2}}{dt} = 0.
\end{equation}

\end{lem}
\begin{proof}
    Consider the time evolution of $\|\mathbf u \|_{L^2}^2$. We use the following vector calculus identity:
    \begin{equation}
        \frac{1}{2}\nabla(\mathbf{u}\cdot \mathbf{u}) = (\mathbf{u}\cdot\nabla)\mathbf{u} + \mathbf{u}\times(\nabla\times \mathbf{u}).
    \end{equation}
    A straightforward calculation shows
    \begin{align}
        \frac{d\|\mathbf u \|_{L^2}^2}{dt} &= \frac{d}{dt} \int_{\mathbb T^d} \mathbf{u}\cdot \mathbf{u} \,  d^3 \mathbf{r} \\
        &= 2 \int_{\mathbb T^d} \mathbf u \cdot \frac{\partial \mathbf u}{\partial t} \, d^3 \mathbf r \\
        &= -2 \int_{\mathbb T^d} \mathbf u \cdot [(\mathbf u \cdot \nabla)\mathbf u + \nabla P] \, d^3 \mathbf r \\
        &= -2 \int_{\mathbb T^d} \mathbf u \cdot \left ( \frac{1}{2} \nabla (\mathbf u \cdot \mathbf u )  - \mathbf u \times (\nabla \times \mathbf u) + \nabla P \right ) d^3 \mathbf r \\
        &= - \int_{\mathbb T^d} \mathbf u \cdot \nabla (u^2 + 2P) \, d^3 \mathbf r,
    \end{align}
    with $u^2 = \mathbf u \cdot \mathbf u$. Using vector calculus identities, we can write the integrand as
    \begin{equation}
        \mathbf u \cdot \nabla (u^2 + 2P)= \nabla \cdot ((u^2+2P) \mathbf u) - (u^2+2P) \nabla \cdot \mathbf u.
    \end{equation}
   The integral of the first term vanishes by the divergence theorem (since the domain $\mathbb T^d$ has no boundary), and the second term vanishes by the incompressibility condition. Given that $\|\mathbf{u}\|_{L^2} \neq 0$,
    
    \begin{equation}
        \frac{d\|\mathbf u \|_{L^2}^2}{dt} = 0 \Rightarrow \frac{d\|\mathbf u \|_{L^2}}{dt}=0.
    \end{equation}
\end{proof}

By narrowing down our focus to $\Omega = \mathbb T^2$, we obtain several more results. The first is the conservation of vorticity.
\begin{lem} [Conservation of vorticity] \label{lem:cons-vorticity}
    Consider the Euler equations in Eq.~\eqref{eq:incomp-euler} on $\Omega = \mathbb T^2$. Let $\omega$ be the vorticity of $\mathbf u$, as defined in Eq.~\eqref{eq:vort-def}. Then
    \begin{equation}
        \frac{d\|\omega \|_{L^4}}{dt} = 0.
    \end{equation}
\end{lem}
\begin{proof}
    Let
    \begin{equation}
        \boldsymbol \omega \coloneq \nabla \times \mathbf u
    \end{equation}
    and
    \begin{equation}
        \omega = \boldsymbol{\omega}\cdot \hat{\mathbf k}.
    \end{equation}
    First we obtain an equation for $\omega$. We take the curl of the momentum equation in Eq.~\eqref{eq:incomp-euler}:
    \begin{align}
        \frac{\partial \boldsymbol{\omega}}{\partial t} &=-\nabla \times [(\mathbf u \cdot \nabla ) \mathbf u ] -  \nabla \times (\nabla P).
    \end{align}
    The second term is zero, and the first can be rewritten via vector calculus identities to give
    \begin{align}
        \frac{\partial \boldsymbol{\omega}}{\partial t} &=-\nabla \times \left[\frac{1}{2}\nabla (\mathbf u \cdot \mathbf u) - \mathbf u \times (\nabla \times \mathbf u) \right ] \\
        &= \nabla \times ( \mathbf u \times \boldsymbol{\omega}) \\
        &= \mathbf u \cdot (\nabla \cdot \boldsymbol{\omega}) - \boldsymbol{\omega} (\nabla \cdot \mathbf u) + (\boldsymbol{\omega} \cdot \nabla) \mathbf u - (\mathbf u \cdot \nabla ) \boldsymbol{\omega} \label{eq:vort-line-3} \\
        &= -(\mathbf u \cdot \nabla ) \boldsymbol{\omega}.
    \end{align}
    In Eq.~\eqref{eq:vort-line-3}, the first term vanishes as $\nabla \cdot \boldsymbol{\omega}$ is the divergence of a curl, the second vanishes due to incompressibility, and the third vanishes since $\boldsymbol{\omega}\cdot \nabla = 0$. Taking the dot product with $\hat{\mathbf k}$ gives
    \begin{equation}
        \frac{\partial \omega}{\partial t} + (\mathbf u \cdot \nabla ) \omega = 0.
    \end{equation}
    
    We now obtain an equation for $\|\omega\|_{L^4}^4$:
    \begin{align}
        \frac{d\|\omega\|_{L^4}^4}{dt} &= \frac{d}{dt}\int_{\mathbb T^2} |\omega^4| \, d^2 \mathbf r \\
        &= 4\int_{\mathbb T^2} \omega^3 \frac{\partial \omega}{\partial t} \, d^2 \mathbf r \\
        &= -4\int_{\mathbb T^2} \omega^3 (\mathbf u \cdot \nabla ) \omega \, d^2 \mathbf r \\
        &= -\int_{\mathbb T^2} \mathbf u \cdot \nabla \omega^4 \, d^2 \mathbf r.
    \end{align}
    We can rewrite the integrand as
    \begin{align}
        \mathbf{u}\cdot \nabla \omega^4 &= \nabla \cdot (\omega^4 \mathbf u) - \omega^4 \nabla \cdot \mathbf u.
    \end{align}
    The second term vanishes due to incompressibility and the integral of the first vanishes due to the divergence theorem. Given that $\|\omega\|_{L^4} \neq 0$,
    
    \begin{equation}
        \frac{d\|\omega \|_{L^4}^4}{dt} = 0 \Rightarrow \frac{d\|\omega \|_{L^4}}{dt}=0
    \end{equation}
    as claimed.
\end{proof}
\begin{remark}
    When $\Omega = \mathbb T^2$, one can show that $\|\omega\|_{L^p}$ is conserved for $1\leq p\leq \infty$~\cite{chemin1998perfect}. In this paper, we only use the $p=4$ case.
\end{remark}

Finally, we prove the following bound on the norm of the pressure.
\begin{lem} \label{lem:pressure-bounds}
Consider the Euler equations in Eq.~\eqref{eq:incomp-euler} on $\Omega = \mathbb T^2$. Let $\omega_0$ be the vorticity at $t=0$. Then
    \begin{equation}
       2\pi \bar p \leq \|P\|_{L^2}  \leq 2\pi\bar p + C \|\omega_0\|_{L^4}^2
    \end{equation}
for some constant $C$ and gauge $\bar p$.
\end{lem}
\begin{proof}
    First, we decompose the pressure into a constant gauge and zero-mean part. In other words,
    \begin{equation}
        P = \bar p + p',
    \end{equation}
    such that $\nabla^2 \bar p = 0$, $\nabla^2 p' = -\nabla \cdot ((\mathbf u \cdot \nabla ) \mathbf u)$, and
    \begin{equation}
        \int_{\mathbb T^2} p' \, d^2\mathbf r = 0.
    \end{equation}
    Now we prove the upper bound. Recall that the pressure is solved via Eq.~\eqref{eq:pressure-laplace}, which is a Laplace equation. From Lemma~\ref{lem:laplace}, we find that
    \begin{equation}
        \|p'\|_{L^2} \leq \|\nabla \cdot ((\mathbf u \cdot \nabla) \mathbf u ) \|_{L^2}.
    \end{equation}
    Let us briefly use Einstein summation notation for the right-hand side:
    \begin{align}
        \nabla \cdot ((\mathbf u \cdot \nabla) \mathbf u ) &= \partial_i (u_j \partial_j u_i) \\
        &= \partial_i u_j \partial_j u_i + u_j \partial_{ji}u_i,
    \end{align}
    with the second term vanishing due to incompressibility. Thus, we get
    \begin{align}
        \|p'\|_{L^2} &\leq \|\partial_i u_j \partial_j u_i\|_{L^2} \\
        &= \left (\int_{\mathbb T^2} (\partial_i u_j \partial_j u_i)^{2} \, d^2\mathbf r \right )^{1/2} \\
        &\leq \left (\int_{\mathbb T^2} \|\nabla \mathbf u \|^4 \, d^2\mathbf r \right )^{1/2} \\
        &= \|\nabla \mathbf u \|_{L^4}^2.
    \end{align}
    From \cref{lem:calderon-zygmund}, we get
    \begin{equation}
        \|p'\|_{L^2} \leq \frac{256}{9} C'\|\omega\|_{L^4}^2
    \end{equation}
    for some constant $C'$. Since $\|\omega\|_{L^4}$ is conserved from \cref{lem:cons-vorticity}, $\|\omega\|_{L^4} = \|\omega_0 \|_{L^4}$. The upper bound follows from $\|P\|_{L^2} \leq \|\bar p\|_{L^2} + \|p'\|_{L^2}$.

    Now we prove the lower bound. A simple calculation shows
    \begin{align}
        \|P\|_{L^2}^2  &= \int_{\mathbb T^2} \bar p^2 + 2\bar p p' + p'^2 \, d^2 \mathbf r \\
        &= (2\pi)^2 \bar p ^2 + \int_{\mathbb T^2} p'^2 \, d^2 \mathbf r\\
        &\geq (2\pi)^2 \bar p ^2,
    \end{align}
    and the lower bound follows.
\end{proof}

While the Euler equations are a natural starting point for understanding fluid dynamics, the Navier-Stokes equations provide a more accurate depiction of fluids. The incompressible Navier-Stokes equations (with $\rho = 1$) are
\begin{equation} \label{eq:incomp-NS}
    \begin{aligned}
        \frac{\partial \mathbf u}{\partial t} + (\mathbf u \cdot \nabla)\mathbf u &= -\nabla P + \nu\nabla^2 \mathbf u  \\
        \nabla \cdot \mathbf u &= 0 ,
    \end{aligned}
\end{equation}
where $\nu$ is the viscosity. A key dimensionless quantity in the Navier-Stokes equations is the Reynolds number
\begin{equation}
    \mathrm{Re} \coloneq  \frac{UL}{\nu},
\end{equation}
where $U$ is a characteristic flow velocity (typically the maximum flow velocity) and $L$ is a characteristic length (related to the physical dimensions of the system being considered). As briefly discussed earlier, the Reynolds number characterizes the strength of the nonlinearity compared to the dissipation, similar to $\mathrm{R}$ in~\cite{liu2021efficient}, with large $\mathrm{Re}$ resulting in turbulent and highly nontrivial flows. While the lower bound we develop is for the Euler equations, which have infinite Reynolds number, we expect our lower bound to apply to the incompressible Navier-Stokes for sufficiently larger $\mathrm{Re}$, but we leave a detailed exploration of this for future work. In particular, can one establish a concrete value of the Reynolds number above which the Navier-Stokes equations are intractable for quantum algorithms? 

\subsection{Problem definition}
\label{sec:probdef}
Quantum algorithms for differential equations generally represent a system by encoding the phase-space variables of the system in amplitudes of a quantum state. 
This raises the challenge that quantum evolution preserves the norm of the state and is invariant under scaling the state by a constant factor, features that do not hold for general differential equations. However, there are many ways around this problem.

For instance, one can always rescale an equation so that all states of interest have norm at most one, leaving only subnormalized states. These can be represented by a quantum state with some component on a flag state, say $\ket{\bot}$, holding the remaining amplitude. In particular, if a system is described by a vector $\boldsymbol v \in \mathbb R^n$ with $\norm{\boldsymbol v} \leq 1$, this can be represented as a quantum state
\begin{align}
    \ket{\psi_{\boldsymbol v}} &= \sqrt{1-\norm{\boldsymbol v}^2} \ket{\bot} + \sum_{x \in [n]} v_x \ket{x} \\
    &= \sqrt{1-\norm{\boldsymbol v}^2} \ket{\bot} + \ket{\boldsymbol v}.
    \label{eq:flag_encoding}
\end{align}

With this setup in mind, we can define the time evolution problem for the KdV equation as follows.
Note that this equation conserves norm, as established in \Cref{lem:kdv-norm-cons}, so a promise of subnormalization on the input norm suffices.

\begin{problem} [Final-state preparation of KdV equation] \label{prob:final-state-kdv}
    Let $T$ be some final simulation time, and consider the KdV equation on the one-dimensional domain $\mathbb R$. For some initial condition $\phi(x,0) \colon \mathbb R \rightarrow \mathbb R$ that is promised to be subnormalized (i.e., $\norm{\phi(x,0)}_{L^2} \leq 1$),
    we are given $k$ copies of the encoding $\ket{\psi_{\phi}(0)}$, as in Eq.~\eqref{eq:flag_encoding}. The goal is to produce a final state $|\psi'\rangle$ such that $\||\psi'\rangle - \ket{\psi_\phi(T)}\| \leq \epsilon$, where $\ket{\psi_\phi(T)} = \sqrt{1-\norm{\phi(T)}^2} \ket{\bot} + \ket{\phi(T)}$.
\end{problem}

For the Euler equations, the subnormalized encoding is not the most natural choice because, while $\|\mathbf{u}\|_{L^2}$ remains normalized by \cref{lem:euler-norm-cons}, $\|P\|_{L^2}$ can fluctuate. Instead, we consider two different encodings. The first is a tensor-product encoding, where the velocity and pressure are kept in separate quantum registers.
\begin{problem} [Final-state preparation of Euler equations, tensor-product encoding] \label{prob:final-state-euler-tensor-prod} 
    Let $\mathbf{r}\in \mathbb T^2$ and $T$ be some final simulation time. For some initial condition $\mathbf{u}(\mathbf{r},0)$ and $P(\mathbf{r},0)$ which are normalized, i.e., $\|\mathbf{u}(\mathbf r,0)\|_{L^2} = \|P(\mathbf{r},0)\|_{L^2} = 1$, we are given $k_{\mathbf u}$ copies of $|\mathbf{u}(0)\rangle$ and $k_P$ copies of $|P(0)\rangle$. The goal is to produce a normalized encoding of $\mathbf{u}(\mathbf{r},T)$ and $P(\mathbf{r},T)$ under evolution of the incompressible Euler equations, i.e., a normalized state $\ket{\psi'(T)}$ such that $\norm{\ket{\psi'(T)} - |\mathbf{u}(T)\rangle \otimes |P(T)\rangle} \leq \epsilon$.
\end{problem}
The tensor-product encoding above is quite natural for the Euler equations. This is because pressure and velocity are inherently separate entities, and they are computed separately from each other. Classical algorithms that simulate the Euler equations via time marching solve the momentum equation at each timestep to update $\mathbf{u}$, and then solve for $P$ via the Laplace equation of Eq.~\eqref{eq:pressure-laplace}. However, it is also natural to consider a direct-sum encoding, where the output of the quantum algorithm is a state that contains both the velocity and the pressure in its components. We define the direct-sum encoding problem below.
\begin{problem} [Final-state preparation of Euler equations, direct-sum encoding] \label{prob:final-state-euler-direct-sum} 
    Let $\mathbf{r}\in \mathbb T^2$ and $T$ be some final simulation time.  Let
    \begin{equation}
        \mathbf w(\mathbf r,t) \coloneq \begin{bmatrix}
            \mathbf u (\mathbf r,t) \\
            P (\mathbf r,t)
        \end{bmatrix}.
    \end{equation}
    For some initial condition $\mathbf{u}(\mathbf{r},0)$ and $P(\mathbf{r},0)$, we are given $k$ copies of $|\mathbf w(0)\rangle$, which is an encoding of $\mathbf w(\mathbf r,0)$ such that $\|\mathbf w(\mathbf r,0)\|_{L^2} = 1$. The goal is to produce a normalized encoding of $\mathbf w(\mathbf r,T)$ under evolution of the incompressible Euler equations, i.e., a normalized state $|\psi'(T)\rangle$ such that $\norm{\ket{\psi'(T)} - |\mathbf w(T)\rangle  } \leq \epsilon$.
\end{problem}

\begin{remark}
    Note that we consider a continuous-variable encoding of the solution. For instance, in two spatial dimensions, the encoded velocity field $|\mathbf{u}(t)\rangle$ for the tensor-product encoding of the Euler equations takes the form
    \begin{align}
        \langle0,\mathbf{r}|\mathbf{u}(t)\rangle = u_x(\mathbf{r},t)/\|\mathbf{u}\|_{L^2} \\
        \langle1,\mathbf{r}|\mathbf{u}(t)\rangle = u_y(\mathbf{r},t)/\|\mathbf{u}\|_{L^2} .
    \end{align}
    In practice, a digital quantum computer would discretize these states over a discrete spatial domain. Our results apply to this setting as well, provided the discretized dynamics approximate the continuous dynamics sufficiently well. The continuous-variable encoding we consider is advantageous as it remains agnostic about the specifics of the discretization method.
\end{remark}

While many quantum algorithms for differential equations output the final state, another form of output is a history state, which contains the solution not just at the final time $T$, but also at previous times. Typically, the time interval $[0,T]$ is divided into $m$ intervals of equal length $h$ such that $T=mh$, and the history state is a superposition of a clock register entangled with the solution at times $jh$ for $j\in[m]$. We also show hardness results for history state preparation as a consequence of the hardness of final state preparation.

\subsection{State discrimination lower bounds} \label{sec:statediscrimination}

Close quantum states can only be reliably distinguished given many copies. This fact can be used to derive lower bounds on quantum algorithms for differential equations. Let $M$ be a (potentially non-unitary) operator on quantum states that maps pure states to pure states. Viewed as an operator on mixed states, $M$ maps rank-one operators to rank-one operators. In the setting of simulating differential equations, this map takes a pure-state encoding of an initial condition to a pure-state encoding of its time evolution.
We say that a quantum algorithm $\mathcal A$ \emph{implements} $M$ to precision $\delta$ with copy complexity $k$ if, for any $\ket{\psi} \in \mathsf{dom}(M)$, we have 
\begin{align}
    D(\mathcal A(\rho_\psi^{\otimes k}), M(\rho_{\psi})) \leq \delta
\end{align}
(recall that $\rho_\psi$ denotes the density matrix of $\ket{\psi}$).
If $M$ is far from unitary, then the copy complexity of implementing $M$ must be high. One particular way $M$ can be far from unitary is if it separates close quantum states. This can be used to lower bound the copy complexity as follows.

\begin{lem}
    Let $M$ be an operator taking pure states to pure states, and $\ket{\psi}, \ket{\phi}$ be quantum states satisfying $D(\rho_{\psi},\rho_{\phi}) = \epsilon_0$. Suppose that $D(M(\rho_\psi), M(\rho_\phi)) = \epsilon_f$. Then any quantum algorithm implementing $M$ to precision $\delta < \epsilon_f/2$ uses at least
    \begin{align}
        k &\geq \left(\frac{\epsilon_f-2\delta}{\epsilon_0}\right)^2.
    \end{align}
    copies.
    \label{lem:state-dist-comp}
\end{lem}
\begin{proof}
    By the correctness of the given algorithm $\mathcal A$, we have \begin{align}
        D(\mathcal A(\rho_\psi^{\otimes k}), \mathcal A(\rho_\phi^{\otimes k}) &\geq \underbrace{D(M(\rho_\psi), M(\rho_\phi))}_{=\epsilon_f} - \underbrace{D(\mathcal A(\rho_\psi^{\otimes k}), M(\rho_\phi))}_{\leq \delta} - \underbrace{D(M(\rho_\psi), \mathcal A(\rho_\phi^{\otimes k}))}_{\leq \delta} \\
        &\geq \epsilon_f - 2\delta
    \end{align}
    where the first line follows from triangle inequality. On the other hand, we have \begin{align}
        D(\mathcal A(\rho_\psi^{\otimes k}), \mathcal A(\rho_\phi^{\otimes k})) &\leq D(\rho_\psi^{\otimes k}, \rho_\phi^{\otimes k}) & \text{(\Cref{fact:maps-contract})} \\
        &\leq \sqrt{k} \cdot \epsilon_0. & \text{(\Cref{fact:tr-prod-rel})}
    \end{align}
    Combining these two inequalities gives $\sqrt{k}\cdot\epsilon_0 \geq \epsilon_f-2\delta$, which is equivalent to the claim.
\end{proof}

These results can be used to lower bound the resources needed to implement time evolution. 
As a corollary of \Cref{lem:state-dist-comp}, differential equations that separate close states require many copies of the initial state to simulate.

\begin{corr}
    Let $\ket{v(0)}, \ket{v'(0)}$ and $\ket{v(T)}, \ket{v'(T)}$ be initial and final states, respectively, of a differential equation. Suppose $\norm{\ket{v(0)}} = \norm{\ket{v(T)}} = \Omega(1)$, and similarly for $v'$. 
    Then exactly evolving for time $T$, as in \Cref{prob:final-state-kdv}, requires
    \begin{align}
        k &= \Omega\left(\left(\frac{d(v(T), v'(T))}{d(v(0),v'(0))}\right)^2\right)
    \end{align}
    copies of the initial state, under the encoding $\ket{\psi_v}$ from \Cref{eq:flag_encoding}.
    \label{cor:eq-evol-lower}
\end{corr}

Our exponential lower bound for the Euler equation uses the following more specific result.
\begin{corr} \label{corr:state-disc-expon}
    Suppose there is a pair of solutions, $v$ and $v'$, to a differential equation satisfying 
    \begin{equation}
        |\langle v(0)|v'(0)\rangle | = 1-O(\epsilon)
    \end{equation}
    \begin{equation}
        |\langle v(T)|v'(T)\rangle | = 1-\Omega(\epsilon^{\alpha})
    \end{equation}
    with $T = O(\log(1/\epsilon))$, for some constant $0<\alpha<1$. Then any quantum algorithm simulating the differential equation exactly requires
    \begin{equation}
        k = e^{\Omega(T)}
    \end{equation}
    copies of the initial state.
\end{corr}
\begin{proof}
    By definition, we have
    \begin{align}
        d(|v(0)\rangle ,|v'(0)\rangle) &= O(\epsilon) \\
        d(|v(T)\rangle ,|v'(T)\rangle) &= \Omega(\epsilon^{\alpha}).
    \end{align}
    Therefore $k = \Omega(\epsilon^{-(1-\alpha)})$. Furthermore, $T = O(\log(1/\epsilon))$ implies $\epsilon = e^{-\Omega(T)}$, so $k = e^{\Omega(T)}$ as claimed.
\end{proof}

\section{Lower bound for simulating the Korteweg-de Vries equation} \label{sec:kdv}

Consider the KdV equation in one dimension, given by Eq.~(\ref{eq:kdv}). It is a standard result that this equation has a single-soliton solution given by \begin{align}
    \phi(x, t) = -\frac{c}{2} \sech^2\left(\frac{\sqrt{c}}{2} (x-ct+a)\right),
\end{align}
for arbitrary constants $a, c$~\cite{drazin1989solitons,linares2014introduction}.

Let $0<\delta\ll 1$ be a small perturbation strength, and consider the divergence of single-soliton solutions with velocity $c=1$ and $c=1+\delta$. We can write down corresponding single-solution solutions (with offset $a=0$) as 
\begin{align}
    \phi(x, t) =& -\frac{1}{2} \sech^2\left(\frac{1}{2}(x-t)\right), & 
    \phi'(x, t) =& -\frac{1+\delta}{2} \sech^2\left(\frac{\sqrt{1+\delta}}{2}(x-(1+\delta)t)\right).
\end{align}
One can compute \begin{align}
    \int_{-\infty}^{\infty} \phi^2(x, t) \, dx =& \frac{2}{3}, &
    \int_{-\infty}^{\infty} \phi'^2(x, t) \, dx =& \frac{2}{3} (1+\delta)^{3/2}, \label{eqn:kdv-norm}
\end{align}
which are clearly subnormalized for small enough $\delta$.

We now show that the distance between these solutions is linear in $\delta$.

\begin{lem}
    The initial distance between $\phi$ and $\phi'$ satisfies \begin{align}
        d(\phi, \phi')\vert_{t=0} = O(\delta).
    \end{align}
    \label{lem:kdv-init}
\end{lem}

\begin{proof}
    At $t=0$ write $y=x/2$ and set $f(z)=\sech^2 z$. Then
    \begin{align}
    \phi(x,0)=-\tfrac12 f(y),
    \qquad
    \phi'(x,0)=-\tfrac12(1+\delta)f(\alpha y),
    \qquad
    \alpha:=\sqrt{1+\delta}.
    \end{align}
    We have
    \begin{align}
    \|\phi-\phi'\|_{L^2}^2
    &=\int_{\mathbb R}(\phi-\phi')^2\,dx \\
    &=\frac12\int_{\mathbb R}
    \Big((1+\delta)f(\alpha y)-f(y)\Big)^2\,dy.
    \end{align}
    Using $(a+b)^2\le 2(a^2+b^2)$,
    \begin{align}
    \|\phi-\phi'\|_{L^2}^2
    &\le \frac12\int_{\mathbb R}
    \Big(\delta^2 f(y)^2
    +(1+\delta)^2\big(f(\alpha y)-f(y)\big)^2\Big)\,dy.
    \end{align}
    
    Since $f(y)=\sech^2 y$ satisfies
    \begin{align}
    \int_{\mathbb R} f(y)^2\,dy
    =\int_{\mathbb R}\sech^4 y\,dy= O(1),
    \end{align}
    the first term is $O(\delta^2)$.
    
    Since
    \begin{align}
    f'(z)=\frac{d}{dz}\sech^2 z
    =-2\sech^2 z\,\tanh z,
    \end{align}
    we have
    \begin{align}
    |f'(z)|\le 2\sech^2 z.
    \end{align}
    By the mean value theorem, for each $y$ there exists
    $\xi$ between $y$ and $\alpha y$ such that
    \begin{align}
    |f(\alpha y)-f(y)|
    =|(\alpha-1)y|\cdot |f'(\xi)|
    \le 2|\alpha-1|\,|y|\,\sech^2(\xi).
    \end{align}
    For $\delta$ small, $\alpha\in(1,2]$, and since $|\xi|\ge |y|$ and $\sech^2$ is even and
    monotone decreasing on $[0,\infty)$,
    \begin{align}
    \sech^2(\xi)\le \sech^2(|y|).
    \end{align}
    Therefore
    \begin{align}
    |f(\alpha y)-f(y)|
    \le 2|\alpha-1|\,|y|\,\sech^2(|y|),
    \end{align}
    and hence
    \begin{align}
    \int_{\mathbb R}(f(\alpha y)-f(y))^2\,dy
    \le 4|\alpha-1|^2
    \int_{\mathbb R} y^2\sech^4(|y|)\,dy.
    \end{align}
    The latter integral is finite because
    $y^2\sech^4(|y|)\sim y^2e^{-4|y|}$ as $|y|\to\infty$.
    Therefore this term is $O(|\alpha-1|^2)=O(\delta)$.
    Combining the two estimates gives
    $
    \|\phi-\phi'\|_{L^2}^2= O(\delta^2)
    $,
    and the result follows.
\end{proof}

After evolving for a time $t=O(\delta^{-1})$, the states have unit distance.

\begin{lem}
    The distance between the time-evolved states at $t=\frac{2}{\delta}$ satisfies
    \begin{align}
        d(\phi, \phi')|_{t=2/\delta} =\Omega(1).
    \end{align}
    \label{lem:kdv-final}
\end{lem}
\begin{proof}
    Let us translate the origin to be about $ct+1$, and normalize each function. This  gives
    \begin{align}
        \hat \phi =& -\sqrt{\frac{3}{8}} \sech^2((x+1)/2), &
        \hat \phi' =& -\sqrt[4]{1+\delta}\sqrt{\frac{3}{8}} \sech^2(\sqrt{1+\delta}(x-1)/2).
    \end{align}
   Consider the overlap between these states, $\abs{\braket{\hat \phi, \hat \phi'}}$.
    If we define \begin{align}
        f(x) = \frac{3\sqrt[4]{1+\delta}}{8} \sech^2( (x+1)/2) \sech^2(\sqrt{1+\delta} (x-1)/2),
    \end{align}
    then we can write the overlap as \begin{align}
        \abs{\braket{\hat \phi, \hat \phi'}} =& \int_{-\infty}^{\infty} f(x) \, dx.
    \end{align}
    If we let $I(\delta) = \int_{-\infty}^{\infty} f(x) dx$, then we have \begin{align}
        I(0) &= \int_{-\infty}^{\infty} \frac{3}{8} \sech^2\left(\frac{x+1}{2}\right) \sech^2\left(\frac{x-1}{2}\right) \\
        &= \frac{24e^2}{e^3-1} \approx 0.68,
    \end{align}
    so for $\delta \ll 1$ sufficiently small we have $I(\delta) = 1 - \Omega(1)$ by continuity.
    Now observe that $\abs{\braket{\hat \phi, \hat \phi'}} \geq 1 - O(d(\phi, \phi')^2)$, as the norms of $\phi, \phi'$ are $\Omega(1)$. Re-arranging, we obtain $O(d(\phi, \phi')^2) \geq \Omega(1)$, which implies the claim. 
\end{proof}

\begin{thm} 
    There exists a universal constant $C$ such that any quantum algorithm for simulating the KdV equation, as outlined in Problem~\ref{prob:final-state-kdv}, for time $T$ to error at most $1/C$ requires $\Omega(T^2)$ copies of the initial state.
    \label{thm:kdv-lowerbound}
\end{thm}

\begin{proof}
This follows from \Cref{lem:kdv-init,lem:kdv-final,cor:eq-evol-lower}.
\end{proof}

\section{Lower bound for simulating the incompressible Euler equations}
\label{sec:lb_final_state}
In this section, we prove lower bounds for Problems~\ref{prob:final-state-euler-tensor-prod} and~\ref{prob:final-state-euler-direct-sum}. First, we consider a smooth equilibrium that is provably unstable in the linear regime of the Euler equations and prove a lower bound for simulating the linearized Euler equations. Then we bound the error from the linearization of the Euler equations. Finally, we combine these results to prove the lower bound for the fully nonlinear equations. 

\subsection{Smooth Kelvin-Helmholtz instability} \label{sec:SKH}

Consider the incompressible Euler equations given by Eqs.~(\ref{eq:incomp-euler}) with the equilibrium profile
\begin{equation} \label{eq:equil} 
    \begin{aligned} 
    \mathbf u_0(\mathbf r) &=  \begin{bmatrix}
        U_0\sin (my) \\ 0
    \end{bmatrix}, \ m \in \mathbb{N} \\
    p_0(\mathbf r) &= \mathrm{const},
\end{aligned}
\end{equation}
on the domain $[0,2\pi) \times [0,2\pi)$ with periodic boundary conditions (i.e., $\Omega = \mathbb T^2$), where $U_0>0$ is the amplitude of the velocity field. One can think of this flow as a smoothed-out version of the Kelvin-Helmholtz instability, which is a shear-driven instability involving two disjoint fluid layers flowing at different velocities~\cite{chandrasekhar2013hydrodynamic}. This equilibrium is much easier to study for our work because the flow is smooth and the boundary conditions are periodic. Other well-known fluid instabilities, such as the Rayleigh-Taylor and conventional Kelvin-Helmholtz instabilities, have discontinuities in the equilibrium profile, which makes developing an analytical lower bound more challenging. 

A disadvantage of this equilibrium is that we do not know how to find an exact dispersion relation that describes the growth rate as a function of the perturbation wavenumber. However, we establish upper and lower bounds on the growth rate of the instability, which suffices for our application. This allows us to analytically understand the linearized equations. Later, we bound the small-time divergence between the linearized and actual time evolutions.

\subsubsection{Linear stability analysis} \label{subsec:linear-stab-analysis}
In this section we perform stability analysis on the equilibrium of Eq.~\eqref{eq:equil} and identify conditions under which the flow is unstable. This analysis is done by linearizing the Euler equations about the equilibrium and calculating the growth rate of perturbations to the equilibrium. Furthermore, we give upper and lower bounds on the growth rate of the instability, which are useful for lower bounding the copy complexity of the simulation problem.

We start with the incompressible Euler equations (i.e., Eq.~\eqref{eq:incomp-euler}). Let $\mathbf u = \mathbf{u}_0 + \epsilon \tilde{\mathbf{v}}$ and $
P= p_0 + \epsilon \tilde p$, 
where $0<\epsilon < 1$. Substituing this in Eq.~\eqref{eq:incomp-euler} and only keeping the terms that are first order in $\epsilon$, we obtain the linearized Euler equations:
\begin{equation}
    \begin{gathered} \label{eq:lin-euler}
        \frac{\partial \tilde{\mathbf{v}}}{\partial t} = -(\mathbf{u}_0 \cdot \nabla) \tilde{\mathbf{v}} - (\tilde{\mathbf{v}}\cdot \nabla)\mathbf{u}_0 - \nabla \tilde{p} \\
    \nabla \cdot \tilde{\mathbf{v}} = 0.
    \end{gathered}
\end{equation}
Any two-dimensional divergence-free vector field on a simply connected domain can be characterized by a scalar function. In fluid dynamics, this is commonly referred to as the stream function. Let $\psi(\mathbf{r},t)$ be the stream function of $\tilde{\mathbf{v}}$, whose $x$- and $y$-components satisfy
\begin{equation} \label{eq:stream-function-def}
    \begin{aligned}
    \tilde{v}_x &= -\frac{\partial \psi}{\partial y} \\
    \tilde{v}_y &= \frac{\partial \psi}{\partial x}.
\end{aligned}
\end{equation}

By looking at each component of Eq.~\eqref{eq:lin-euler}, taking the $y$-derivative of the $x$-component and the $x$-derivative of the $y$-component, and subtracting the two from each other, one can write Eq.~(\ref{eq:lin-euler}) as
\begin{equation} \label{eq:stream function-realspace}
    \frac{\partial}{\partial t} \nabla^2 \psi + U_0\sin (my) \frac{\partial}{\partial x} (\nabla^2 + m^2)\psi = 0.
\end{equation}
Since there is no spatial variation in the $x$-direction, we can seek a stream function of the form
\begin{equation} \label{eq:stream-function-ansatz}
    \psi = e^{\gamma t + ikx} \chi(y),
\end{equation}
for some $k\in \mathbb N$, which we can choose. Plugging this ansatz into Eq.~(\ref{eq:stream function-realspace}), we get the following generalized eigenvalue problem\footnote{Note that Eq.~\eqref{eq:stream function-realspace} is real-valued, but the ansatz of Eq.~\eqref{eq:stream-function-ansatz} yields Eq.~\eqref{eq:stream function-fourierspace}, which is complex-valued. This is a fairly standard technique in stability analysis, where the perturbations are assumed to be complex-valued to facilitate the calculations~\cite{chen1984introduction,drazin2004hydrodynamic,chandrasekhar2013hydrodynamic}. Once the stability analysis is complete, one can take the real part of Eq.~\eqref{eq:stream-function-ansatz} to obtain the correct physical form of the perturbation.}:
\begin{equation} \label{eq:stream function-fourierspace}
    \gamma \left (-k^2 + \frac{d^2}{dy^2} \right )\chi + ikU_0\sin(my)\left (m^2 - k^2 + \frac{d^2}{dy^2} \right)\chi = 0.
\end{equation}
We now write $\chi(y)$ as a Fourier series:
\begin{equation}
    \chi(y) = \sum_{j\in \mathbb Z} c_j e^{ijy}.
\end{equation}
Substituting this into Eq.~(\ref{eq:stream function-fourierspace}) gives the recurrence relation
\begin{equation}
    \frac{2\gamma}{kU_0} (j^2 + k^2)c_j + (k^2-m^2+(j-m)^2)c_{j-m} - (k^2 - m^2 + (j+m)^2) c_{j+m} = 0.
\end{equation}
We can also seek a similar ansatz for the perturbed pressure in Eq.~\eqref{eq:lin-euler}:
\begin{equation}
    \tilde p = e^{\gamma t + ikx} \sum_{j\in \mathbb Z} b_j e^{ijy}.
\end{equation}
The task now is to find sets of Fourier coefficients $\{b_j\}$ and $\{c_j\}$ that correspond to a smooth $\tilde p$ and $\psi$, respectively, with the additional constraint $\gamma >0$. A similar version of the following has been proven in~\cite{friedlander1997nonlinear}.

\begin{restatable}{thm}{stability} \label{thm:stability}
    Let $m > 1$ and $m^2 \neq m_1^2 + m_2^2$ where $m_1,m_2\in\mathbb Z$ and $m_1m_2 \neq 0$. Let $k\in[m-1]$ be the perturbation wavenumber. Then there exist sets of coefficients $\{b_j\}$ and $\{c_j\}$ such that
    \begin{enumerate}
        \item $\gamma > 0$ when
        \begin{equation}
            k < k_{\mathrm{cutoff}} \coloneq \sqrt{\frac{\sqrt{177}-9}{6}}m \approx 0.85m.
        \end{equation}
        \item When $k<k_{\mathrm{cutoff}}$, $\gamma$ has the bounds
        \begin{equation}
        \gamma_l < \gamma < \gamma_u,
        \end{equation}
        where
        \begin{align}
            \gamma_u &\coloneq U_0 k\sqrt{\frac{1}{2}\frac{m^2 - k^2}{m^2+k^2}}  \\
            \gamma_l &\coloneq U_0k \sqrt{\frac{1}{2} \frac{8m^4-9m^2k^2-3k^4}{(m^2+k^2)(8m^2+2k^2)}},
        \end{align}        
        \item $\psi$ and $\tilde p$ are smooth (analytic).
        \item The coefficients $\{c_j\}$ can be chosen such that
        \begin{align}
            \|\tilde{\mathbf{v}}\|_{L^2} &= \aleph \\
            \|\tilde p \|_{L^2} &= q\aleph,
        \end{align}
        where $\aleph > 0$ can be arbitrarily chosen and $q>0$.
    \end{enumerate}

\end{restatable}
The proof of this theorem can be found in Appendix~\ref{app:proof_stability}. The proof follows the one in~\cite{friedlander1997nonlinear}, but is slightly modified by generalizing it to equilibria of arbitrary amplitude and also developing upper and lower bounds on the growth rate of the instability. The proof also constructs the coefficients for $\tilde p$.

\subsection{Exponential separation of the linear solution}

In this section, we show that two trajectories, one at the equilibrium and another perturbed about it, diverge exponentially fast under linearized dynamics. We will consider the tensor-product encoding of Problem~\ref{prob:final-state-euler-tensor-prod}.

The linearized equations we consider are a modification of Eq.~\eqref{eq:lin-euler} that govern the dynamics of $\tilde{\mathbf{u}}= \mathbf{u}_0 + \epsilon \tilde{\mathbf{v}}$ and $\tilde P = p_0 + \epsilon \tilde p$:
\begin{equation} \label{eq:lin-euler-total-vel}
    \begin{gathered}
        \frac{\partial \tilde{\mathbf{u}}}{\partial t} +(\mathbf{u}_0 \cdot \nabla) \tilde{\mathbf{u}} + (\tilde{\mathbf{u}}\cdot \nabla)\mathbf{u}_0   = - \nabla \tilde P \\
    \nabla \cdot \tilde{\mathbf{u}} = 0.
    \end{gathered}
\end{equation}

\begin{thm} \label{thm:lin-diverges}
    For any $0<\epsilon<\min\{1, \frac{1}{q}\}$, with $q$ defined as in \cref{thm:stability}, there exists a pair of solutions, $\mathbf{u}_0(t),P_0(t)$ and $\tilde{\mathbf u}(t), \tilde P(t)$, to Eq.~\eqref{eq:lin-euler-total-vel} such that 
    \begin{align}
    \langle P_0(0)|\tilde P(0)\rangle &= 1-O(\epsilon^2) \\
    \langle \mathbf{u}_0(0)|\tilde{\mathbf{u}}(0)\rangle  &\geq 1-\epsilon^2,
    \end{align}
    and
    \begin{equation}
        g(t) \leq \langle \mathbf{u}_0(t) | \tilde{\mathbf{u}}(t)\rangle \leq f(t),
    \end{equation}
    where
    \begin{align}
        f(t) &\coloneq \sqrt{\frac{1-\epsilon^2}{1+\epsilon^2(e^{2\gamma_l t}-1)}} \\
        g(t) &\coloneq \sqrt{\frac{1-\epsilon^2}{1+\epsilon^2(e^{2\gamma_u t}-1)}},
    \end{align}
    with $\gamma_l$ and $\gamma_u$ defined as in \Cref{thm:stability}.
\end{thm}

\begin{proof}
We pick two states, the equilibrium $\mathbf{u}_0$ and the perturbed state $\tilde{\mathbf{u}}$, with the following profiles:
\begin{align}
    \mathbf{u}_0 &= \begin{bmatrix}
        U_0 \sin(my) \\
        0
    \end{bmatrix} \ \\
    \tilde{\mathbf{u}} &= \begin{bmatrix}
        U_{\epsilon} \sin(my) + \epsilon \tilde{v}_x \\
        \epsilon \tilde{v}_y
    \end{bmatrix} \eqcolon \mathbf{u}_{\epsilon} + \epsilon\tilde{\mathbf{v}}, 
\end{align}
    with $\tilde{\mathbf{v}}$ defined as in Eq.~\eqref{eq:stream-function-def}. From \cref{thm:stability}, we choose $\aleph = 1$, so $\|\tilde{\mathbf{v}}\|_{L^2}=1$. We let $U_0 = (2\pi^2)^{-1/2}$ to ensure $\|\mathbf u_0\|_{L^2} = 1$. 
    We determine $U_\epsilon$ using the normalization condition
\begin{align}
    1= \|\tilde{\mathbf{u}}\|_{L^2}^2(0) &= \int_{\mathbb T^2} (\mathbf{u}_\epsilon + \epsilon\tilde{\mathbf{v}})\cdot (\mathbf{u}_\epsilon + \epsilon\tilde{\mathbf{v}}) \, d^2 \mathbf r \\
    &= \|\mathbf{u}_\epsilon\|_{L^2}^2 + \epsilon^2\|\tilde{\mathbf{v}}\|_{L^2}^2 + 2\epsilon\int_{\mathbb T^2}  \mathbf{u}_\epsilon \cdot \tilde{\mathbf{v}} \, d^2 \mathbf{r} \\
    &= 2\pi^2U_\epsilon^2 + \epsilon^2 + 2\epsilon\int_{\mathbb T^2}  \mathbf{u}_\epsilon \cdot \tilde{\mathbf{v}} \, d^2 \mathbf{r}.
\end{align}
Evaluating the integral gives
\begin{align}
    \int_{\mathbb T^2} \mathbf{u}_\epsilon \cdot \tilde{\mathbf{v}} \, d^2 \mathbf{r} &= -\sum_{j \in \mathbb Z} U_\epsilon \int_{\mathbb T^2} ijc_j \sin(my) e^{\gamma t + ikx + ijy} \, d^2 \mathbf r \\
    &= 0
\end{align}
using Eq.~\eqref{eq:stream-function-def}. Therefore
\begin{equation}
    U_\epsilon = \sqrt{\frac{1-\epsilon^2}{2\pi^2}}.
\end{equation}
Furthermore, we have
\begin{align}
    \langle \mathbf u_0(0)|\tilde{\mathbf u}(0)\rangle &= \int_{\mathbb T^2} \mathbf{u}_0 \cdot \tilde{\mathbf{u}}(0) \, d^2 \mathbf r \\
    &= \sqrt{1-\epsilon^2} \\
    &\geq 1-\epsilon^2.
\end{align}
For the pressure, we construct $P_0$ for the equilibrium and $\tilde P$ for the perturbation. Since the equilibrium requires constant pressure and we want $\|P_0\|_{L^2} = 1$, we set
\begin{equation}
    P_0 =  \frac{1}{2\pi}.
\end{equation}
For the perturbed pressure, we decompose it into the gauge $\bar p$ and the zero-mean part:
\begin{equation}
    \tilde P = \bar p + \epsilon\tilde p.
\end{equation}
From \cref{thm:stability}, if $\aleph = 1$, $\|\tilde p\|_{L^2} = q$. We then calculate the gauge as follows:
\begin{align}
    1 = \|\tilde P\|^2_{L^2} &= \int_{\mathbb T^2} \bar p^2 + 2\epsilon \bar p \tilde p + \epsilon^2 \tilde p^2 \, d^2 \mathbf r \\
    &= (2\pi)^2 \bar p + q^2\epsilon^2 .
\end{align}
So
\begin{equation}
    \bar p = \frac{\sqrt{1-q^2\epsilon^2}}{2\pi},
\end{equation}
which is valid for $0<\epsilon< \min\left \{1, \frac{1}{q(1)}\right \}$. It is then easy to verify that
\begin{equation}
    \langle P_0(0)|\tilde P(0)\rangle = 1-O(\epsilon^2).
\end{equation}

Now we give bounds on
\begin{equation}
    \langle \mathbf{u}_0(t)|\tilde{\mathbf{u}}(t)\rangle = \frac{1}{\|\tilde{\mathbf{u}}\|_{L^2}(t)}\int_{\mathbb T^2} \mathbf u_0\cdot \tilde{\mathbf{u}}(t) \, d^2 \mathbf r.
\end{equation}
Note that
\begin{align}
    \|\tilde{\mathbf{u}}\|_{L^2}(t) &= \sqrt{2\pi^2U_\epsilon^2 + \epsilon^2 e^{2\gamma t}} \\
    &= \sqrt{1+\epsilon^2(e^{2\gamma t}-1)}.
\end{align}
The following bounds apply:
\begin{equation}
   \sqrt{1+\epsilon^2(e^{2\gamma_l t}-1)} \leq \|\tilde{\mathbf{u}}\|_{L^2}(t) \leq \sqrt{1+\epsilon^2(e^{2\gamma_u t}-1)}.
\end{equation}
Computing the integral gives
\begin{align}
    \langle \mathbf{u}_0|\tilde{\mathbf{u}}\rangle &= \frac{1}{\sqrt{1+\epsilon^2(e^{2\gamma t}-1)}} \int_{\mathbb T^2} \mathbf{u}_0\cdot\tilde{\mathbf{u}}(t) \, d^2 \mathbf r \\
    &= \sqrt{\frac{1-\epsilon^2}{1+\epsilon^2(e^{2\gamma t}-1)}}.
\end{align}
Using \cref{thm:stability}, we find the following bounds for the inner product:
\begin{equation}
    \sqrt{\frac{1-\epsilon^2}{1+\epsilon^2(e^{2\gamma_u t}-1)}} \leq \langle \mathbf{u}_0|\tilde{\mathbf{u}}\rangle \leq \sqrt{\frac{1-\epsilon^2}{1+\epsilon^2(e^{2\gamma_l t}-1)}}.
\end{equation}
This completes the proof.
\end{proof}

\begin{corr}
    Simulating the evolution of the linearized Euler equations, Eq.~\eqref{eq:lin-euler-total-vel}, for time $T$ requires $\exp(\Omega(T))$ copies of the initial state.
    \label{cor:lin-evol}
\end{corr}
\begin{proof}
Setting $t = \frac{1}{\gamma_u}\log(1/\epsilon)$ gives
\begin{align}
    \langle \mathbf{u}_0|\tilde{\mathbf{u}}\rangle &\leq \sqrt{\frac{1-\epsilon^2}{1+\epsilon^2(e^{2\gamma_l t}-1)}} \\
    &=\sqrt{\frac{1-\epsilon^2}{1-\epsilon^2 + \epsilon^{2\left (1 - \frac{\gamma_l}{\gamma_u} \right )}}} \\
    &= \left (1 + \frac{\epsilon^{2\left (1 - \frac{\gamma_l}{\gamma_u} \right )}}{1-\epsilon^2} \right )^{-1/2} \\
    &\leq 1- \frac{1}{2} \frac{\epsilon^{2\left (1 - \frac{\gamma_l}{\gamma_u} \right )}}{1-\epsilon^2} + \frac{3}{8}\frac{\epsilon^{4\left (1 - \frac{\gamma_l}{\gamma_u} \right )}}{(1-\epsilon^2)^2} \\
    &= 1-O\left ( \epsilon^{2 \left ( 1 - \frac{\gamma_l}{\gamma_u} \right )} \right ).
\end{align}
Since $0<\gamma_l/\gamma_u <1 $ from \cref{thm:stability}, Corollary~\ref{corr:state-disc-expon} shows that $e^{\Omega(t)}$ copies are needed.
\end{proof}
\begin{remark}
    A cleaner proof can be obtained by setting $t = \gamma_l^{-1}\log(1/\epsilon)$ and taking the limit $\epsilon\rightarrow 0$. However, the above proof maintains consistency with later sections, where we constrain the time interval to $t\in [0,\gamma^{-1}_u \log(1/\epsilon)]$.
\end{remark}

\subsection{Bounding the error from linearization}
In the previous section we demonstrated that the linearized solution deviates from the equilibrium exponentially quickly. However, this does not tell us whether the nonlinear solution behaves similarly. In this section, we establish a bound on the error from linearization, so that we can probe the behavior of the nonlinear solution. First, we derive a PDE for the time evolution of the error.

\begin{lem} \label{lem:error-eqn}
    Consider the incompressible Euler equations in Eq.~(\ref{eq:incomp-euler}). Suppose we have an equilibrium given by $\mathbf{u}_0(\mathbf r)$ and $p_0$ satisfying
    \begin{gather}
        (\mathbf{u}_0 \cdot \nabla)\mathbf{u}_0 = - \nabla p_0.
    \end{gather}
     Let $\mathbf{u}$ and $P$ be the fully nonlinear solutions satisfying Eq.~\eqref{eq:incomp-euler} and $\tilde{\mathbf{u}}=\mathbf{u}_0 + \epsilon \tilde{\mathbf{v}}$ and $\tilde P = p_0 + \epsilon\tilde p$ be the linearized solutions satisfying Eq.~\eqref{eq:lin-euler-total-vel}. Define the velocity error vector field as $\boldsymbol{\eta} \coloneq \mathbf{u} - \tilde{\mathbf u}$ and the pressure error as $\eta_p \coloneq P - \tilde P$. Then $\boldsymbol{\eta}$ satisfies the following PDE:
    \begin{gather}
         \frac{\partial \boldsymbol{\eta}}{\partial t} = -\nabla \left(\tilde{\mathbf {u}} \cdot \boldsymbol{\eta} + \frac{1}{2} \boldsymbol{\eta}\cdot \boldsymbol{\eta} + \eta_p \right) + \boldsymbol{\eta} \times (\nabla \times \tilde{\mathbf{u}} ) + \tilde{\mathbf{u}}\times(\nabla \times \boldsymbol{\eta}) + \boldsymbol{\eta}\times (\nabla \times \boldsymbol{\eta}) - \epsilon^2(\tilde{\mathbf{v}}\cdot \nabla) \tilde{\mathbf{v}} \\
        \nabla \cdot \boldsymbol{\eta} = 0.
    \end{gather}
\end{lem}
\begin{proof}
Let $\mathbf{u} = \mathbf{u}_0 + \epsilon\mathbf{v}$ and $P = p_0 + \epsilon p$ be the nonlinear solutions and let $\tilde{\mathbf{u}} = \mathbf{u}_0 + \epsilon \tilde{\mathbf{v}}$ and $\tilde P = p_0 +\tilde p$ be the linearized solutions. Plugging this into Eq.~(\ref{eq:incomp-euler}), we obtain
\begin{gather}
    \frac{\partial \mathbf{v}}{\partial t} = -(\mathbf{u}_0 \cdot \nabla) \mathbf{v} - (\mathbf{v}\cdot \nabla)\mathbf{u}_0 - \epsilon(\mathbf{v}\cdot \nabla)\mathbf{v} -  \nabla p \\
    \nabla \cdot \mathbf{v} = 0,
\end{gather}
and pressure is determined by the Laplace equation
\begin{equation}
    \nabla^2 p = -\nabla \cdot \left[ (\mathbf u_0 \cdot \nabla)\mathbf v + (\mathbf v \cdot \nabla ) \mathbf u_0 + \epsilon (\mathbf v \cdot \nabla)\mathbf v \right ].
\end{equation}
We approximate this by the linear PDE
\begin{gather}
    \frac{\partial \tilde{\mathbf{v}}}{\partial t} = -(\mathbf{u}_0 \cdot \nabla) \tilde{\mathbf{v}} - (\tilde{\mathbf{v}}\cdot \nabla)\mathbf{u}_0 -  \nabla \tilde{p}  \\
    \nabla \cdot \tilde{\mathbf{v}} = 0,
\end{gather}
which is precisely Eq.~(\ref{eq:lin-euler}), with pressure being determined by the following Laplace equation
\begin{equation}
    \nabla^2 \tilde p  = -\nabla \cdot \left[ (\mathbf u_0 \cdot \nabla)\mathbf v + (\mathbf v \cdot \nabla ) \mathbf u_0  \right ].
\end{equation}
Letting $\boldsymbol{\eta}(\mathbf{r},t) \coloneq \mathbf{u}(\mathbf{r},t) - \tilde{\mathbf{u}}(\mathbf{r},t)$ be the error vector, we see that
\begin{equation}
    \boldsymbol{\eta} = \epsilon\mathbf{v} - \epsilon\tilde{\mathbf{v}}.
\end{equation}
The equation for the time evolution of the error, then, is
\begin{gather}
    \frac{\partial \boldsymbol{\eta}}{\partial t} = -(\mathbf{u}_0 \cdot \nabla)\boldsymbol{\eta} - (\boldsymbol{\eta} \cdot \nabla) \mathbf{u}_0 - \epsilon^2(\mathbf{v}\cdot \nabla) \mathbf{v} - \nabla \eta_p \\
    \nabla \cdot \boldsymbol{\eta} = 0.
\end{gather}
with $\eta_p = P-\tilde P$ satisfying
\begin{equation}
    \nabla^2 \eta_p = -\epsilon^2 \nabla \cdot \left [(\mathbf v \cdot \nabla ) \mathbf v \right].
\end{equation}
The initial condition for the PDE is $\boldsymbol{\eta}(\mathbf{r},0) = 0$ and $\eta_p(\mathbf r,0) = 0$. Using the fact that $\mathbf v = \frac{1}{\epsilon}\boldsymbol\eta + \tilde{ \mathbf{v}}$, we obtain
\begin{gather}
    \frac{\partial \boldsymbol{\eta}}{\partial t} = -(\tilde{\mathbf{u}} \cdot \nabla)\boldsymbol{\eta} - (\boldsymbol{\eta} \cdot \nabla) \tilde{\mathbf{u}} - (\boldsymbol{\eta}\cdot \nabla) \boldsymbol{\eta} -  \epsilon^2(\tilde{\mathbf{v}}\cdot \nabla) \tilde{\mathbf{v}} - \nabla \eta_p \\
    \nabla \cdot \boldsymbol{\eta} = 0.
\end{gather}
The first two terms on the right-hand side are linear, the second term is nonlinear, and the final term is a forcing term. Using the vector calculus identity
\begin{equation}
    \nabla(\mathbf{A}\cdot \mathbf{B}) = (\mathbf{A}\cdot \nabla ) \mathbf B + (\mathbf B \cdot \nabla ) \mathbf A + \mathbf{A}\times (\nabla \times \mathbf B) + \mathbf{B}\times (\nabla \times \mathbf A),
\end{equation}
we can rewrite the PDE as
\begin{gather}
    \frac{\partial \boldsymbol{\eta}}{\partial t} = -\nabla \left(\tilde{\mathbf {u}} \cdot \boldsymbol{\eta} + \frac{1}{2} \boldsymbol{\eta}\cdot \boldsymbol{\eta} + \eta_p \right) + \boldsymbol{\eta} \times (\nabla \times \tilde{\mathbf{u}} ) + \tilde{\mathbf{u}}\times(\nabla \times \boldsymbol{\eta}) + \boldsymbol{\eta}\times (\nabla \times \boldsymbol{\eta}) - \epsilon^2(\tilde{\mathbf{v}}\cdot \nabla) \tilde{\mathbf{v}} \\
    \nabla \cdot \boldsymbol{\eta} = 0,
\end{gather}
which completes the proof.
\end{proof}

From \cref{lem:error-eqn}, we can prove the following. Let $\hat{\mathbf k}$ denote the unit vector normal to $\mathbb T^2$.

\begin{restatable}{lem}{errorbound} \label{lem:error_bound}
Let $\boldsymbol{\eta}$ be defined as in \cref{lem:error-eqn}. Then
    \begin{equation}
        \frac{d\|\boldsymbol{\eta}\|_{L^2}^2}{dt} \leq 2\|\boldsymbol{\eta}\|_{L^2} \left [ \|\tilde{\mathbf{u}} \|_{L^\infty} \|\omega_{\boldsymbol{\eta}}\|_{L^2} + \epsilon^2 \|(\tilde{\mathbf{v}}\cdot \nabla)\tilde{\mathbf{v}}\|_{L^2}  \right ],
    \end{equation}
    where
    \begin{equation}
        \omega_{\boldsymbol{\eta}} \coloneq (\nabla \times \boldsymbol{\eta})\cdot \hat{\mathbf{k}}
    \end{equation}
    is the vorticity of the error.
\end{restatable}
\noindent Next we give a differential inequality for $\|\omega_{\boldsymbol{\eta}}\|_{L^2}$.
\begin{restatable}{lem}{vorticity} \label{lem:vort_error_bound}
    Let $\omega_{\boldsymbol{\eta}}$ be defined as in \cref{lem:error_bound} and $\boldsymbol{\eta}$ defined as in \cref{lem:error-eqn}. Then
    \begin{equation}
         \frac{d \|\omega_{\boldsymbol{\eta}}\|_{L^2}^2}{dt} \leq 2\|\omega_{\boldsymbol{\eta}}\|_{L^2} \left [ \|\nabla \omega_{\tilde{\mathbf{u}}}\|_{L^\infty}\|\boldsymbol{\eta}\|_{L^2} + \epsilon^2 \|\tilde{\mathbf{v}}\cdot \nabla \omega_{\tilde{\mathbf{v}}}\|_{L^2}   \right],
    \end{equation}
    where
    \begin{align}
        \omega_{\tilde{\mathbf u}} &\coloneq (\nabla \times \tilde{\mathbf{u}})\cdot \hat{\mathbf k} \\
        \omega_{\tilde{\mathbf v}} &\coloneq (\nabla \times \tilde{\mathbf{v}})\cdot \hat{\mathbf k}.
    \end{align}
\end{restatable}
\noindent The proofs for \cref{lem:error_bound} and \cref{lem:vort_error_bound} can be found in Appendix~\ref{app:error_bound}. Combining these results gives us a bound on $\|\boldsymbol{\eta}\|_{L^2}$:
\begin{restatable}{thm}{errorboundvalue}\label{thm:error_bound_value}
Let $\kappa,\alpha, \beta, b_1, b_2$ be positive constants with the following satisfied: 
\begin{align}
    \alpha \beta &> 2\gamma_u \\
     \kappa &= \frac{1}{\alpha \beta - 2\gamma_u} \left (b_1 + \frac{\alpha b_2}{2\beta}\right ),
\end{align}
with $\gamma_u$ being defined in \cref{thm:stability}. Consider the time interval $t\in[0,\gamma_u^{-1}\log(1/\epsilon)]$. Then, for $0<\epsilon<1$, we have
    \begin{equation}
        \|\boldsymbol{\eta}\|_{L^2} \leq \epsilon^2 \kappa \left ( e^{\alpha \beta t} - 1 \right ).
    \end{equation}
\end{restatable}
\noindent The proof can be found in Appendix~\ref{app:error_bound_value}.

We have successfully bounded the error from linearization, $\|\boldsymbol{\eta}\|_{L^2}$. However, this bound does not directly apply to the normalized states. Taking into account the normalization of $\tilde{\mathbf{u}}$, we establish the following.

\begin{lem} \label{lem:normalized-error}
    Let
    \begin{equation}
        \hat{\boldsymbol{\eta}} \coloneq \frac{\mathbf{u}}{\|\mathbf{u}\|_{L^2}} - \frac{\tilde{\mathbf{u}}}{\|\tilde{\mathbf{u}}\|_{L^2}}.
    \end{equation}
    Then, for $0<\epsilon < 1$ and $t\in[0,\gamma_u^{-1}\log(1/\epsilon)]$,
    \begin{equation}
        \|\hat{\boldsymbol{\eta}} \|_{L^2} \leq h(t)  
        \coloneq \frac{2\epsilon^2 \kappa \left ( e^{\alpha \beta t} - 1 \right )}{\sqrt{1+\epsilon^2(e^{2\gamma_l t} -1 )}},
    \end{equation}
    where $ \kappa,\alpha, \beta$ are defined as in \cref{thm:error_bound_value} and $\gamma_l$ and $\gamma_u$ are defined as in \cref{thm:stability}.
\end{lem}
\begin{proof}
    First, note that $\|\mathbf{u}\| - \|\tilde{\mathbf{u}}\| \leq \|\mathbf{u} - \tilde{\mathbf{u}}\|$ by the triangle inequality. Using this, we get
    \begin{align}
        \| \hat{\boldsymbol{\eta}}\| &= \left \| \frac{\mathbf{u}}{\|\mathbf{u}\| } - \frac{\tilde{\mathbf{u}}}{\|\tilde{\mathbf{u}}\|} \right \| \\
        &\leq \left \| \frac{\tilde{\mathbf{u}}}{\|\tilde{\mathbf{u}}\| } - \frac{\mathbf{u}}{\|\tilde{\mathbf{u}}\|} \right \| + \left \| \frac{\mathbf{u}}{\|\tilde{\mathbf{u}}\|}  - \frac{\mathbf{u}}{\|\mathbf{u}\| }   \right \| \\
        &= \frac{\|\mathbf{u} - \tilde{\mathbf{u}} \|}{\|\tilde{\mathbf{u}}\|} + \|\mathbf{u}\| \left | \frac{1}{\|\tilde{\mathbf{u}}\|} - \frac{1}{\|\mathbf{u}\|} \right | \\
        &= \frac{\|\mathbf{u} - \tilde{\mathbf{u}} \|}{\|\tilde{\mathbf{u}}\|} + \frac{|\|\mathbf{u}\| - \|\tilde{\mathbf{u}}\||}{\|\tilde{\mathbf{u}}\|} \\
        &\leq \frac{2\|\mathbf{u} - \tilde{\mathbf{u}} \|}{\|\tilde{\mathbf{u}}\|} \\
        &= \frac{2\|\boldsymbol{\eta}\|}{\|\tilde{\mathbf{u}}\|}.
    \end{align}
    Using \cref{thm:lin-diverges} and \cref{thm:error_bound_value} completes the proof.
\end{proof}

\subsection{Exponential separation of the nonlinear solution}

Finally, we prove lower bounds for simulating the incompressible Euler equations. First we will prove the lower bound for the tensor-product encoding.

\begin{thm} \label{thm:euler-lower-bound}
    There exists a universal constant $C$ such that any quantum algorithm for simulating the incompressible Euler equations, as outlined in Problem~\ref{prob:final-state-euler-tensor-prod}, for time $T$ to error at most $e^{-CT}$ requires $e^{\Omega(T)}$ copies of the initial state.
\end{thm}

\begin{proof}
    For the tensor-product encoding of Problem~\ref{prob:final-state-euler-tensor-prod}, we pick two states:
    \begin{align}
        |\psi_0\rangle &= |\mathbf u_0 \rangle \otimes |p_0 \rangle \\
        |\psi\rangle &= |\mathbf u \rangle \otimes |P \rangle,
    \end{align}
    where $\mathbf u_0$ and $P_0$ correspond to the equilibrium of Eq.~\eqref{eq:equil}, and $\mathbf u$ and $P$ are perturbations about the equilibrium. Note that
    \begin{equation}
        \langle \psi_0 |\psi\rangle \leq \langle \mathbf u_0|\mathbf u \rangle,
    \end{equation}
    so it suffices to consider the overlap of the velocity components only.
    
    Using the exponential deviation of the linear solution as well as the bound on the error from the linearization, we show exponential deviation of the nonlinear solution. From \cref{thm:lin-diverges} and \cref{lem:normalized-error} we have the bounds
    \begin{gather}
         g(t) \leq \langle \mathbf{u}_0|\tilde{\mathbf{u}}\rangle \leq f(t) 
         \\
         \|  \hat{\boldsymbol{\eta}}  \| \leq  h(t)
    \end{gather}
    with the latter bound holding for $t\in[0,\gamma_u^{-1}\log(1/\epsilon)]$. Thus, we restrict time to $t\leq \gamma_u^{-1}\log(1/\epsilon)$ for the rest of the analysis.
    
    Now we upper bound $\langle \mathbf{u}_0|\mathbf{u}\rangle$ by lower bounding $\||\mathbf{u}\rangle - |\mathbf{u}_0 \rangle \|$. We have
    \begin{align}
        \| |\tilde{\mathbf{u}}\rangle- |\mathbf{u}_0 \rangle  \| &\leq \| |\mathbf{u}\rangle - |\tilde{\mathbf{u}}\rangle \| +  \||\mathbf{u}\rangle - |\mathbf{u}_0 \rangle \| \\
        &= \| \hat{\boldsymbol{\eta}} \| + \||\mathbf{u}\rangle - |\mathbf{u}_0 \rangle \|,
    \end{align}
    which gives
    \begin{align}
        \||\mathbf{u}\rangle - |\mathbf{u}_0 \rangle \| &\geq \| |\tilde{\mathbf{u}}\rangle- |\mathbf{u}_0 \rangle  \| - \| \hat{\boldsymbol{\eta}} \|  \\
        &\geq \| |\tilde{\mathbf{u}}\rangle- |\mathbf{u}_0 \rangle  \| - h(t).
    \end{align}
    Expanding the norms, we find
    \begin{align}
        \sqrt{2 - 2\langle \mathbf{u}_0|\mathbf{u}\rangle} &\geq \sqrt{2 - 2 \langle\mathbf{u_0}|\tilde{\mathbf{u}}\rangle} - h(t) \\
        \Rightarrow \langle\mathbf{u_0}|{\mathbf{u}}\rangle & \leq \langle \mathbf{u}_0|\tilde{\mathbf{u}}\rangle + h(t)\sqrt{2-2\langle\mathbf{u}_0|\tilde{\mathbf{u}}\rangle} - \frac{h(t)^2}{2} \\
        &\leq f(t) + h(t)\sqrt{2-2g(t)} \\
        &\leq f(t) + \sqrt{2}\epsilon h(t)e^{\gamma_u t} \eqcolon \tilde H(t), \label{eq:inner-product-bound}
    \end{align}
    where in the last line we have used the fact that
    \begin{align}
        \sqrt{1-g(t)} &\leq \sqrt{1- \sqrt{1-\epsilon^2}\left ( 1-\frac{\epsilon^2}{2}\left (e^{2\gamma_u t} -1 \right) \right )} \\
        &\leq \sqrt{\epsilon^2 + \frac{\epsilon^2}{2}\left (e^{2\gamma_u t} - 1\right ) } \\
        &\leq \epsilon e^{\gamma_u t}
    \end{align}
    for $0< \epsilon < 1$. There is competition between the two terms in Eq.~\eqref{eq:inner-product-bound}, as the first term is monotonically decreasing and bounded between $0$ and $1$, and the second is monotonically increasing. Ideally, the first term dominates at the beginning before the second term dominates, meaning that the bound reaches a local minimum. We now find the location and value of the local minimum.
     
    We have the simple bounds 
    \begin{align}
        f(t) &\leq \sqrt{\frac{1}{1+\epsilon^2(e^{2\gamma_l t}-1)}} \\
        &\leq 1-\frac{\epsilon^2}{2}(e^{2\gamma_l t}-1) + \frac{3\epsilon^4}{8}(e^{2\gamma_l t}-1)^2 \\
        &\leq 1-\frac{\epsilon^2}{2}(e^{2\gamma_l t}-1) + \frac{3\epsilon^4}{8}e^{4\gamma_l t}
    \end{align}
    and
    \begin{align}
        h(t) &\leq 2\epsilon^2 \tilde\kappa (e^{\alpha\beta t}-1) \\
        &\leq 2\epsilon^2 \tilde\kappa e^{\alpha\beta t},
    \end{align}
    which gives
    \begin{equation}
        \langle \mathbf{u}_0|\mathbf{u}\rangle \leq 1-\frac{\epsilon^2}{2} (e^{2\gamma_l t} - 1) + \frac{3\epsilon^4}{8}e^{4\gamma_l t} + 2\sqrt{2}\epsilon^3\kappa e^{(\alpha\beta+\gamma_u) t} \eqcolon H(t).
    \end{equation}
    We compute the local minimum as follows. Let $x = e^{2\gamma_l t}$. We solve
    \begin{align}
        0 &= H'(t) \\
        &= -\epsilon^2 \gamma_l x + \frac{3}{2}\epsilon^4\gamma_l x^2 + 4\sqrt{2}\epsilon^3 \kappa K \gamma_l  x^K \\
        &= -1 + \frac{3}{2}\epsilon^2 x + 4\sqrt{2} \epsilon  \kappa K x^{K-1}, \label{eq:min_of_bound}
    \end{align}
    where 
    \begin{equation} \label{eq:K}
        K\coloneq\frac{\alpha\beta + \gamma_u}{2\gamma_l}.
    \end{equation}
    From \cref{thm:error_bound_value}, we have $K> 3/2$, so we can compute two upper bounds on Eq.~\eqref{eq:min_of_bound} by ignoring either term, obtaining $x \leq 2/3\epsilon^2$ and $x \leq (4\sqrt{2}\epsilon \tilde \kappa K)^{-1/(K-1)}$, with the latter being larger. This corresponds to \begin{align}
        x&= (4\sqrt{2}\epsilon  \kappa K)^{-1/(K-1)},&t^*&=\frac{1}{2\gamma_l(K-1)}\log\left(\frac{1}{4\sqrt{2}\epsilon  \kappa K}\right).
    \end{align}
    Note that the time interval we are considering has an upper limit of $t_f=\gamma_u^{-1}\log(1/\epsilon)$. We need to ensure that $t^* \leq t_f$. We see that
    \begin{align}
        2\gamma_l(K-1) = \alpha\beta + \gamma_u - 2\gamma_l > \gamma_u.
    \end{align}
    This shows that
    \begin{equation}
        t^* = \frac{1}{2\gamma_l(K-1)}\log\left (\frac{1}{4\sqrt{2}\kappa K \epsilon} \right ) < \frac{1}{\gamma_u}\log\left (\frac{1}{\epsilon} \right ) = t_f
    \end{equation}
    for
    \begin{equation}
        \epsilon < (4\sqrt{2}\kappa K)^{\frac{\gamma_u}{2\gamma_l(K-1) - \gamma_u}}.
    \end{equation}
    
    Let us now compute 
    \begin{align}
        H(t^*) &= 1 - \frac{\epsilon^2}{2}(x-1) + \frac{3\epsilon^4}{8}x^2 + 2\sqrt{2}\epsilon^3 \kappa x^{K} \\
        &= 1 +\frac{\epsilon^2}{2} - \frac{1}{2}\left (4\sqrt{2}\kappa K \right)^{-1/(K-1)}\epsilon^{2-1/(K-1)} + 2\sqrt{2} \kappa \left ( 4\sqrt{2}  \kappa K \right )^{-K/(K-1)}\epsilon^{3-K/(K-1)}\\
        &+ \frac{3}{8}(4\sqrt{2}\kappa K)^{-2/(K-1)}\epsilon^{2 \left ( 2 - 1/(K-1) \right )}\\
        &= 1 - (\kappa K)^{-1/(K-1)} \left [\frac{1}{2}(4\sqrt{2})^{-1/(K-1)} -  2\sqrt{2}(4\sqrt{2})^{-K/(K-1)}K^{-1} \right ] \epsilon^{2 - 1/(K-1)} \\
        &+ O\left(\epsilon^{2(2-1/(K-1))} \right )\\
        &= 1-O\left (\epsilon^{2-1/(K-1)} \right ). \label{eq:asymptotics_euler}
    \end{align}
    Also at $t=0$ we have
    \begin{align}
        \langle \mathbf{u}_0(0) | \mathbf{u}(0)\rangle &= \langle \mathbf{u}_0(0)|\tilde{\mathbf{u}}(0)\rangle \\
        &= \sqrt{1-\epsilon^2} \\
        &\geq 1-\epsilon^2.
    \end{align} 
    Thus, by Corollary~\ref{corr:state-disc-expon}, we need $e^{\Omega(T)}$ copies of the initial condition. 
    
    This lower bound applies to exact simulation, but we can relax this constraint slightly. \Cref{lem:state-dist-comp} shows that $\Omega(1/\epsilon)$ copies are required to implement evolution for time $t$ with error $o(\epsilon^{1-1/(2K-2)})$, which combined with $t= O(\log \epsilon^{-1})$ shows that $k=\exp(\Omega(t))$ copies are required to achieve error $\exp(-Ct)$ for some universal constant $C$.
\end{proof}
Note that \cref{thm:euler-lower-bound} applies to any pair of $m$ and $k$ that results in an instability. We numerically confirm the claim of \cref{thm:euler-lower-bound} in Figure~\ref{fig:numerics} for $m = 2$ and $k=1$.

Similarly, we obtain a lower bound for the direct-sum encoding of Problem~\ref{prob:final-state-euler-direct-sum}.
\begin{thm} \label{thm:euler-lower-bound-direct-sum}
    There exists a universal constant $C$ such that any quantum algorithm for simulating the incompressible Euler equations, as outlined in Problem~\ref{prob:final-state-euler-direct-sum}, for time $T$ to error at most $e^{-CT}$ requires $e^{\Omega(T)}$ copies of the initial state.
\end{thm}
\begin{proof}
    We pick the equilibrium state
    \begin{equation}
        \mathbf w_0 (0) = \begin{bmatrix}
            \frac{1}{\sqrt{2\pi^2}}\sin(my) \\
            0 \\
            0
        \end{bmatrix},
    \end{equation}
    where we have set $\bar p = p_0 = 0$. From \cref{thm:stability}, we pick $\aleph$ such that $\|\tilde{\mathbf v}\|_{L^2}^2+\|\tilde p\|_{L^2}^2 = 1$. Then, for the perturbed state, we choose
    \begin{equation}
        \mathbf w(0) = \begin{bmatrix}
            \sqrt{\frac{1-\epsilon^2}{2\pi^2}}\sin(my) + \epsilon \tilde v_x \\
            \epsilon \tilde v_y \\
            \epsilon \tilde p
        \end{bmatrix},
    \end{equation}
    where $0< \epsilon < 1$. It is easy to check that $\|\mathbf w_0(0)\|_{L^2} = \|\mathbf w(0)\|_{L^2}=1$. Since the pressure gauge (Eq.~\eqref{eq:pressure_gauge}) is zero, the inner product of the two states only takes into account the overlap of the velocity components. The rest of the proof then follows from Theorem~\ref{thm:lin-diverges}, Lemma~\ref{lem:normalized-error}, and Theorem~\ref{thm:euler-lower-bound}.
\end{proof}

\begin{figure}[ht!] 
    \centering
    \begin{subfigure}[t]{0.48\textwidth}
        \centering
        \includegraphics[width=\textwidth]{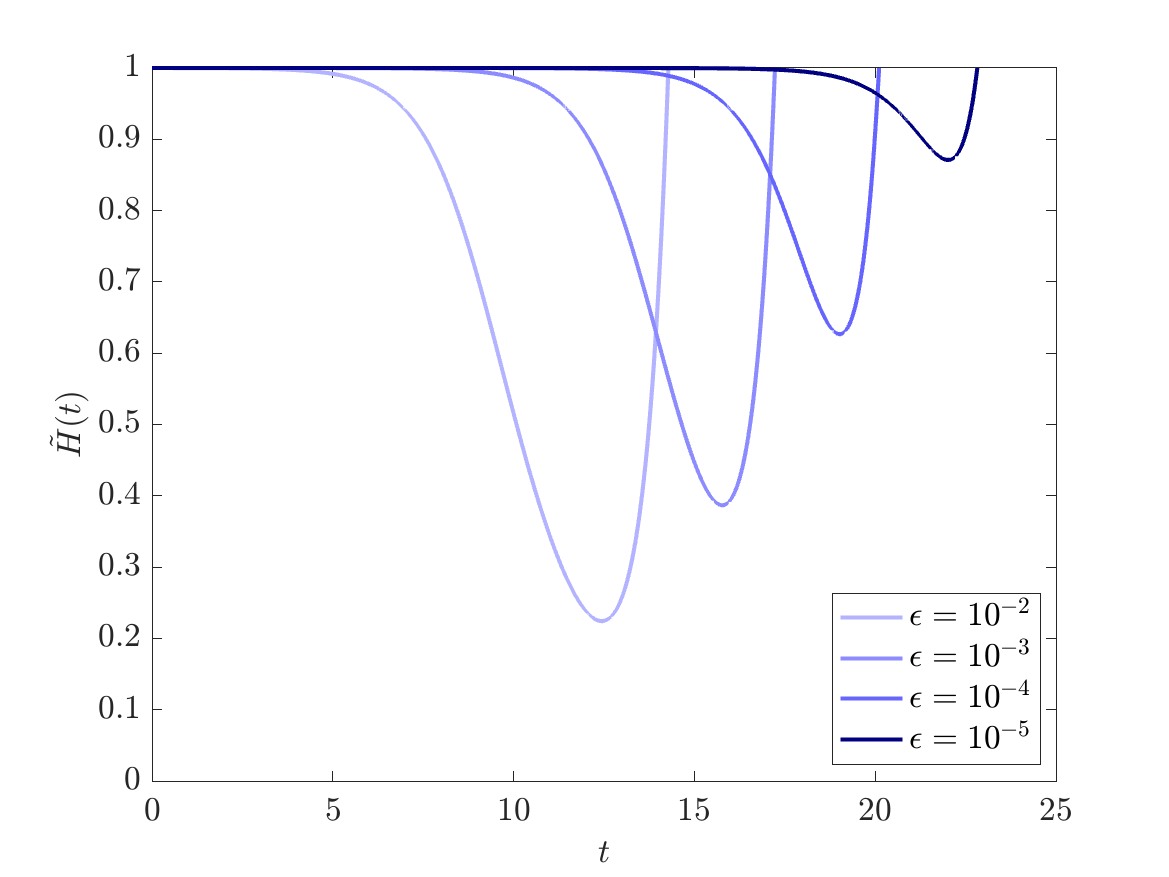}
        \caption{ $K=2$.}
        \label{fig:min}
    \end{subfigure}
    \begin{subfigure}[t]{0.48\textwidth}
        \centering
        \includegraphics[width=0.905\textwidth]{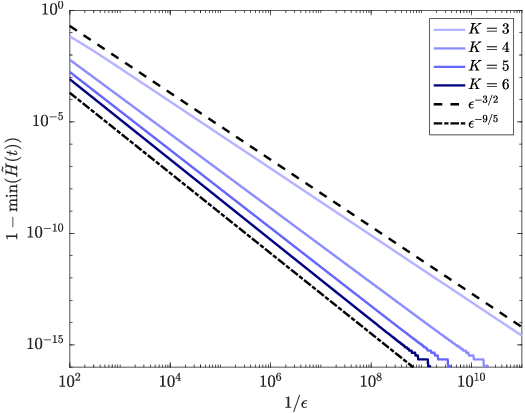}
        \caption{}
        \label{fig:eps-scaling}
    \end{subfigure}
    \caption{Numerical confirmation of the proof of \cref{thm:euler-lower-bound}. $\tilde H(t)$ denotes the bound on the inner product of the equilibrium with the nonlinear solution, as given in Eq.~\eqref{eq:inner-product-bound}. Figure~\ref{fig:min} plots $\tilde H(t)$ as a function of simulation time $t$, verifying the existence of the local minimum. Figure~\ref{fig:eps-scaling} plots the maximum value of $1-\tilde H(t)$ as a function of $1/\epsilon$, showing that the asymptotic scalings match Eq.~\eqref{eq:asymptotics_euler}. In both figures, $m=2$, $k=1$, and we arbitrarily set $\kappa = 10^{-6}$.}
    \label{fig:numerics} 
\end{figure}

\section{Lower bound for history state preparation} \label{sec:history-state}

The lower bounds in the previous section apply for preparing the solution at the final time. Instead of outputting the final-time state, some quantum algorithms prepare a history state, which contains the solution as a superposition over timesteps. In this section we show that a lower bound for final state preparation implies a lower bound for history state preparation, as one can post-select on the history state to obtain the final-time state.\footnote{Note that there is also a reduction in the other direction: an algorithm for creating the state at a desired final time can be used to produce a history state. We can do this by preparing the ancilla state $\sum_{j=1}^{m}|jh\rangle/\sqrt{m}$ using $O(\log_2m)$ gates, and then applying the final state preparation algorithm with the evolution time conditioned on the value $jh$ in this ancilla, resulting in the history state. Therefore the complexity of history state preparation is at most the complexity of final state preparation, with no overhead.} We use the fact that the KdV and Euler equations are norm-preserving to prove the lower bound.

We define the history state as follows. Let $|\psi(t)\rangle$ be the solution, as described in \Cref{sec:probdef}, to a differential equation at time $t$, where $t\in[0,T]$ with $T$ being the final simulation time. We discretize the interval $[0,T]$ into $m\in \mathbb N$ intervals of equal length $h\coloneq T/m$. The history state takes the form
\begin{equation}
    |\Psi\rangle \propto \sum_{j=1}^{m} |j\rangle |\psi(jh)\rangle.
\end{equation}
\noindent The following theorem enables lower bounds for history-state preparation.

\begin{thm} \label{thm:hist-state-lower-bound}
    Suppose we are given $k$ copies of $|\psi(0)\rangle$, which is an encoding of the initial condition of a differential equation. Suppose the solution to the differential equation satisfies $\||\psi(t)\rangle \| = \Theta(1)$ for all $t \in [0,T]$. If preparing the final-time solution, $|\psi(T)\rangle$, of the differential equation requires $\Omega(f(T))$ copies of the initial state, then preparing the history state requires $k=\Omega(f(T)/T)$ copies of the initial state.
\end{thm}
\begin{proof}
    For a contradiction, suppose we have an algorithm that prepares the history state using $o(f(T)/T)$ copies of the initial state. If $|\psi(jh)\rangle$ is the solution at timestep $j$, the history state takes the form
    \begin{equation}
        |\Psi\rangle = \frac{1}{\mathcal N}\sum_{j=1}^m |j\rangle |\psi(jh)\rangle,
    \end{equation}
    where $\mathcal N$ is the normalization factor. Since $\||\psi(t)\rangle \| = \Theta(1)$, $\mathcal N = \Theta(\sqrt{m})$. One can then use $O(m) = O(T)$ copies of the history state to post-select on the final solution, resulting in a final-state-preparation algorithm with complexity $o(f(T))$, which violates the lower bound. Thus the complexity of history state preparation is $\Omega(f(T)/T)$.
\end{proof}

\begin{corr} [History state lower bound for the KdV equation] \label{cor:kdv-hist}
    Any quantum algorithm producing the history state of the KdV equation requires $\Omega\left ( T \right )$ copies of the initial state.
\end{corr}
\begin{proof}
    Let $|\psi_{\phi}(t)\rangle$ be the solution to the KdV equation at time $t$, as defined in Problem~\ref{prob:final-state-kdv}. Note that $|\psi_{\phi}(t)\rangle$ is norm-preserving due to \cref{lem:kdv-norm-cons}. The proof then follows from \cref{thm:kdv-lowerbound} and \cref{thm:hist-state-lower-bound}.
\end{proof}
\begin{corr} [History state lower bound for the Euler equations, tensor-product encoding] \label{cor:euler-hist-tens-prod}
    Any quantum algorithm producing the history state of the Euler equations, as encoded in Problem~\ref{prob:final-state-euler-tensor-prod}, requires $e^{\Omega(T)}$ copies of the initial state.
\end{corr}
\begin{proof}
    Let $|\mathbf{u}(t)\rangle \otimes |p(t)\rangle $ be the solution to the Euler equations at time $t$, as defined in Problem~\ref{prob:final-state-euler-tensor-prod}. Note that $\||\mathbf{u}(t)\rangle\| = 1$ from \cref{lem:euler-norm-cons}. To properly encode $|p(t)\rangle$, we need $\bar p \neq 0$. From \cref{lem:pressure-bounds}, we then see that $\||p(t)\rangle \| = \Theta(1)$, so $\||\mathbf{u}(t)\rangle \otimes |p(t)\rangle\| = \Theta(1)$. The proof then follows from \cref{thm:euler-lower-bound} and \cref{thm:hist-state-lower-bound}.
\end{proof}

\begin{corr} [History state lower bound for the Euler equations, direct-sum encoding] \label{cor:euler-hist}
    Any quantum algorithm producing the history state of the Euler equations, as encoded in Problem~\ref{prob:final-state-euler-direct-sum}, requires $e^{\Omega(T)}$ copies of the initial state.
\end{corr}
\begin{proof}
    Let $|\mathbf w(t)\rangle $ be the solution to the Euler equations at time $t$, as defined in Problem~\ref{prob:final-state-euler-direct-sum}. Note that $\||\mathbf w(t)\rangle \| = \Theta(1)$ from \cref{lem:euler-norm-cons,lem:pressure-bounds} for any $\bar p$. The proof then follows from \cref{thm:euler-lower-bound-direct-sum} and \cref{thm:hist-state-lower-bound}.
\end{proof}
\begin{remark}
    For the Euler equations, as the velocity and pressure are typically treated separately, another natural encoding of the history state is a tensor product of history states:
    \begin{equation}
        |\Psi\rangle \propto \left ( \sum_{j=1}^m |j\rangle |\mathbf u(jh) \rangle \right ) \otimes \left ( \sum_{j=1}^m |j\rangle |p(jh) \rangle \right ).
    \end{equation}
    It is easy to see that the same exponential lower bound applies for preparing this state as well, as one can simply discard the pressure history state and use \cref{lem:euler-norm-cons}, \cref{thm:euler-lower-bound}, and \cref{thm:hist-state-lower-bound}.
\end{remark}

\section{Discussion and future work} \label{sec:future-work}

In this paper we proved lower bounds for quantum algorithms that simulate fluid dynamical equations, namely the KdV and Euler equations. Both of these equations are widely used in numerical analysis of fluids. We proved lower bounds for both final state preparation and history state preparation. For the KdV equation we showed a polynomial (in simulation time $T$) lower bound, and for the Euler equations we showed an exponential lower bound. 

Note that both the KdV equation and the Euler equations have no dissipation. Using the language of~\cite{liu2021efficient}, $\mathrm R=\infty$ for both these equations when discretized. While the lack of dissipation may be responsible for the difficulty of simulating these equations, we find lower bounds that are exponentially separated. It would be interesting to find an exponential lower bound for the KdV equation, although from numerical simulations alone this seems unlikely. If such a lower bound cannot be shown, a fruitful research direction might be to develop a polynomial-time quantum algorithm for simulating the KdV equation.

It is also worth thinking about the limitations of our lower bound and whether there are ways of circumventing it. A main limitation of the lower bound for the Euler equations is the requirement for the error to be exponentially small. Ideally, we would like a lower bound for constant error, similar to the one we have for the KdV equation. One should be able to relax this constraint by obtaining a tighter bound on the linearization error. We leave this task to future work.

Another notable limitation is the fact that our lower bound for the Euler equations applies near unstable fixed points of dynamical systems. It is unclear whether a similar lower bound holds further away from the fixed point. If the dynamics are chaotic away from the fixed point, a similar exponential lower bound should apply as per~\cite{lewis2024limitations}, although this has yet to be rigorously proven. However, if the dynamics is not chaotic, perhaps there is not an exponential lower bound. If that is the case, perhaps one can develop a polynomial-time quantum algorithm for simulating nonlinear dynamical systems in regimes of phase space that are sufficiently far from instabilities, thereby circumventing our lower bound.

Another way to circumvent the lower bound is to note that certain properties of fluids do not depend sensitively on initial conditions. A notable example is the appearance of the Kolmogorov scaling law in turbulence, where two nearby initial trajectories in a turbulent fluid can diverge very quickly but exhibit the same energy spectrum~\cite{pope2000turbulent}. Can one give an efficient algorithm for fluid simulation under some assumption that captures this notion? This direction has been somewhat explored by~\cite{bravyi2025quantum} for stochastic nonlinear differential equations, and it merits further investigation.

While we consider amplitude-encoded states in our paper, one could consider encoding the solution digitally, where an $n$-dimensional vector $\mathbf{u}$ is encoded as 
\begin{equation}
    |\psi\rangle \frac{1}{\sqrt{n}} \sum_{j=0}^{n-1} |j\rangle |u_j\rangle,
\end{equation}
with $|u_j\rangle$ being a binary encoding of $u_j$. It turns out that digitally-encoded states are equivalent to amplitude-encoded ones~\cite{mitarai2019quantum}. Thus, our lower bound for amplitude-encoded states translates directly to the digital setting.

Note that the lower bounds for history-state preparation rely on the norm-preserving property of the KdV and Euler equations. While this is true in the continuous case, the discretized Euler and KdV equations do not necessarily preserve the norm. In fact, certain discretizations of these equations introduce ``numerical dissipation", whereby high-wavenumber modes get artificially damped. Thus, the norm of the solution would decay with time at the expense of offering better numerical properties (such as smoothness and stability). However, in such cases, the strength of the numerical dissipation decreases as the resolution increases. Thus, any ``sufficiently faithful" discretization of the original equations should inherit the same lower bound. Our lower bounds for the continuous equations are advantageous in that they are agnostic about the specifics of the discretization of the equations.

As we briefly discussed at the beginning of the paper, while our exponential lower bound applies to the incompressible Euler equations we also expect it to apply to the incompressible Navier-Stokes equations for sufficiently large Reynolds number. This is because the incompressible Euler equations correspond to the inviscid, infinite-Reynolds-number limit of the incompressible Navier-Stokes equations. We leave the task of finding a threshold Reynolds number above which our lower bound applies to future work. We expect that similar lower bounds apply to more complicated equations, such as the magnetohydrodynamic equations in plasma physics, though the exact parameter regimes that support such lower bounds would need to be determined. On the other hand, it is unknown whether a similar lower bound applies to other fluid dynamics models, such as large eddy simulations (LES) and Reynolds-averaged Navier-Stokes (RANS), which are typically preferred for simulating turbulent flows over direct numerical simulation~\cite{pope2000turbulent}.

From a broader perspective, this work reframes the quest for quantum CFD: rather than seeking universal speedups for realistic highly nonlinear flows, future progress is more likely to come from identifying restricted regimes (short times, strong dissipation, or special observables) where quantum speedups are viable. In this sense, our results sharpen our understanding of where quantum computers can and cannot help in scientific computing, guiding effort away from unattainable goals and toward problem formulations where genuine, physically meaningful quantum advantage may still exist.

\section*{Acknowledgements}
AA thanks Rodolfo Rosales for insight on the Euler equations, as well as Susan Friedlander for insight on the instability proof of the smooth Kelvin-Helmholtz flow. AMC thanks Xiangyu Li for discussions of quantum algorithms for fluid dynamics, including \cite{li2025potential}. The authors acknowledge support from the US Department of Energy (grant DE-SC0020264).

\bibliography{ref}

@book{hesthaven2008nodal,
  title={Nodal discontinuous Galerkin methods: Algorithms, analysis, and applications},
  author={Hesthaven, Jan S. and Warburton, Tim},
  year={2008},
  publisher={Springer},
  url={https://doi.org/10.1007/978-0-387-72067-8}
}

@article{monaghan1992smoothed,
  title={Smoothed particle hydrodynamics},
  author={Monaghan, Joe J.},
  journal={Annual Review of Astronomy and Astrophysics},
  volume={30},
  pages={543--574},
  year={1992},
  url={https://doi.org/10.1088/0034-4885/68/8/R01}
}

@book{kruger2017lattice,
  title={The lattice Boltzmann method},
  author={Kr{\"u}ger, Timm and Kusumaatmaja, Halim and Kuzmin, Alexandr and Shardt, Orest and Silva, Goncalo and Viggen, Erlend Magnus},
  series={Graduate Texts in Physics},
  volume={10},
  year={2017},
  publisher={Springer},
  url={https://doi.org/10.1007/978-3-319-44649-3}
}

@book{toro2013riemann,
  title={Riemann solvers and numerical methods for fluid dynamics: A practical introduction},
  author={Toro, Eleuterio F.},
  year={2013},
  publisher={Springer},
  url={https://doi.org/10.1007/b79761}
}

@book{durran2010numerical,
  title={Numerical methods for fluid dynamics: With applications to geophysics},
  author={Durran, Dale R.},
  series={Texts in Applied Mathematics},
  volume={32},
  year={2010},
  publisher={Springer},
  url={https://doi.org/10.1007/978-1-4419-6412-0}
}

@book{chandrasekhar2013hydrodynamic,
  title={Hydrodynamic and hydromagnetic stability},
  author={Chandrasekhar, Subrahmanyan},
  year={2013},
  publisher={Courier Corporation}
}

@book{ferziger2019computational,
  title={Computational methods for fluid dynamics},
  author={Ferziger, Joel H. and Peri{\'c}, Milovan and Street, Robert L},
  year={2019},
  publisher={Springer},
  url={https://doi.org/10.1007/978-3-642-56026-2}
}

@article{jennings2025end,
  title={An end-to-end quantum algorithm for nonlinear fluid dynamics with bounded quantum advantage},
  author={Jennings, David and Korzekwa, Kamil and Lostaglio, Matteo and Ashworth, Richard and Marsili, Emanuele and Rolston, Stephen},
  journal={arXiv preprint arXiv:2512.03758},
  year={2025},
  url={https://doi.org/10.48550/arXiv.2512.03758}
}

@article{itani2024quantum,
  title={Quantum algorithm for lattice Boltzmann (QALB) simulation of incompressible fluids with a nonlinear collision term},
  author={Itani, Wael and Sreenivasan, Katepalli R. and Succi, Sauro},
  journal={Physics of Fluids},
  volume={36},
  number={1},
  year={2024},
  publisher={AIP Publishing},
  url={https://doi.org/10.1063/5.0176569}
}

@article{itani2022analysis,
  title={Analysis of {C}arleman linearization of lattice {B}oltzmann},
  author={Itani, Wael and Succi, Sauro},
  journal={Fluids},
  volume={7},
  number={1},
  pages={24},
  year={2022},
  publisher={MDPI},
  url={https://doi.org/10.3390/fluids7010024}
}

@article{li2025potential,
  title={Potential quantum advantage for simulation of fluid dynamics},
  author={Li, Xiangyu and Yin, Xiaolong and Wiebe, Nathan and Chun, Jaehun and Schenter, Gregory K and Cheung, Margaret S. and M{\"u}lmenst{\"a}dt, Johannes},
  journal={Physical Review Research},
  volume={7},
  number={1},
  pages={013036},
  year={2025},
  publisher={APS},
  url={https://doi.org/10.1103/PhysRevResearch.7.013036}
}

@article{harrow2009quantum,
  title={Quantum algorithm for linear systems of equations},
  author={Harrow, Aram W. and Hassidim, Avinatan and Lloyd, Seth},
  journal={Physical Review Letters},
  volume={103},
  number={15},
  pages={150502},
  year={2009},
  publisher={APS},
  url={https://doi.org/10.1103/PhysRevLett.103.150502}
}

@article{childs2017quantum,
  title={Quantum algorithm for systems of linear equations with exponentially improved dependence on precision},
  author={Childs, Andrew M. and Kothari, Robin and Somma, Rolando D.},
  journal={SIAM Journal on Computing},
  volume={46},
  number={6},
  pages={1920--1950},
  year={2017},
  publisher={SIAM},
  url={https://doi.org/10.1137/16M1087072}
}

@article{costa2022optimal,
  title={Optimal scaling quantum linear-systems solver via discrete adiabatic theorem},
  author={Costa, Pedro C.~S. and An, Dong and Sanders, Yuval R. and Su, Yuan and Babbush, Ryan and Berry, Dominic W.},
  journal={PRX Quantum},
  volume={3},
  number={4},
  pages={040303},
  year={2022},
  publisher={APS},
  url={https://doi.org/10.1103/PRXQuantum.3.040303}
}

@article{fang2023time,
  title={Time-marching based quantum solvers for time-dependent linear differential equations},
  author={Fang, Di and Lin, Lin and Tong, Yu},
  journal={Quantum},
  volume={7},
  pages={955},
  year={2023},
  publisher={Verein zur F{\"o}rderung des Open Access Publizierens in den Quantenwissenschaften},
  url={https://doi.org/10.22331/q-2023-03-20-955}
}

@article{an2023linear,
  title={Linear combination of {H}amiltonian simulation for nonunitary dynamics with optimal state preparation cost},
  author={An, Dong and Liu, Jin-Peng and Lin, Lin},
  journal={Physical Review Letters},
  volume={131},
  number={15},
  pages={150603},
  year={2023},
  publisher={APS},
  url={https://doi.org/10.1103/PhysRevLett.131.150603}
}

@article{an2023quantum,
  title={Quantum algorithm for linear non-unitary dynamics with near-optimal dependence on all parameters},
  author={An, Dong and Childs, Andrew M. and Lin, Lin},
  eprint={arXiv:2312.03916},
  journal={Communications in Mathematical Physics},
  volume={407},
  pages={19},
  year={2026},
  url={https://doi.org/10.1007/s00220-025-05509-w}
}

@inproceedings{low2024quantum,
  title={Quantum eigenvalue processing},
  author={Low, Guang Hao and Su, Yuan},
  booktitle={65th Annual IEEE Symposium on Foundations of Computer Science (FOCS)},
  pages={1051--1062},
  year={2024},
  url={https://doi.org/10.1109/FOCS61266.2024.00070}
}

@article{shang2025designing,
  title={Designing a nearly optimal quantum algorithm for linear differential equations via {L}indbladians},
  author={Shang, Zhong-Xia and Guo, Naixu and An, Dong and Zhao, Qi},
  journal={Physical Review Letters},
  volume={135},
  number={12},
  pages={120604},
  year={2025},
  publisher={APS},
  url={https://doi.org/10.1103/cvl9-97qg}
}

@article{low2025optimal,
  title={Optimal quantum simulation of linear non-unitary dynamics},
  author={Low, Guang Hao and Somma, Rolando D.},
  journal={arXiv preprint arXiv:2508.19238},
  year={2025},
  url={https://doi.org/10.48550/arXiv.2508.19238}
}

@article{leyton2008quantum,
  title={A quantum algorithm to solve nonlinear differential equations},
  author={Leyton, Sarah K and Osborne, Tobias J.},
  journal={arXiv preprint arXiv:0812.4423},
  year={2008},
  url={https://doi.org/10.48550/arXiv.0812.4423}
}

@article{berry2014high,
  title={High-order quantum algorithm for solving linear differential equations},
  author={Berry, Dominic W.},
  journal={Journal of Physics A: Mathematical and Theoretical},
  volume={47},
  number={10},
  pages={105301},
  year={2014},
  publisher={IOP Publishing},
  url={https://doi.org/10.1088/1751-8113/47/10/105301}
}

@article{berry2017quantum,
  title={Quantum algorithm for linear differential equations with exponentially improved dependence on precision},
  author={Berry, Dominic W. and Childs, Andrew M. and Ostrander, Aaron and Wang, Guoming},
  journal={Communications in Mathematical Physics},
  volume={356},
  number={3},
  pages={1057--1081},
  year={2017},
  publisher={Springer},
  url={https://doi.org/10.1007/s00220-017-3002-y}
}

@article{krovi2023improved,
  title={Improved quantum algorithms for linear and nonlinear differential equations},
  author={Krovi, Hari},
  journal={Quantum},
  volume={7},
  pages={913},
  year={2023},
  publisher={Verein zur F{\"o}rderung des Open Access Publizierens in den Quantenwissenschaften},
  url={https://doi.org/10.22331/q-2023-02-02-913}
}

@article{lewis2024limitations,
  title={Limitations for quantum algorithms to solve turbulent and chaotic systems},
  author={Lewis, Dylan and Eidenbenz, Stephan and Nadiga, Balasubramanya and Suba{\c{s}}{\i}, Yi{\u{g}}it},
  journal={Quantum},
  volume={8},
  pages={1509},
  year={2024},
  publisher={Verein zur F{\"o}rderung des Open Access Publizierens in den Quantenwissenschaften},
  url={https://doi.org/10.22331/q-2024-10-24-1509}
}

@article{liu2021efficient,
  title={Efficient quantum algorithm for dissipative nonlinear differential equations},
  author={Liu, Jin-Peng and Kolden, Herman {\O}ie and Krovi, Hari and Loureiro, Nuno F. and Trivisa, Konstantina and Childs, Andrew M.},
  journal={Proceedings of the National Academy of Sciences},
  volume={118},
  number={35},
  pages={e2026805118},
  year={2021},
  publisher={National Academy of Sciences},
  url={https://doi.org/10.1073/pnas.2026805118}
}

@article{bravyi2025quantum,
  title={Quantum simulation of a noisy classical nonlinear dynamics},
  author={Bravyi, Sergey and Manson-Sawko, Robert and Zayats, Mykhaylo and Zhuk, Sergiy},
  journal={arXiv preprint arXiv:2507.06198},
  year={2025},
  url={https://doi.org/10.48550/arXiv.2507.06198}
}

@book{smith1995monotone,
  title={Monotone dynamical systems: An introduction to the theory of competitive and cooperative systems},
  author={Smith, Hal L.},
  series={Mathematical Surveys and Monographs},
  volume={41},
  year={1995},
  publisher={AMS},
  url={https://doi.org/10.1090/surv/041}
}

@article{friedlander1997nonlinear,
  title={Nonlinear instability in an ideal fluid},
  author={Friedlander, Susan and Strauss, Walter and Vishik, Misha},
  journal={Annales de l'Institut Henri Poincar{\'e} C, Analyse non lin{\'e}aire},
  volume={14},
  number={2},
  pages={187--209},
  year={1997},
  url={https://doi.org/10.1016/S0294-1449(97)80144-8}
}

@article{polyanin2009nonlinearnavier,
  title={Nonlinear instability of the solutions of the {N}avier-{S}tokes equations: Formulas for constructing exact solutions},
  author={Polyanin, Andrei Dmitrievich},
  journal={Theoretical Foundations of Chemical Engineering},
  volume={43},
  number={6},
  pages={881--888},
  year={2009},
  publisher={Springer},
  url={https://doi.org/10.1134/S0040579509060050}
}

@article{polyanin2009nonlinearhydro,
  title={On the nonlinear instability of the solutions of hydrodynamic-type systems},
  author={Polyanin, Andrei Dmitrievich},
  journal={JETP Letters},
  volume={90},
  number={3},
  pages={217--221},
  year={2009},
  publisher={Springer},
  url={https://doi.org/10.1134/S0021364009150120}
}

@book{drazin2004hydrodynamic,
  title={Hydrodynamic stability},
  author={Drazin, Philip G. and Reid, William Hill},
  year={2004},
  publisher={Cambridge University Press},
  url={https://doi.org/10.1017/CBO9780511616938}
}

@article{abrams1998nonlinear,
  title={Nonlinear quantum mechanics implies polynomial-time solution for {NP}-complete and {{\#{P}}} problems},
  author={Abrams, Daniel S. and Lloyd, Seth},
  journal={Physical Review Letters},
  volume={81},
  number={18},
  pages={3992},
  year={1998},
  publisher={APS},
  url={https://doi.org/10.1103/PhysRevLett.81.3992}
}

@article{childs2016optimal,
  title={Optimal state discrimination and unstructured search in nonlinear quantum mechanics},
  author={Childs, Andrew M. and Young, Joshua},
  journal={Physical Review A},
  volume={93},
  number={2},
  pages={022314},
  year={2016},
  publisher={APS},
  url={https://doi.org/10.1103/PhysRevA.93.022314}
}

@article{aaronson2005np,
      title={{NP}-complete Problems and Physical Reality}, 
      author={Scott Aaronson},
      year={2005},
      journal={ACM SIGACT News},
      url={https://doi.org/10.48550/arXiv.quant-ph/0502072}
}

@article{lloyd1996universal,
  title={Universal quantum simulators},
  author={Lloyd, Seth},
  journal={Science},
  volume={273},
  number={5278},
  pages={1073--1078},
  year={1996},
  publisher={American Association for the Advancement of Science},
  url={https://doi.org/10.1126/science.273.5278.1073}
}

@article{shor1999polynomial,
  title={Polynomial-time algorithms for prime factorization and discrete logarithms on a quantum computer},
  author={Shor, Peter W.},
  journal={SIAM Review},
  volume={41},
  number={2},
  pages={303--332},
  year={1999},
  publisher={SIAM},
  url={https://doi.org/10.1137/S0036144598347011}
}

@article{papageorgiou1991route,
  title={The route to chaos for the {K}uramoto-{S}ivashinsky equation},
  author={Papageorgiou, Demetrios T. and Smyrlis, Yiorgos S.},
  journal={Theoretical and Computational Fluid Dynamics},
  volume={3},
  number={1},
  pages={15--42},
  year={1991},
  publisher={Springer},
  url={https://doi.org/10.1007/BF00271514}
}

@article{smyrlis1991predicting,
  title={Predicting chaos for infinite dimensional dynamical systems: The {K}uramoto-{S}ivashinsky equation, a case study.},
  author={Smyrlis, Yiorgos S. and Papageorgiou, Demetrios T.},
  journal={Proceedings of the National Academy of Sciences},
  volume={88},
  number={24},
  pages={11129--11132},
  year={1991},
  url={https://doi.org/10.1073/pnas.88.24.11129}
}

@incollection{sivashinsky1988nonlinear,
  title={Nonlinear analysis of hydrodynamic instability in laminar flames—{I}. {D}erivation of basic equations},
  author={Sivshinsky, Gregory I.},
  booktitle={Dynamics of Curved Fronts},
  pages={459--488},
  year={1988},
  publisher={Elsevier},
  url={https://doi.org/10.1016/B978-0-08-092523-3.50048-4}
}

@article{michelson1977nonlinear,
  title={Nonlinear analysis of hydrodynamic instability in laminar flames—{II}. {N}umerical experiments},
  author={Michelson, Daniel M. and Sivashinsky, Gregory I.},
  journal={Acta Astronautica},
  volume={4},
  number={11-12},
  pages={1207--1221},
  year={1977},
  publisher={Elsevier},
  url={https://doi.org/10.1016/0094-5765(77)90097-2}
}

@article{kuramoto1978diffusion,
  title={Diffusion-induced chaos in reaction systems},
  author={Kuramoto, Yoshiki},
  journal={Progress of Theoretical Physics Supplement},
  volume={64},
  pages={346--367},
  year={1978},
  publisher={Oxford University Press},
  url={https://doi.org/10.1143/PTPS.64.346}
}

@book{linares2014introduction,
  title={Introduction to nonlinear dispersive equations},
  author={Linares, Felipe and Ponce, Gustavo},
  year={2014},
  publisher={Springer},
  url={https://doi.org/10.1007/978-1-4939-2181-2}
}

@book{strogatz2024nonlinear,
  title={Nonlinear dynamics and chaos: with applications to physics, biology, chemistry, and engineering},
  author={Strogatz, Steven H.},
  year={2024},
  publisher={Chapman and Hall/CRC},
  url={https://doi.org/10.1201/9780429398490}
}

@book{chen1984introduction,
  title={Introduction to plasma physics and controlled fusion},
  author={Chen, Francis F.},
  year={1984},
  publisher={Springer},
  url={https://doi.org/10.1007/978-3-319-22309-4}
}

@article{korteweg1895xli,
  title={{O}n the change of form of long waves advancing in a rectangular canal, and on a new type of long stationary waves},
  author={Korteweg, Diederik Johannes and de Vries, Gustav},
  journal={The London, Edinburgh, and Dublin Philosophical Magazine and Journal of Science},
  volume={39},
  number={240},
  pages={422--443},
  year={1895},
  publisher={Taylor \& Francis}
}

@book{nielsen2010quantum,
  title={Quantum computation and quantum information},
  author={Nielsen, Michael A. and Chuang, Isaac L.},
  year={2010},
  publisher={Cambridge University Press},
  url={https://doi.org/10.1017/CBO9780511976667}
}

@article{costa2025further,
  title={Further improving quantum algorithms for nonlinear differential equations via higher-order methods and rescaling},
  author={Costa, Pedro C.~S. and Schleich, Philipp and Morales, Mauro E.~S. and Berry, Dominic W.},
  journal={npj Quantum Information},
  volume={11},
  number={1},
  pages={141},
  year={2025},
  publisher={Nature Publishing Group UK London},
  url={https://doi.org/10.1038/s41534-025-01084-z}
}

@article{wu2025quantum,
  title={Quantum algorithms for nonlinear dynamics: Revisiting {C}arleman linearization with no dissipative conditions},
  author={Wu, Hsuan-Cheng and Wang, Jingyao and Li, Xiantao},
  journal={SIAM Journal on Scientific Computing},
  volume={47},
  number={2},
  pages={A943--A970},
  year={2025},
  publisher={SIAM},
  url={https://doi.org/10.1137/24M1665799}
}

@article{jennings2025quantum,
  title={Quantum algorithms for general nonlinear dynamics based on the {C}arleman embedding},
  author={Jennings, David and Korzekwa, Kamil and Lostaglio, Matteo and Sornborger, Andrew T. and Suba{\c{s}}{\i}, Yi{\u{g}}it and Wang, Guoming},
  journal={arXiv preprint arXiv:2509.07155},
  year={2025},
  url={https://doi.org/10.48550/arXiv.2509.07155}
}

@book{pope2000turbulent,
  title={Turbulent Flows},
  author={Pope, Stephen B},
  year={2000},
  publisher={Cambridge University Press},
  url={https://doi.org/10.1017/CBO9781316179475}
}

@article{mitarai2019quantum,
  title={Quantum analog-digital conversion},
  author={Mitarai, Kosuke and Kitagawa, Masahiro and Fujii, Keisuke},
  journal={Physical Review A},
  volume={99},
  number={1},
  pages={012301},
  year={2019},
  publisher={APS},
  url={https://doi.org/10.1103/PhysRevA.99.012301}
}

@book{drazin1989solitons,
  title={Solitons: an introduction},
  author={Drazin, Philip G and Johnson, Robin Stanley},
  volume={2},
  year={1989},
  publisher={Cambridge university press},
  url={https://doi.org/10.1017/CBO9781139172059}
}

@book{bahouri2011fourier,
  title={Fourier analysis and nonlinear partial differential equations},
  author={Bahouri, Hajer and Chemin, Jean-Yves and Danchin, Raphaël},
  year={2011},
  publisher={Springer},
  url={https://doi.org/10.1007/978-3-642-16830-7}
}

@book{chemin1998perfect,
  title={Perfect incompressible fluids},
  author={Chemin, Jean-Yves},
  year={1998},
  publisher={Oxford University Press},
  url={https://doi.org/10.1093/oso/9780198503972.001.0001}
}

@book{jardin2010computational,
  title={Computational methods in plasma physics},
  author={Jardin, Stephen},
  year={2010},
  publisher={CRC press},
  url={https://doi.org/10.1201/EBK1439810958}
}

\appendix
\section{Proof of Lemma~\ref{lem:laplace}} \label{app:laplace}
\laplace*
\begin{proof}
    First, let us decompose $\phi$ into a constant $\tilde \phi$ and a zero-mean part $\phi'$. In other words,
    \begin{equation}
        \phi = \tilde \phi + \phi',
    \end{equation}
    such that $\nabla^2 \tilde \phi = 0$, $\nabla^2 \phi' = f$, and
    \begin{align}
        \int_\Omega \phi' \, d\Omega = 0.
    \end{align}
    We write $\phi'$ and $f$ in the Fourier basis. Let $\mathbf r\coloneq (x_1,\dots, x_d) \in \mathbb T^d$. Then
    \begin{align}
        f = \sum_{\mathbf{j}\in\mathbb Z^{d}} c_{\mathbf j} e^{i\mathbf{j}\cdot \mathbf r} \\
        \phi' = \sum_{\mathbf j\in\mathbb Z^d} b_{\mathbf j}e^{i\mathbf{j}\cdot \mathbf r}.
    \end{align}
    In the Fourier basis, the Laplacian operator's action on $\phi$ yields
    \begin{align}
        \nabla ^2 \phi' =- \sum_{\mathbf j \in \mathbb Z^d}\|\mathbf{j}\|^2 b_{\mathbf j} e^{i\mathbf j \cdot \mathbf r}
    \end{align}
    Since $\phi'$ has zero mean, $b_0 = 0$. Then, solving for $b_{\mathbf j}$ gives:
    \begin{equation}
        b_{\mathbf j} = \begin{cases}
            0, & \mathbf j = 0 \\
            -\frac{1}{\|\mathbf j\|^2} c_{\mathbf j}, & \mathbf j \neq 0
        \end{cases}\ .
    \end{equation}
    Evaluating $\|\phi'\|_{L^2}$ gives
    \begin{align}
        \|\phi'\|_{L^2}^2 &= \sum_{j\in \mathbb Z^d} |b_{\mathbf j}|^2 \\
        &= \sum_{j \in \mathbb Z^d\setminus \{0\}} \frac{1}{\|\mathbf j\|^4}|c_\mathbf j|^2 \\
        &\leq \sum_{j\in \mathbb Z^d} |c_{\mathbf j}|^2 \\
        &= \|f\|^2_{L^2}.
    \end{align}
    As $\tilde \phi$ satisfies $\nabla^2 \tilde \phi = 0$, $\phi$ satisfies Eq.~\eqref{eq:laplace}. Putting things together, we get
    \begin{align}
        \|\phi\|_{L^2} &\leq \|\tilde \phi \|_{L^2} + \|\phi'\|_{L^2} \\
        &\leq \tilde \phi |\Omega|^{1/2} + \|f\|_{L^2}.
    \end{align}
    Setting $\bar \phi \coloneq \tilde \phi |\Omega|^{1/2}$ completes the proof.
\end{proof}

\section{Proofs of Lemmas~\ref{fact:euc-tr-rel} and \ref{fact:tr-prod-rel}}
\label{app:std-proofs}

\euclidtrace*
\begin{proof}
Recall that for pure states $\ket{\psi}, \ket{\phi}$ we have $D(\rho_\psi, \rho_\phi) = \sqrt{1-\abs{\braket{\psi \, | \, \phi}}^2}$. Letting $\alpha = \abs{\braket{\psi \, | \, \phi}}$ denote the absolute value of the overlap, $d=d(\ket{\psi}, \ket{\phi})$, and $D =D(\rho_\psi, \rho_\phi)$ we have \begin{align}
    \alpha &= 1-d^2/2 \\
    D &= \sqrt{1-(1-d^2/2)^2} \\
    &= d\sqrt{1-d^2/4} \\
    \frac{1}{\sqrt{2}} &\leq \sqrt{1-d^2/4} \leq 1,
\end{align}
where the last line follows from $0 \leq d \leq \sqrt{2}$, establishing the claim.
\end{proof}

\kcopies*
\begin{proof}
Observe that overlap is multiplicative over copies, i.e., $\abs{\bra{\psi}^{\otimes k} \ket{\psi}^{\otimes k}} = \alpha^k.$ Writing $\alpha=1-\epsilon$, we have
\begin{align}
    D(\rho_\psi, \rho_\phi) = \sqrt{1-\alpha^2} = \sqrt{2\epsilon - \epsilon^2}.
\end{align}
On the other hand, we have \begin{align}
    D(\rho_\psi^{\otimes k}, \rho_\phi^{\otimes k}) &=\sqrt{1-\alpha^{2k}} \\
    &= \sqrt{1-((1-\epsilon)^k)^2} \\
    &\leq \sqrt{1-(1-k\epsilon)^2} \\
    &= \sqrt{2k\epsilon - k^2 \epsilon^2} \\
    &= \sqrt{k} \cdot \sqrt{2\epsilon-k\epsilon^2} \\
    &\leq \sqrt{k} \cdot \sqrt{2\epsilon - \epsilon^2},
\end{align}
showing the claim.
\end{proof}

\section{Proof of \cref{thm:stability}}  \label{app:proof_stability}
\stability*
\begin{proof}
    \begin{enumerate}
        \item Recall that in Section~\ref{subsec:linear-stab-analysis}, we obtained the recurrence relation
    \begin{equation} \label{eq:recurrence-relation}
        \frac{2\gamma}{kU_0} (j^2 + k^2)c_j + (k^2-m^2+(j-m)^2)c_{j-m} - (k^2 - m^2 + (j+m)^2) c_{j+m} = 0
    \end{equation}
    for the Fourier coefficients $\{c_j\}$ of the stream function of the perturbation with $j\in \mathbb Z$. Using this as our starting point, let
    \begin{equation}
        a_j \coloneq \frac{2\gamma}{kU_0} \frac{j^2 + k^2}{j^2 + k^2 - m^2}.
    \end{equation}
    Note that $a_j = a_{-j}$ The denominator of $a_j$ is non-vanishing due to the constraint of the theorem. Furthermore, let
    \begin{equation}
        d_j \coloneq c_j (j^2 + k^2 - m^2).
    \end{equation}
    Using this, we can write Eq.~\eqref{eq:recurrence-relation} as
    \begin{equation} \label{eq:d}
        d_{j+m} = a_j d_j + d_{j-m}.
    \end{equation}
    For $j \in m\mathbb{Z}$, we construct a sequence $\{d_j\}$. We set $d_j = 0$ for $j \not\equiv 0 \ (\mathrm{mod} \ m)$. Let
    \begin{align}
        \rho_j &= \frac{d_j}{d_{j-m}}, \ j>0 \label{eq:rho} \\
        \tilde{\rho}_j &= \frac{d_{j-m}}{d_j}, \ j\leq 0. \label{eq:rhotilde}
    \end{align}
    From Eqs.~\eqref{eq:d} and~\eqref{eq:rho}, we obtain
    \begin{align}
        \rho_m &= \frac{-1}{a_m - \rho_{2m}} \\
        &= \frac{-1}{a_m + \frac{1}{a_{2m}-\rho_{3m}}} \\
        &= \frac{-1}{a_m + \frac{1}{a_{2m}+\frac{1}{a_{3m}+\dots}}} \\
        &\eqqcolon \frac{-1}{[a_m,a_{2m},\dots]}
    \end{align}
    for $\gamma$ real and positive. We extend this definition for all $p\in\mathbb N$:
    \begin{equation}
        \rho_{pm} = \rho_{pm}(\gamma) = \frac{-1}{[a_{pm},a_{(p+1)m},\dots]}
    \end{equation}
    for $\gamma$ real and positive. Note that the continued fraction is convergent as all elements are positive and bounded away from zero.

    Furthermore, from Eqs.~\eqref{eq:d} and~\eqref{eq:rhotilde} we obtain
    \begin{align}
        \tilde \rho_0 &= \frac{1}{a_{-m} + \tilde \rho_{-m}} \\
        &= \frac{1}{a_{-m} + \frac{1}{a_{-2m} + \rho_{-2m}}} \\
        &= \frac{1}{[a_{-m},a_{-2m},\dots]}
    \end{align}
    for $\gamma$ real and positive. Again, for $p \in \mathbb N_0$, we let
    \begin{equation}
       \tilde \rho_{-pm} = \tilde \rho_{-pm}(\gamma) = \frac{1}{[a_{-(p+1)m},a_{-(p+2)m},\dots]}
    \end{equation}
    for $\gamma$ real and positive, which is also convergent.

    Now we turn our attention to Eq.~\eqref{eq:d} for $j = 0$, which we write as
    \begin{equation} \label{eq:j-zero-relation}
        \rho_m = a_0 + \tilde{\rho}_{0}.
    \end{equation}
    It is easy to see that $\rho_m = -\tilde\rho_0$. Thus, we can rearrange Eq.~\eqref{eq:j-zero-relation} to obtain
    \begin{equation}
        -\rho_m (\gamma)= \frac{1}{[a_m,a_{2m},\dots]} = -\frac{a_0}{2}.
    \end{equation}
    We can bound $-a_0/2$ using the convergents of $-\rho_m$:
    \begin{equation}
        g(\gamma) \coloneq \frac{1}{a_{m}+\frac{1}{a_{2m}}} < -\frac{a_0}{2} < \frac{1}{a_{m}} \eqqcolon f(\gamma).
    \end{equation}
    The functional forms are as follows:
    \begin{align}
        -\frac{a_0(\gamma)}{2} &= \frac{\gamma}{U_0} \frac{k}{m^2-k^2} \\
        f(\gamma) &= \frac{kU_0}{2\gamma} \frac{k^2}{m^2+k^2} \\
        g(\gamma) &= \frac{1}{\frac{2\gamma}{kU_0} \frac{m^2 + k^2}{k^2} + \frac{kU_0}{2\gamma} \frac{3m^2+k^2}{4m^2+k^2}}.
    \end{align}
    Figure~\ref{fig:bounds} shows plots of these functions for different values of $k$. The intersection of $-a_0/2$ with $f$ and $g$ is the upper and lower bounds for the growth rate $\gamma$, which we denote by $\gamma_u$ and $\gamma_l$, respectively. We see that for smaller values of $k$, $\gamma_l>0$, so we are guaranteed to have an instability. However, for larger values of $k$, $\gamma_l = 0$, which is not sufficient to guarantee an instability.

    \begin{figure*}[ht!]
        \centering
        \begin{subfigure}[t]{0.49\textwidth}
            \centering
            \includegraphics[width=0.95\textwidth]{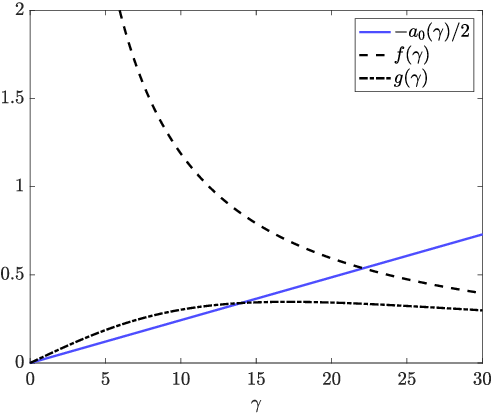}
            \caption{$k=63$.}
            \label{fig:1}
        \end{subfigure}%
        ~ 
        \begin{subfigure}[t]{0.49\textwidth}
            \centering
            \includegraphics[width=0.95\textwidth]{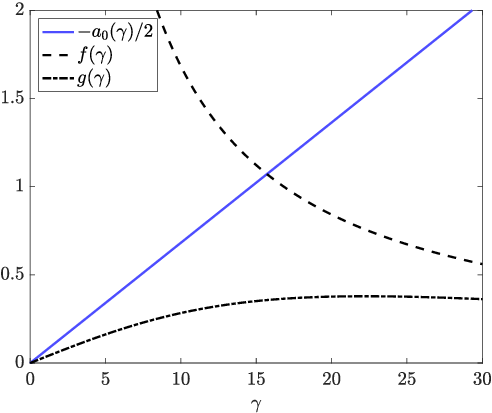}
            \caption{$k=74$.}
            \label{fig:2}
        \end{subfigure}
        \caption{Upper and lower bounds on the growth rate, with $U_0 = 1$ and $m=81$. When $k > k_{\mathrm{cutoff}} \approx 69$, the lower bound for the growth rate is $0$, preventing us from showing the existence of an instability.}
        \label{fig:bounds}
    \end{figure*}
    Formalizing this, we are guaranteed to have a positive $\gamma$ as long as the slope of $g(\gamma)$ is greater than that of $-a_0(\gamma)/2$ when $\gamma = 0$. In other words, $g'(0) > -\frac{a_0'(0)}{2}$, giving the  inequality
    \begin{equation}
        \frac{d}{d\gamma} \frac{\gamma}{kU_0} \frac{k^2}{m^2-k^2} \Bigg |_{\gamma = 0} < \frac{d}{d\gamma } \frac{1}{\frac{2\gamma}{kU_0} \frac{m^2+k^2}{k^2} + \frac{kU_0}{2\gamma} \frac{3m^2 + k^2}{4m^2 + k^2}} \Bigg |_{\gamma = 0}.
    \end{equation}
    Solving this inequality gives
    \begin{equation} \label{eq:inequality}
        k < k_{\mathrm{cutoff}} \coloneq \sqrt{\frac{\sqrt{177}-9}{6}}m \approx 0.85 m.
    \end{equation}
    \item The growth rate $\gamma$ of the instability can be upper- and lower-bounded using the inequalities
    \begin{equation}
        \frac{1}{\frac{2\gamma}{kU_0} \frac{m^2+k^2}{k^2} + \frac{kU_0}{2\gamma} \frac{3m^2 + k^2}{4m^2 + k^2}} < \frac{\gamma}{kU_0} \frac{k^2}{m^2-k^2} < \frac{kU_0}{2\gamma} \frac{k^2}{m^2+k^2},
    \end{equation}
    solving for $\gamma$ gives
    \begin{equation}
        \gamma_l < \gamma < \gamma_u,
    \end{equation}
    where
    \begin{align}
        \gamma_u &\coloneq U_0 k\sqrt{\frac{1}{2}\frac{m^2 - k^2}{m^2+k^2}}  \\
        \gamma_l &\coloneq U_0k \sqrt{\frac{1}{2} \frac{8m^4-9m^2k^2-3k^4}{(m^2+k^2)(8m^2+2k^2)}}.
    \end{align}
    \item Note that
    \begin{equation}
        \lim_{j\rightarrow \infty}a_{j} = \frac{2\gamma}{kU_0}.
    \end{equation}
    We see that
    \begin{align}
        \rho_\infty \coloneq \lim_{p\rightarrow \infty}\rho_{pm} &=  \frac{\gamma}{kU_0} - \sqrt{1+\frac{\gamma^2 }{k^2U_0^2}} \\
        \lim_{p\rightarrow \infty}\tilde \rho_{pm} &= -\frac{\gamma}{kU_0} + \sqrt{1+\frac{\gamma^2 }{k^2U_0^2}} = -\rho_\infty.
    \end{align}
    Furthermore, for real and positive $\gamma$,
    \begin{equation} \label{eq:rhoinfty}
       -1< \rho_\infty < 0.
    \end{equation}
    To construct the eigenfunction, we pick the sequence
    \begin{align}
        d_0 &= \Delta \label{eq:Delta} \\
        d_{pm} &= \prod_{n = 1}^{p}\rho_{nm} \\
        d_{-pm} &= \prod_{n = 0}^{p-1} \tilde{\rho}_{-nm},
    \end{align}
    where $\Delta$ is an arbitrary constant. The above sequence automatically satisfies Eq.~\eqref{eq:d}, and $d_{pm}, \, d_{-pm} \rightarrow 0$ exponentially because of Eq.~\eqref{eq:rhoinfty}, meaning that $\psi$ is smooth.
    
    We now determine the Fourier coefficients $\{b_j\}$ of the perturbed pressure $\tilde p$ and show that they also decay exponentially, yielding a smooth profile.
    
    We first construct the velocity perturbation, $\tilde{\mathbf{v}}$, via Eq.~\eqref{eq:stream-function-def}:
\begin{align}
    \tilde v_x &= -\frac{\partial \psi}{\partial y} \\
    &= -ie^{\gamma t + ikx} \sum_{j\in \mathbb Z} jc_je^{ijy} ,
\end{align}
and
\begin{align}
    \tilde{v}_y &= \frac{\partial \psi}{\partial x} \\
    &= ike^{\gamma t +ikx} \sum_{j\in \mathbb Z} c_j e^{ijy}.
\end{align}
The pressure perturbation is determined by taking the divergence of Eq.~\eqref{eq:lin-euler}, which gives
\begin{equation} \label{eq:pert-pressure-laplace}
    -\nabla^2 \tilde p = \nabla \cdot \left [(\mathbf{u}_0\cdot \nabla) \tilde{\mathbf{v}} + (\tilde{\mathbf{v}}\cdot \nabla) \mathbf{u}_0 \right ].
\end{equation}
Since the right-hand side is known, we obtain a Laplace equation, which we can solve to obtain $\tilde{p}$. We can write the left-hand side of Eq.~\eqref{eq:pert-pressure-laplace} as
\begin{equation}
    -\nabla^2 \tilde p = e^{\gamma t + ikx}\sum_{j\in \mathbb Z}(k^2+j^2) b_je^{ijy}.
\end{equation}
The right-hand side can be calculated as follows:
\begin{align}
    \nabla \cdot \left [(\mathbf{u}_0\cdot \nabla) \tilde{\mathbf{v}} + (\tilde{\mathbf{v}}\cdot \nabla) \mathbf{u}_0 \right ] &= \nabla \cdot \begin{bmatrix}
        U_0 \sin(my) \partial_x\tilde{v}_x +\tilde{v}_y\partial_y U_0 \sin(my) \\
         U_0 \sin(my) \partial_x\tilde{v}_y
    \end{bmatrix}\\
    &= U_0 \left [ \sin(my) \partial_{xx}\tilde v_x +  2m\cos(my) \partial_{x} \tilde v_y + \sin(my) \partial_{xy} \tilde v_y \right] \\
    &= -U_0 e^{\gamma t +ikx} \sum_{j \in \mathbb Z} 2k^2 m \cos(my)  c_j e^{ijy} \\
    &= -U_0 e^{\gamma t +ikx} \sum_{j\in \mathbb Z} k^2 m (c_{j-m} + c_{j+m}) e^{ijy}.
\end{align}
Plugging the results into Eq.~\eqref{eq:pert-pressure-laplace}, we obtain the following expression for the $b_j$ coefficients:
\begin{equation} \label{eq:pressure-coefficients}
    b_j = -U_0 \frac{mk^2}{j^2+k^2} (c_{j-m} + c_{j+m}).
\end{equation}
Since the $c_j$ coefficients decay exponentially, it is easy to see that the $b_j$ coefficients also decay exponentially. Thus, the profile for $\tilde p$ is smooth.
\item Since we have freedom in choosing $\Delta$ in Eq.~\eqref{eq:Delta}, we choose it so $\|\tilde{\mathbf{v}}\|_{L^2}=\aleph$ for some $\aleph >0$. It is easy to see, from Eq.~\eqref{eq:pressure-coefficients} that $\|\tilde p\|_{L^2} = q\aleph$ for some $q>0$.
    \end{enumerate}

\end{proof}

\section{Proofs of Lemmas~\ref{lem:error_bound} and~\ref{lem:vort_error_bound}} \label{app:error_bound}
\errorbound*
\begin{proof}
Consider the PDE for the error, which we have rewritten below for the sake of convenience:
\begin{gather}
    \frac{\partial \boldsymbol{\eta}}{\partial t} = -\nabla \left(\tilde{\mathbf {u}} \cdot \boldsymbol{\eta} + \frac{1}{2} \boldsymbol{\eta}\cdot \boldsymbol{\eta} + \eta_p \right) + \boldsymbol{\eta} \times (\nabla \times \tilde{\mathbf{u}} ) + \tilde{\mathbf{u}}\times(\nabla \times \boldsymbol{\eta}) + \boldsymbol{\eta}\times (\nabla \times \boldsymbol{\eta}) - \epsilon^2(\tilde{\mathbf{v}}\cdot \nabla) \tilde{\mathbf{v}} \label{eq:error_appendix}\\
    \nabla \cdot \boldsymbol{\eta} = 0 .\label{eq:error_div_appendix}
\end{gather}
The equation for the norm of the error can be found as follows:
\begin{align}
    \frac{d \|\boldsymbol{\eta}\|_{L^2}^2}{dt} &= \frac{d}{dt} \int_{\mathbb T^2} \boldsymbol{\eta}\cdot \boldsymbol{\eta}  \, d^2 \mathbf r \\
    &= 2 \int_{\mathbb T^2} \boldsymbol{\eta}\cdot \frac{\partial \boldsymbol{\eta}}{\partial t} d^2 \mathbf r \\
    &= 2\int_{\mathbb T^2} \boldsymbol{\eta}\cdot \left [ -\nabla \left(\tilde{\mathbf {u}} \cdot \boldsymbol{\eta} + \frac{1}{2} \boldsymbol{\eta}\cdot \boldsymbol{\eta}+\eta_p \right) + \tilde{\mathbf{u}}\times(\nabla \times \boldsymbol{\eta}) - \epsilon^2(\tilde{\mathbf{v}}\cdot \nabla) \tilde{\mathbf{v}} \right] d^2 \mathbf r,
\end{align}
where the second and fourth term of Eq.~(\ref{eq:error_appendix}) vanish due to orthogonality. Using vector calculus identities, we can rewrite the first term in the integral as
\begin{align}
    \boldsymbol{\eta}\cdot \nabla \left(\tilde{\mathbf {u}} \cdot \boldsymbol{\eta} + \frac{1}{2} \boldsymbol{\eta}\cdot \boldsymbol{\eta}+\eta_p \right) &= \nabla\cdot  \left ( \left [\tilde{\mathbf {u}} \cdot \boldsymbol{\eta} + \frac{1}{2} \boldsymbol{\eta}\cdot \boldsymbol{\eta}+\eta_p  \right] \boldsymbol{\eta} \right ) - \left (\tilde{\mathbf {u}} \cdot \boldsymbol{\eta} + \frac{1}{2} \boldsymbol{\eta}\cdot \boldsymbol{\eta}+\eta_p \right )\nabla \cdot \boldsymbol{\eta} \\
    &= \nabla\cdot  \left ( \left [\tilde{\mathbf {u}} \cdot \boldsymbol{\eta} + \frac{1}{2} \boldsymbol{\eta}\cdot \boldsymbol{\eta}+\eta_p  \right] \boldsymbol{\eta} \right ) , \label{eq:div_error}
\end{align}
where we have used the incompressibility of the error to remove the second term on the right-hand side. Integrating Eq.~(\ref{eq:div_error}) yields zero due to the divergence theorem, since the domain, $\mathbb T^2$, has no boundary. Thus, we obtain
\begin{align}
    \frac{d \|\boldsymbol{\eta}\|_{L^2}^2}{dt} &= 2 \int_{\mathbb T^2} \boldsymbol{\eta} \cdot \left [-\epsilon^2(\tilde{\mathbf{v}}\cdot \nabla) \tilde{\mathbf{v}}  + \tilde{\mathbf{u}}\times(\nabla \times \boldsymbol{\eta})   \right] d^2 \mathbf r.
\end{align}
This implies the following differential inequality:
\begin{align}
     \frac{d \|\boldsymbol{\eta}\|_{L^2}^2}{dt} &\leq 2 \left | \int_{\mathbb T^2} \boldsymbol{\eta} \cdot \left [-\epsilon^2(\tilde{\mathbf{v}}\cdot \nabla) \tilde{\mathbf{v}}  + \tilde{\mathbf{u}}\times(\nabla \times \boldsymbol{\eta})   \right] d^2 \mathbf r  \right | \\
     &\leq 2 \int_{\mathbb T^2} \left \| \boldsymbol{\eta} \cdot \left [-\epsilon^2(\tilde{\mathbf{v}}\cdot \nabla) \tilde{\mathbf{v}}  + \tilde{\mathbf{u}}\times(\nabla \times \boldsymbol{\eta})   \right]  \right \| \, d^2 \mathbf r \\
     &\leq 2\int_{\mathbb T^2} \|\boldsymbol{\eta}\| \ \|\epsilon^2(\tilde{\mathbf{v}}\cdot \nabla )\tilde{\mathbf{v}}\| \, d^2 \mathbf r+ 2\int_{\mathbb T^2} \|\boldsymbol{\eta}\| \ \|\tilde{\mathbf{u}}\times (\nabla \times \boldsymbol{\eta})\| \, d^2 \mathbf r  \\
     &\leq  2\int_{\mathbb T^2} \|\boldsymbol{\eta}\| \ \|\epsilon^2(\tilde{\mathbf{v}}\cdot \nabla )\tilde{\mathbf{v}}\| \, d^2 \mathbf r+ 2\int_{\mathbb T^2} \|\boldsymbol{\eta}\| \ \|\tilde{\mathbf{u}}\| \ \| \nabla \times \boldsymbol{\eta}\| \, d^2 \mathbf r.
     \end{align}
     Using Hölder's inequality, we obtain
     \begin{align}
         \frac{d \|\boldsymbol{\eta}\|_{L^2}^2}{dt} &\leq 2\left ( \int_{\mathbb T^2} \|\boldsymbol{\eta}\|^2 d^2 \mathbf r \right)^{1/2} \left [ \epsilon^2\left( \int_{\mathbb T^2} \|(\tilde{\mathbf{v}}\cdot \nabla )\tilde{\mathbf{v}}\|^2 d^2 \mathbf r \right)^{1/2} + \left (\int_{\mathbb T^2} \|\tilde{\mathbf{u}}\|^2 \ \|\nabla \times \boldsymbol{\eta}\|^2 d^2 \mathbf r \right)^{1/2}  \right ] \\
        &=2\|\boldsymbol{\eta}\|_{L^2} \left [\epsilon^2\|(\tilde{\mathbf{v}}\cdot \nabla)\tilde{\mathbf{v}}\|_{L^2} + \left (\int_{\mathbb T^2} \|\tilde{\mathbf{u}}\|^2 \ \|\nabla \times \boldsymbol{\eta}\|^2 d^2 \mathbf r \right)^{1/2}  \right ] \\
        &\leq 2\|\boldsymbol{\eta}\|_{L^2} \left [\epsilon^2 \|(\tilde{\mathbf{v}}\cdot \nabla)\tilde{\mathbf{v}}\|_{L^2} + \|\tilde{\mathbf{u}} \|_{L^\infty} \|\nabla \times \boldsymbol{\eta}\|_{L^2}  \right ].
     \end{align}
\end{proof}

\vorticity*
\begin{proof}
Recall \cref{eq:error_appendix,eq:error_div_appendix} for the error. We can modify them slightly using the vector calculus identity
\begin{equation}
    \frac{1}{2}\nabla(\mathbf{A}\cdot\mathbf{A}) = (\mathbf{A}\cdot \nabla )\mathbf A + \mathbf{A}\times(\nabla \times \mathbf A)
\end{equation}
to obtain
\begin{gather}
    \frac{\partial \boldsymbol{\eta}}{\partial t} = -\nabla \left(\tilde{\mathbf {u}} \cdot \boldsymbol{\eta} + \frac{1}{2}( \boldsymbol{\eta}\cdot \boldsymbol{\eta} + \epsilon^2\tilde{\mathbf{v}}\cdot\tilde{\mathbf{v}}) +\eta_p \right)+ \boldsymbol{\eta} \times (\nabla \times \tilde{\mathbf{u}} ) +\tilde{\mathbf{u}}\times(\nabla \times \boldsymbol{\eta}) + \boldsymbol{\eta}\times (\nabla \times \boldsymbol{\eta}) + \epsilon^2\tilde{\mathbf{v}}\times (\nabla\times\tilde{\mathbf{v}}) \\
    \nabla \cdot \boldsymbol{\eta} = 0.
\end{gather}
Let $\boldsymbol{\omega}_{\boldsymbol{\eta}}\coloneq  \nabla \times \boldsymbol{\eta}$ be the vorticity of the error. Let us also define $\boldsymbol{\omega}_{\tilde{\mathbf{u}}} \coloneq \nabla \times \tilde{\mathbf{u}}$ and $\boldsymbol{\omega}_{\tilde{\mathbf{v}}} \coloneq \nabla \times \tilde{\mathbf v}$. Note that all these quantities are incompressible by definition. The equation for the vorticity of the error is
\begin{align}
    \frac{\partial \boldsymbol{\omega}_{\boldsymbol{\eta}}}{\partial t} &= \nabla \times \left [ \boldsymbol{\eta} \times (\nabla \times \tilde{\mathbf{u}} ) +\tilde{\mathbf{u}}\times(\nabla \times \boldsymbol{\eta}) + \boldsymbol{\eta}\times (\nabla \times \boldsymbol{\eta}) + \epsilon^2\tilde{\mathbf{v}}\times (\nabla\times\tilde{\mathbf{v}}) \right ] \\
    &= \nabla \times \left [ \boldsymbol{\eta} \times \boldsymbol{\omega}_{\tilde{\mathbf{u}}} +\tilde{\mathbf{u}}\times\boldsymbol{\omega}_{\boldsymbol{\eta}} + \boldsymbol{\eta}\times \boldsymbol{\omega}_{\boldsymbol{\eta}} + \epsilon^2 \tilde{\mathbf{v}}\times \boldsymbol{\omega}_{\tilde{\mathbf{v}}} \right ] \\
    &= (\boldsymbol{\omega}_{\tilde{\mathbf{u}}}\cdot \nabla)\boldsymbol{\eta} -(\boldsymbol{\eta}\cdot \nabla) \boldsymbol{\omega}_{\tilde{\mathbf u}} +(\boldsymbol{\omega}_{\boldsymbol{\eta}}\cdot \nabla)\tilde{\mathbf u} - (\tilde{\mathbf{u}}\cdot \nabla) \boldsymbol{\omega}_{\boldsymbol{\eta}} \label{eq:vort_line_1} \\
    &\quad +(\boldsymbol{\omega}_{\boldsymbol{\eta}}\cdot \nabla)\boldsymbol{\eta} - (\boldsymbol{\eta}\cdot\nabla)\boldsymbol{\omega}_{\boldsymbol{\eta}} + \epsilon^2(\boldsymbol{\omega}_{\tilde{\mathbf{v}}}\cdot \nabla)\tilde{\mathbf{v}} - \epsilon^2(\tilde{\mathbf{v}}\cdot \nabla)\boldsymbol{\omega}_{\tilde{\mathbf{v}}} \label{eq:vort_line_2} \\
    &=-(\boldsymbol{\eta}\cdot \nabla) \boldsymbol{\omega}_{\tilde{\mathbf u}} -(\tilde{\mathbf{u}}\cdot \nabla) \boldsymbol{\omega}_{\boldsymbol{\eta}} - (\boldsymbol{\eta}\cdot\nabla)\boldsymbol{\omega}_{\boldsymbol{\eta}} - \epsilon^2(\tilde{\mathbf{v}}\cdot \nabla)\boldsymbol{\omega}_{\tilde{\mathbf{v}}},
\end{align}
where we have used the vector calculus identity
\begin{equation}
    \nabla \times (\mathbf A \times \mathbf B) = \mathbf A(\nabla \cdot \mathbf B) - \mathbf B(\nabla \cdot \mathbf A) + (\mathbf B \cdot \nabla ) \mathbf A - (\mathbf A \cdot \nabla ) \mathbf B,
\end{equation}
as well as the incompressibility conditions. We have also used the fact that the vorticities are in the $\mathbf{\hat{k}}$-direction, so the first and third terms of \cref{eq:vort_line_1,eq:vort_line_2} vanish. We take advantage of the orientation of the vorticities to define
\begin{align}
    \omega_{\boldsymbol{\eta}} &\coloneq (\nabla \times \boldsymbol{\eta}) \cdot \mathbf{\hat{k}} \\
    \omega_{{\tilde{\mathbf{u}}}} &\coloneq (\nabla \times \tilde{\mathbf{u}}) \cdot \mathbf{\hat{k}} \\
    \omega_{{\tilde{\mathbf{v}}}} &\coloneq (\nabla \times \tilde{\mathbf{v}}) \cdot \mathbf{\hat{k}},
\end{align}
which gives us a scalar PDE:
\begin{align}
    \frac{\partial \omega_{\boldsymbol{\eta}}}{\partial t} = -(\boldsymbol{\eta}\cdot \nabla) \omega_{\tilde{\mathbf u}} - (\tilde{\mathbf{u}}\cdot \nabla) \omega_{\boldsymbol{\eta}} - (\boldsymbol{\eta}\cdot\nabla)\omega_{\boldsymbol{\eta}} - \epsilon^2(\tilde{\mathbf{v}}\cdot \nabla)\omega_{\tilde{\mathbf{v}}}.
\end{align}
Now let us find a bound for $\|\omega_{\boldsymbol{\eta}}\|_{L^2}$:
\begin{align}
    \frac{d \|\omega_{\boldsymbol{\eta}}\|_{L^2}^2}{dt} &= \frac{d}{dt}\int_{\mathbb T^2} \omega_{\boldsymbol{\eta}}^2 d^2 \mathbf r \\
    &= 2\int_{\mathbb T^2} \omega_{\boldsymbol{\eta}} \frac{\partial \omega_{\boldsymbol{\eta}}}{\partial t} d^2 \mathbf r \\
    &= -2\int_{\mathbb T^2} \omega_{\boldsymbol{\eta}} \left [ (\boldsymbol{\eta}\cdot \nabla) \omega_{\tilde{\mathbf u}} + (\tilde{\mathbf{u}}\cdot \nabla) \omega_{\boldsymbol{\eta}} + (\boldsymbol{\eta}\cdot\nabla)\omega_{\boldsymbol{\eta}} + \epsilon^2(\tilde{\mathbf{v}}\cdot \nabla)\omega_{\tilde{\mathbf{v}}}\right] d^2 \mathbf r  \\
    &= -2 \int_{\mathbb T^2} \omega_{\boldsymbol{\eta}} (\boldsymbol{\eta}\cdot \nabla) \omega_{\tilde{\mathbf u}} + \omega_{\boldsymbol{\eta}} (\mathbf{u}\cdot \nabla) \omega_{\boldsymbol{\eta}} + \epsilon^2 \omega_{\boldsymbol{\eta}}(\tilde{\mathbf{v}}\cdot \nabla)\omega_{\tilde{\mathbf{v}}} \, d^2 \mathbf r .
\end{align}
We can rewrite the second term on the right-hand side as
\begin{align}
    \omega_{\boldsymbol{\eta}} (\mathbf{u}\cdot \nabla) \omega_{\boldsymbol{\eta}} &= \frac{1}{2}\mathbf{u}\cdot \nabla \omega_{\boldsymbol{\eta}}^2 \\
    &= \frac{1}{2} \nabla \cdot (\omega_{\boldsymbol{\eta}}^2 \mathbf u) - \frac{1}{2} \omega_{\boldsymbol{\eta}}^2 \nabla \cdot \mathbf u \\
    &= \frac{1}{2} \nabla \cdot (\omega_{\boldsymbol{\eta}}^2 \mathbf u).
\end{align}
Integrating this term yields zero due to the divergence theorem. Thus, we find
\begin{equation}
    \frac{d\|\omega_{\boldsymbol{\eta}}\|_{L^2}^2}{dt} = -2\int_{\mathbb T^2} \omega_{\boldsymbol{\eta}} \boldsymbol{\eta}\cdot \nabla \omega_{\tilde{\mathbf u}} +\epsilon^2 \omega_{\boldsymbol{\eta}}\tilde{\mathbf{v}}\cdot \nabla \omega_{\tilde{\mathbf{v}}}   d^2 \mathbf r.
\end{equation}
This yields the following differential inequality:
\begin{align}
    \frac{d\|\omega_{\boldsymbol{\eta}}\|_{L^2}^2}{dt}&\leq 2 \left | \int_{\mathbb T^2} \omega_{\boldsymbol{\eta}} \boldsymbol{\eta}\cdot \nabla \omega_{\tilde{\mathbf u}} + \epsilon^2 \omega_{\boldsymbol{\eta}}\tilde{\mathbf{v}}\cdot \nabla \omega_{\tilde{\mathbf{v}}}   d^2 \mathbf r  \right | \\
    & \leq 2  \int_{\mathbb T^2} \left \| \omega_{\boldsymbol{\eta}} \boldsymbol{\eta}\cdot \nabla \omega_{\tilde{\mathbf u}} + \epsilon^2 \omega_{\boldsymbol{\eta}}\tilde{\mathbf{v}}\cdot \nabla \omega_{\tilde{\mathbf{v}}} \right \|  d^2 \mathbf r  \\
    &\leq 2 \int_{\mathbb T^2}  \| \omega_{\boldsymbol{\eta}}\| \ \| \boldsymbol{\eta}\| \ \|\nabla \omega_{\tilde{\mathbf u}} \| d^2 \mathbf{r} + \epsilon^2  \| \omega_{\boldsymbol{\eta}}\| \ \| \tilde{\mathbf{v}}\cdot \nabla \omega_{\tilde{\mathbf{v}}} \|  d^2 \mathbf r.
\end{align}
Finally, Hölder's inequality gives
\begin{align}
    \frac{d\|\omega_{\boldsymbol{\eta}}\|_{L^2}^2}{dt}&\leq 2 \|\omega_{\boldsymbol{\eta}}\|_{L^2} \left [ \left (\int_{\mathbb T^2}  \|\boldsymbol{\eta}\|^2 \|\nabla \omega_{\tilde{\mathbf u}}\|^2 d^2 \mathbf r \right )^{1/2} + \epsilon^2 \|\tilde{\mathbf{v}}\cdot \nabla \omega_{\tilde{\mathbf{v}}}\|_{L^2} \right ] \\
    &= 2\|\omega_{\boldsymbol{\eta}}\|_{L^2} \left [ \|\nabla \omega_{\tilde{\mathbf{u}}} \|_{L^\infty}\|\boldsymbol{\eta}\|_{L^2} + \epsilon^2 \|\tilde{\mathbf{v}}\cdot \nabla \omega_{\tilde{\mathbf{v}}}\|_{L^2}   \right].
\end{align}
\end{proof}

\section{Proof of \cref{thm:error_bound_value}} \label{app:error_bound_value}
To prove this theorem, we need some background information on monotone dynamical systems~\cite{smith1995monotone}. Here we describe the key concepts and then use them to solve the differential inequality, thereby proving the theorem.

Let us begin by introducing the following notions and definition. We define the positive and negative orthants of $\mathbb R^n$ as
\begin{align}
    R_+^n &\coloneq \{\mathbf{x}\in\mathbb R^n: x_i \geq 0 \} \\
    R_-^n &\coloneq \{\mathbf{x}\in\mathbb R^n: x_i \leq 0 \}.
\end{align}
\begin{defn}
    Let $\mathbf{u},\mathbf{v}\in\mathbb R^{n}$ be vectors. The \textit{component-wise inequality}, denoted by $\leq_c$, is defined as follows:
    \begin{equation}
        \mathbf{u} \leq_c \mathbf v \iff u_i \leq v_i \ \forall i \in [n].
    \end{equation}
\end{defn}
This relation is sometimes referred to as the \textit{vector order} in $\mathbb R^n$. 
\begin{defn}
    Let $\mathbf{f}(\mathbf{r},t)\colon\mathbb R^n \times \mathbb R_+ \rightarrow \mathbb R^n$ be a vector field.  We say that $\mathbf{f}$ is \textit{quasimonotone nondecreasing} if, for all $\mathbf{u},\mathbf{v}\in\mathbb R^n$ and $t\geq 0$,
    \begin{equation}
        \mathbf{u} \leq_c \mathbf{v}
    \end{equation}
    implies
    \begin{equation}
        \mathbf{f}(t,\mathbf{u}) \leq_c \mathbf{f}(t,\mathbf{v}).
    \end{equation}
\end{defn}
This allows us to prove the following.

\begin{thm} \label{thm:monotone-bound}
    Let $\mathbf f(t,\mathbf{u})\colon\mathbb R^n \times \mathbb R_+ \rightarrow \mathbb R^n$ be a quasimonotone nondecreasing vector field. Suppose
    \begin{equation}
        \frac{d\mathbf{u}}{dt}\leq_c \mathbf f(t,\mathbf{u}), \ \frac{d\mathbf{v}}{dt} = \mathbf f(t,\mathbf{v}), \ \mathbf{u}(0) \leq_c \mathbf{v}(0).
    \end{equation}
    Then
    \begin{equation}
        \mathbf{u}(t) \leq_c \mathbf{v}(t) \ \forall t\geq 0.
    \end{equation}
\end{thm}
\begin{proof}
    Let
    \begin{equation}
        \mathbf{w} \coloneq \mathbf{u} - \mathbf{v}.
    \end{equation}
    Our goal is to show that $\mathbf{w} \leq_c 0$, i.e., $\mathbf{w} \in \mathbb R^n_-$ for $t\geq 0$. We prove this by contradiction. Suppose $\mathbf{w} \notin \mathbb R^n_-$. This means that there exists an exit time
    \begin{equation}
        \tau \coloneq \inf \{t\geq 0: \mathbf{w}(t) \notin \mathbb R^n_- \}.
    \end{equation}
    As $w$ is continuous and $\tau>0$ is well-defined, we have that  $\mathbf{w}(t) \in \mathbb R_-^n$ for $t\leq \tau$, and that at $t=\tau$, there is at least one $i$ satisfying
    \begin{equation}
        w_i(\tau) = 0,
    \end{equation}
    which is equivalent to
    \begin{equation}
        u_i(\tau) = v_i(\tau).
    \end{equation}
    Now, the equation for the $i$th component is
    \begin{align}
        \frac{dw_i}{dt} &= \frac{du_i}{dt} - \frac{dv_i}{dt} \\
        &\leq f_i(t,\mathbf{u}) - f_i(t,\mathbf{v}).
    \end{align}
    Since $\tau$ is the exit time, we must have $\frac{dw_i(\tau)}{dt} > 0$, meaning that
    \begin{equation}
        f_i(\tau,\mathbf{u}) > f_i(\tau,\mathbf{v}).
    \end{equation}
    However, by quasimonotonicity, $\mathbf{u}(\tau) \leq_c \mathbf{v}(\tau)$ implies
    \begin{equation}
        f_i(\tau,\mathbf{u}) \leq f_i(\tau,\mathbf{v}),
    \end{equation}
    which is a contradiction. Hence, $\mathbf{w}\in\mathbb R_-^n$ for all $t\geq 0$, so
    \begin{equation}
        \mathbf u(t) \leq_c \mathbf{v}(t) \ \forall t\geq 0.
    \end{equation}
\end{proof}
\noindent We are now ready to prove \cref{thm:error_bound_value}.
\errorboundvalue*
\begin{proof}
Combining Lemmas~\ref{lem:error_bound} and~\ref{lem:vort_error_bound}, we have two coupled differential inequalities:
\begin{equation} \label{eq:coupled-diff-ineqs}
    \begin{aligned}
            \frac{d\|\boldsymbol{\eta}\|_{L^2}^2}{dt}&\leq 2\|\boldsymbol{\eta}\|_{L^2} \left [\|\tilde{\mathbf{u}} \|_{L^\infty} \|\omega_{\boldsymbol{\eta}}\|_{L^2}  + \epsilon^2 \|(\tilde{\mathbf{v}}\cdot \nabla)\tilde{\mathbf{v}}\|_{L^2}  \right ] \\
        \frac{d\|\omega_{\boldsymbol{\eta}}\|_{L^2}^2}{dt}&\leq 2\|\omega_{\boldsymbol{\eta}}\|_{L^2} \left [ \|\nabla \omega_{\tilde{\mathbf{u}}} \|_{L^\infty}\|\boldsymbol{\eta}\|_{L^2} + \epsilon^2 \|\tilde{\mathbf{v}}\cdot \nabla \omega_{\tilde{\mathbf{v}}}\|_{L^2}   \right] \\
    \end{aligned}
\end{equation}
with initial condition $\|\boldsymbol{\eta}\|_{L^2}(0) = \|\omega_{\boldsymbol{\eta}}\|_{L^2}(0) = 0$. We would like a differential inequality that bounds the time derivatives of $\|\boldsymbol{\eta}\|_{L^2}$ and $\|\omega_{\boldsymbol{\eta}}\|_{L^2}$ as opposed to their squares. As an interlude, consider the following ODE:
\begin{equation}
    \frac{du}{dt} = a\sqrt{u}, \ u(0) = 0.
\end{equation}
Note that the right-hand side is not Lipschitz at the origin. This means that there is not a unique solution to the ODE. In fact, $u(t) = 0$ and $u(t) = a^2t^2/4$ are both valid solutions, with the latter also being a solution\footnote{
If we allow $u$ to be $C^0$, we can have infinitely many solutions, which take the form:
\begin{equation*}
    u(t) = \begin{cases}
    0, \ t\leq \tau\\
    \frac{a^2t^2}{4}, \ t> \tau
    \end{cases}
\end{equation*}
for any $\tau \geq 0$.
}
to $d\sqrt{u}/dt = a/2$.
Returning to Eq.~\eqref{eq:coupled-diff-ineqs}, we see that the right-hand side is also not Lipschitz at the origin, so a similar non-uniqueness problem arises. Clearly, one possible solution is
\begin{equation}
    \|\boldsymbol{\eta}\|_{L^2}(t) = \|\omega_{\boldsymbol{\eta}}\|_{L^2}(t) = 0.
\end{equation}
This would imply that Eqs.~\eqref{eq:incomp-euler} and~\eqref{eq:lin-euler-total-vel} have identical solutions given the same domain and initial condition, which is clearly false. Hence, this is not a physical solution. Similarly, one can rule out solutions where either $\|\boldsymbol{\eta}\|_{L^2}(t) = 0$ or $\|\omega_{\boldsymbol{\eta}}\|_{L^2}(t) = 0$. We then resort to the other branch of solutions, given by the following coupled differential inequalities:
\begin{align}
    \frac{d\|\boldsymbol{\eta}\|_{L^2}}{dt} &\leq \|\tilde{\mathbf{u}}\|_{L^\infty}  \|\omega_{\boldsymbol{\eta}}\|_{L^2} + \epsilon^2\|(\tilde{\mathbf{v}}\cdot \nabla)\tilde{\mathbf{v}}\|_{L^2}  \\
    \frac{d\|\omega_{\boldsymbol{\eta}}\|_{L^2}}{dt} &\leq  \|\nabla \omega_{\tilde{\mathbf{u}}}|\|_{L^\infty}\|\boldsymbol{\eta}\|_{L^2} + \epsilon^2\|\tilde{\mathbf{v}}\cdot \nabla \omega_{\tilde{\mathbf{v}}}\|_{L^2}.
\end{align}
These inequalities can be obtained by setting $\|\boldsymbol{\eta}\|_{L^2}(0) = \|\omega_{\boldsymbol{\eta}}\|_{L^2}(0) = \delta$ for some $\delta>0$, simplifying Eq.~\eqref{eq:coupled-diff-ineqs}, and then taking the limit of $\delta \rightarrow 0$. We can write these inequalities in the form
\begin{equation}
    \frac{d\mathbf{y}}{dt} \leq_c A(t)\mathbf{y} + \mathbf{b}(t),
\end{equation}
where
\begin{align}
    \mathbf y &= 
    \begin{bmatrix}
        \|\boldsymbol{\eta}\|_{L^2} \\
        \|\omega_{\boldsymbol{\eta}}\|_{L^2}
    \end{bmatrix}, ~ \mathbf{y}(0) = 0 \\
    A(t) &= 
    \begin{bmatrix}
        0 & \|\tilde{\mathbf{u}}\|_{L^\infty}  \\
        \|\nabla \omega_{\tilde{\mathbf{u}}} \|_{L^\infty} & 0
    \end{bmatrix} \\
    \mathbf b(t) &= \epsilon^2
    \begin{bmatrix}
        \|(\tilde{\mathbf{v}}\cdot \nabla)\tilde{\mathbf{v}}\|_{L^2} \\
        \|\tilde{\mathbf{v}}\cdot \nabla \omega_{\tilde{\mathbf{v}}}\|_{L^2}
    \end{bmatrix}.
\end{align}
With some algebra, we find that
\begin{align}
    \|\tilde {\mathbf u} \|_{L^\infty} &\leq \sqrt{\alpha_0 + \epsilon\alpha_1 e^{\gamma_u t} + \epsilon^2\alpha_2 e^{2\gamma_u t}} \\
    \|\nabla \omega_{\tilde{\mathbf u}} \|_{L^\infty} &\leq \sqrt{\beta_0 + \epsilon \beta_1 e^{\gamma_u t} + \epsilon^2\beta_2 e^{2\gamma_u t}}  \\
    \|(\tilde{\mathbf{v}}\cdot \nabla)\tilde{\mathbf{v}}\|_{L^2} &\leq b_1 e^{2\gamma_u t} \\
    \|\tilde{\mathbf{v}}\cdot \nabla \omega_{\tilde{\mathbf{v}}}\|_{L^2} &\leq b_2 e^{2\gamma_u t},
\end{align}
with all the coefficients being positive.

For the rest of the analysis, we consider evolution times up to $t \leq \frac{1}{\gamma_u} \log(1/\epsilon)$. This allows us to impose a constant bound on the matrix elements in the differential inequality. The bounds are as follows:
\begin{align}
    \|\tilde {\mathbf u} \|_{L^\infty} &\leq \sqrt{\alpha_0 + \alpha_1 + \alpha_2} \\
    \|\nabla \omega_{\tilde{\mathbf u}} \|_{L^\infty}&\leq \sqrt{\beta_0 + \beta_1 + \beta_2}.
\end{align}
Let $\alpha,\beta$ be positive constants such that
\begin{align}
    \alpha^2 &> \max \{\sqrt{\alpha_0 + \alpha_1 + \alpha_2}, 2\gamma_u \} \\
    \beta^2 &> \max \{\sqrt{\beta_0 + \beta_1 + \beta_2}, 2\gamma_u\},
\end{align}
so $\alpha\beta > 2\gamma_u$. The differential inequality can then be written as
\begin{equation}
    \frac{d}{dt}
        \begin{bmatrix}
        \|\boldsymbol{\eta}\|_{L^2} \\
        \|\omega_{\boldsymbol{\eta}}\|_{L^2}
    \end{bmatrix}  \leq_c \begin{bmatrix}
        0 & \alpha^2 \\
        \beta^2 & 0
    \end{bmatrix}
    \begin{bmatrix}
        \|\boldsymbol{\eta}\|_{L^2} \\
        \|\omega_{\boldsymbol{\eta}}\|_{L^2}
    \end{bmatrix} 
    +\epsilon^2e^{2\gamma_u t} \begin{bmatrix}
        b_1 \\
        b_2
    \end{bmatrix}.
\end{equation}
It is easy to verify that the right-hand side is a quasimonotone nondecreasing vector field. Using \cref{thm:monotone-bound}, we find the solution to be
\begin{equation}
    \mathbf{y} \leq_c \int_0^t e^{A(t-s)}\mathbf{b}(s) \, ds.
\end{equation}
This yields the following bound for $\|\boldsymbol{\eta}\|_{L^2}$:
\begin{align}
    \|\boldsymbol{\eta}\|_{L^2} &\leq \frac{\epsilon^2}{2} \left [\int_{0}^{t}e^{2\gamma_u s}\left (e^{\alpha\beta(t-s)}b_1 + e^{-\alpha\beta(t-s)}b_1 + e^{\alpha\beta(t-s)}\frac{\alpha b_2}{\beta} - e^{-\alpha\beta(t-s)}\frac{\alpha b_2}{\beta}\right )  \, ds \right ] \\
    &\leq \frac{\epsilon^2}{2}\left [ e^{\alpha\beta t} \int_{0}^{t} e^{2\gamma_u s-\alpha\beta s} \left ( b_1 + \frac{\alpha b_2}{\beta} \right ) \, ds + e^{-\alpha \beta t} \int_{0}^{t} e^{2\gamma_u s + \alpha\beta s} b_1 \, ds  \right ] . \label{eq:bound-intermediate}
\end{align}
Note that for $t\geq 0$,
\begin{equation}
    e^{-\alpha \beta t} \int_0^t e^{2\gamma_u s + \alpha \beta s} \, ds \leq e^{\alpha \beta t} \int_0^t e^{2\gamma_u s - \alpha \beta s} \, ds.
\end{equation}
This can be shown as follows:
\begin{align}
    \frac{e^{\alpha\beta t}\int_0^t e^{2\gamma_u s - \alpha \beta s} \, ds}{ e^{-\alpha \beta t} \int_0^t e^{2\gamma_u s + \alpha \beta s} \, ds} &=\frac{\frac{1}{\alpha \beta - 2\gamma_u} \left [e^{\alpha \beta t } - e^{2\gamma_u t}\right ]}{\frac{1}{\alpha \beta + 2\gamma_u} \left [e^{2\gamma_u t} - e^{-\alpha \beta t}\right]}  \\
    &=\frac{\frac{1}{\alpha \beta - 2\gamma_u} \left [e^{(\alpha \beta-2\gamma_u) t } - 1\right ]}{\frac{1}{\alpha \beta + 2\gamma_u} \left [1 - e^{-(\alpha \beta+2\gamma_u) t}\right]} \eqcolon \frac{n(t)}{d(t)}.
\end{align}
Since $\alpha \beta > 2\gamma_u$, both the numerator and denominator are monotonically increasing. L'Hôpital's rule yields $n(0)/d(0) = 1$. We can also see that
\begin{equation}
    \frac{n'(t)}{d'(t)} = e^{2\alpha \beta t} \geq 1.
\end{equation}
So, $n(t)/d(t) \geq 1$. Using this to bound Eq.~\eqref{eq:bound-intermediate}, we get
\begin{align}
    \|\boldsymbol{\eta}\|_{L^2} &\leq \epsilon^2\left (b_1 + \frac{\alpha b_2}{2\beta} \right ) e^{\alpha \beta t}\int_0^t e^{(2\gamma_u - \alpha \beta) s } \, ds \\
    &=  \epsilon^2 \kappa  e^{\alpha \beta t} \left [ 1- e^{(2\gamma_u - \alpha \beta)t} \right] \\
    &\leq \epsilon^2 \kappa  \left ( e^{\alpha \beta t} - 1 \right ),
\end{align}
where $\kappa  \coloneq \frac{1}{\alpha \beta - 2\gamma_u } \left (b_1 + \frac{\alpha b_2}{2\beta} \right )$. As $\alpha\beta > 2\gamma_u$, we have that $\kappa >0$.
\end{proof}

\end{document}